\documentclass[a4paper,USenglish,cleveref,nameinlink, autoref, numberwithinsect]{lipics-v2021}
\usepackage{graphicx} 
\usepackage{amsmath}
\usepackage{tikz}
\usepackage{thm-restate}
\usetikzlibrary{decorations.pathreplacing,decorations.pathmorphing,arrows,math,shapes,fit,calc}
\usepackage{amsthm}
\theoremstyle{definition}
\newtheorem{problem}{Problem}
\crefname{observation}{Observation}{Observations}

\usetikzlibrary{positioning}

\usepackage{todonotes}

\title{Simultaneous Graph Parameters and How to Bound Them}

\author{Robert Scheffler}{Institute of Mathematics, Brandenburg University of Technology, Cottbus, Germany}{robert.scheffler@b-tu.de}{https://orcid.org/0000-0001-6007-4202}{}
\author{Philipp Wolf Schleicher}{Institute of Mathematics, Brandenburg University of Technology, Cottbus, Germany}{philippwolf.schleicher@b-tu.de}{https://orcid.org/0009-0001-6652-2042}{}

\authorrunning{R. Scheffler and P. W. Schleicher}

\tikzstyle{vertex}=[draw, circle, fill=black, inner sep=1.5pt]

\definecolor{BTUred}{RGB}{228,35,19}
\definecolor{coolblue}{HTML}{5FABF5}
\definecolor{BTUmagenta}{RGB}{215,0,127}

\renewcommand{\P}{\ensuremath{\mathsf{P}}}
\newcommand{\NP}{\ensuremath{\mathsf{NP}}}
\newcommand{\XP}{\ensuremath{\mathsf{XP}}{}}
\newcommand{\FPT}{\ensuremath{\mathsf{FPT}}}
\newcommand{\W}{\ensuremath{\mathsf{W[1]}}}
\newcommand{\cP}{\ensuremath{\mathcal{P}}}
\renewcommand{\O}{\ensuremath{\mathcal{O}}}
\newcommand{\C}{\ensuremath{\mathcal{C}}}
\newcommand{\N}{\ensuremath{\mathbb{N}}}

\newcommand{\bbR}{\ensuremath{\mathbb{R}}}
\newcommand{\sfT}{\ensuremath{\mathsf{T}}}

\newcommand{\bbG}{\mathbb{G}}
\newcommand{\nd}[1]{\mathsf{nd}\!\br{#1}}
\newcommand{\td}[1]{\mathsf{td}\!\br{#1}}
\newcommand{\tdbr}[1]{\mathsf{td}\br{#1}}
\newcommand{\bw}[1]{\mathsf{bw}\!\br{#1}}
\newcommand{\tbw}[1]{\mathsf{tbw}\!\br{#1}}
\newcommand{\thin}{\mathsf{thin}}
\newcommand{\precthin}{\mathsf{prec}\text{-}\mathsf{thin}}
\newcommand{\propthin}{\mathsf{pthin}}
\newcommand{\precpropthin}{\mathsf{prec}\text{-}\mathsf{pthin}}
\newcommand{\ecc}[1]{\mathsf{ecc}\!\br{#1}}
\newcommand{\cw}[1]{\mathsf{cw}\!\br{#1}}
\newcommand{\cocluster}{\mathsf{co\text{-}cluster}}

\newcommand{\mwr}[1]{\mathsf{mw_r}\!\br{#1}}
\newcommand{\fwr}[1]{\mathsf{fw_r}\!\br{#1}}
\newcommand{\tww}[1]{\mathsf{tww}\!\br{#1}}

\newcommand{\Gmc}{\text{-}\mathsf{mc}}

\newcommand{\br}[1]{\left( #1 \right)} 
\newcommand{\ekbr}[1]{\left[ #1 \right]} 
\newcommand{\set}[1]{\left\{ #1 \right\}} 
\newcommand{\abs}[1]{\left| #1 \right|} 
\newcommand{\setcond}[2]{\left\{ #1 \  : \  #2 \right\}} 
\newcommand{\cC}{\mathcal{C}}
\newcommand{\cD}{\mathcal{D}}
\newcommand{\cK}{\mathcal{K}}
\newcommand{\cB}{\mathcal{B}}
\newcommand{\cA}{\mathcal{A}}
\newcommand{\cT}{\mathcal{T}}
\newcommand{\cG}{\mathcal{G}}
\newcommand{\GC}{2^{\bbG}}
\newcommand{\TM}{\mathcal{TM}}

\NewDocumentCommand{\simul}{ O{} O{} }{%
\ifblank{#1}%
{\mathsf{simul}}%
{\ifblank{#2}%
{\mathsf{simul}_{#1}}%
{\mathsf{simul}_{#1} \! \br{#2}}%
}%
}

\newcommand{\simulbr}[2]{\mathsf{simul}_{#1} \br{#2}}

\NewDocumentCommand{\simulnonempty}{ O{} O{} }{%
\ifblank{#1}%
{\mathsf{simul}^\exists}%
{\ifblank{#2}%
{\mathsf{simul}^\exists_{#1}}%
{\mathsf{simul}^\exists_{#1} \! \br{#2}}%
}%
}

\newcommand{\simulnonemptybr}[2]{\mathsf{simul}^\exists_{#1} \br{#2}}

\newcommand{\distanceto}[1]{\mathsf{distance}\text{-}\mathsf{to}\text{-}#1}

\NewDocumentCommand{\intdim}{ O{} O{} }{%
\ifblank{#1}%
{\mathsf{idim}}%
{\ifblank{#2}%
{\mathsf{idim}_{#1}}%
{\mathsf{idim}_{#1} \! \br{#2}}%
}
}

\NewDocumentCommand{\cov}{ O{} O{} }{%
\ifblank{#1}%
{\mathsf{cov}}%
{\ifblank{#2}%
{\mathsf{cov}_{#1}}%
{\mathsf{cov}_{#1} \! \br{#2}}%
}%
}

\NewDocumentCommand{\covl}{ O{} O{} }{%
\ifblank{#1}%
{\mathsf{cov}^{\ell}}%
{\ifblank{#2}%
{\mathsf{cov}_{#1}^{\ell}}%
{\mathsf{cov}_{#1}^{\ell} \! \br{#2}}%
}%
}

\NewDocumentCommand{\covf}{ O{} O{} }{%
\ifblank{#1}%
{\mathsf{cov}^{f}}%
{\ifblank{#2}%
{\mathsf{cov}_{#1}^{f}}%
{\mathsf{cov}_{#1}^{f} \! \br{#2}}%
}%
}

\newcommand{\clos}[1][]{
 \ifblank{#1}
 {\operatorname{clos}}
 {\operatorname{clos} \! \left( #1 \right)}
}

\crefname{claim}{Claim}{Claims}
\Crefname{claim}{Claim}{Claims}

\ccsdesc[500]{Mathematics of computing~Graph algorithms}
\ccsdesc[500]{Theory of computation~Parameterized complexity and exact algorithms}

\keywords{graph parameter, graph class, simultaneous representation, simultaneous interval number}

\hideLIPIcs

\nolinenumbers

\begin{document}

\maketitle

\begin{abstract}
    Beisegel, Chiarelli, K\"{o}hler, Milani\v{c}, Mur\v{s}i\v{c}, and Scheffler [SWAT 2024] introduced the concept of simultaneous $\cC$-numbers which associate a graph class $\C$ with a graph parameter. Given a graph~$G$, the simultaneous $\cC$-number is the smallest number $d$ for which there is a graph $H \in \cC$ and a function $L : V(G) \to \cP(\{1,\dots,d\})$ such that two vertices $u$ and $v$ are adjacent in $G$ if and only if they are adjacent in $H$ and their sets $L(u)$ and $L(v)$ are not disjoint. We study the relation of these simultaneous $\cC$-numbers to other graph parameters. In particular, we investigate which parameters fulfill the following property: Parameter $p$ is bounded on class $\cC$ if and only if $p$ is bounded on the class of graphs of simultaneous $\cC$-number~$d$ for any fixed $d$. We show that many well-known graph parameters have this property. Examples are cliquewidth, twin-width, mim-width, tree independence number, thinness as well as boxicity. We furthermore present some parameters, including modular-width and tree-length, that do no have this property. 
    
    We also study when a parameter forms an upper bound on a simultaneous $\C$-number. We characterize those graph classes $\cC$ for which the parameters treewidth, pathwidth, bandwidth, and treedepth upper bound the simultaneous $\cC$-number. Furthermore, we present sufficient conditions on a class $\cC$, such that $\cP$-modular cardinality upper bounds the simultaneous $\cC$-number, where $\cP$ is replaced by the complete graphs, the edgeless graphs, cographs, or the class $\cC$ itself. On the contrary, we show that modular width never forms an upper bound on a non-trivial simultaneous $\cC$-number.

    Finally, we present some general algorithmic results for simultaneous $\cC$-numbers. We show that, given the graph $H \in \C$ and the function $L$ with $d$ labels, the clique problem can be solved in \FPT{} time parameterized by $d$ if and only if the clique problem can be solved in polynomial time on $\cC$. Additionally, we give for certain classes $\cC$ algorithms that compute the simultaneous $\cC$-number of a graph in \FPT{} time when parameterized by some $\cP$-modular cardinalities.
\end{abstract}

\section{Introduction}

Graph parameters play an important role in the fields of structural and algorithmic graph theory as well as parameterized complexity.  In essence, such a parameter is a mapping from the set of graphs to the set of real numbers. The idea is that graphs mapped to large values are in some sense more complex than graphs mapped to small values. In many cases, graph parameters are strongly related to certain graph classes. If there is a graph class where many problems can be solved efficiently, then it is a common approach to introduce parameters that describe how far away a graph is from this graph class.

An example is the notion of \emph{treewidth} that somehow describes how tree-like a graph is. Here, one can observe that many algorithmic problems that are easy to solve on trees can also be solved in linear time on graphs of bounded treewidth. Another famous graph class with nice structural and algorithmic properties is the class of interval graphs, i.e., the intersection graphs of intervals on the real line. Many problems that are hard on general graphs can be solved in polynomial time on interval graphs. Examples are the \textsc{Coloring}, \textsc{Independent Set}, and \textsc{Dominating Set}. Graph parameters with nice algorithmic properties, such as treewidth, cliquewidth and twin-width are unbounded on intervals graphs. On the contrary, graph parameters that are bounded on interval graphs do in most cases not allow for \FPT{} algorithms. This fact motivated the search for another graph width parameter that reflects the structural property of interval graphs more precisely. 

One attempt of such a parameter was introduced by Beisegel, Chiarelli, K\"{o}hler, Milani\v{c}, Mur\v{s}i\v{c}, and Scheffler~\cite{beisegel2024simultaneousinterval}. This approach was motivated by the \emph{simultaneous representation problem} which considers a graph class $\C$ that is characterized by intersection representations. Examples beside the interval graphs are permutation graphs, the intersection graphs of line segments between parallel lines, or circle graphs, the intersection graphs of chords of a cycle. In the simultaneous representation problem of $\C$, we are given $d$ graphs of $\C$ which may share some vertices and ask for intersection representations of them such that each vertex has the same representative in each representation it is contained in. The problem of deciding whether a set of graphs has such a simultaneous representation was introduced by Jampani and Lubiw~\cite{jampani2010siminterval,jampani2012simultaneous} and has gained significant attention for different classes~\cite{blaesius2016simultan,bok2018note,rutter2025simultaneous}.

Beisegel et al.~\cite{beisegel2024simultaneousinterval} used this concept to define \emph{$d$-simultaneous $\cC$-representations}. Such a representation of a graph $G$ consists of a graph $H \in \C$ with $V(H) = V(G)$ and a mapping $L : V(G) \to \cP(\{1,\dots,d\})$ such that two vertices $u$ and $v$ are adjacent in $G$ if and only if they are adjacent in $H$ and $L(u) \cap L(v)$ is not empty.\footnote{In the original definition of Beisegel et al.~\cite{beisegel2024simultaneousinterval}, the graph $H$ is replaced by some intersection representation belonging to a graph of $\cC$. However, there is no difference whether the graph or its representation is considered. Furthermore, this allows us to also consider classes that do not have an intersection model.} If the class $\C$ contains all complete graphs, then every graph has a $d$-simultaneous $\cC$-representation for some value~$d$. So one can define the \emph{simultaneous $\cC$-number} $\simul[\cC][G]$ of a graph $G$ as the smallest value $d$ such that $G$ has a $d$-simultaneous $\cC$-representation. 

Beisegel et al.~\cite{beisegel2024simultaneousinterval} focused on the class of interval graphs. On the one hand, they studied the algorithmic properties of the \emph{simultaneous interval number}~$d$. They showed that there are \FPT{} algorithms for \textsc{Clique} when parameterized by~$d$ as well as for \textsc{Independent Set}, \textsc{Dominating Set}, and \textsc{Coloring} when parameterized by $d + k$ where $k$ is the solution size. Furthermore, they showed that \textsc{Coloring} is \NP-hard for $d=2$ and \textsc{Independent Dominating Set} is \W-hard when parameterized by $d$.
On the other hand, Beisegel et al.~\cite{beisegel2024simultaneousinterval} studied how the simultaneous interval number is related to other graph parameters. They showed that bounded pathwidth or edge clique cover number implies bounded simultaneous interval number, while bounded simultaneous interval number implies boundedness of parameters such as linear mim-width, tree independence number, and boxicity.

These results raise the question how other simultaneous $\cC$-numbers are related to known graph parameters. Clearly, a parameter that is unbounded for class $\cC$ is also unbounded for graphs of simultaneous $\cC$-number~$d$. For the converse, we observe the following when looking at the parameters considered by Beisegel et al.~\cite{beisegel2024simultaneousinterval}: If a parameter is bounded for interval graphs, then it also bounded for graphs of simultaneous interval number~$d$. This property is of particular interest as it suggests that the simultaneous interval number generalizes the structure of interval graphs to such an extent that results obtained for interval graphs also hold for graphs of bounded simultaneous interval number. Here, we will study whether this is a property that is characteristic of interval graphs and the simultaneous interval number, or whether this property applies to all graph classes~$\C$ and their respective simultaneous $\C$-number. The latter would imply that the concept of simultaneous $\C$-numbers might be a promising direction to design graph parameters that generalize particular graph classes. This is of particular interest for those classes have the same problem as interval graph: The class itself has nice algorithmic properties, but those parameters that are bounded on the class do not. Examples are permutation graphs and split graphs as well as their superclasses of chordal graphs and cocomparability graphs.

\subparagraph{Our Contribution}
We present a comprehensive study of the relation between known graph parameters and the simultaneous $\C$-numbers. Here, we focus only on those numbers that are both well-defined and monotone, i.e., they are bounded by a constant for each graph and taking induced subgraphs of a graph does not increase the number. This is ensured by forcing the class $\cC$ to be hereditary and to contain all complete graphs. 

Our goal is to identify those graph parameters who behave well with the simultaneous parameters. To this end, we say that a graph parameter $p$ is \emph{$\simul$-bounded} if for each hereditary graph class $\C$ that contains all complete graphs and has bounded $p$, it holds that $p$ is bounded for graphs of simultaneous $\C$-number~$d$. In \cref{sec:simulbounded}, we identify a wide range of graph parameters that are $\simul$-bounded. For some of these parameters, $\simul$-boundedness is a direct consequence of being closed under first-order transductions. This includes cliquewidth, twin-width and flip-width. However, many parameters are not closed under these transductions and, thus, we prove $\simul$-boundedness from scratch. These parameter include the tree independence number, boxicity, and mim-width. We also present parameters that are not $\simul$-bounded. Examples are tree-length, modular-width, and iterated type partition. In particular, these examples show that we cannot use any transitivity property as there are non-$\simul$-bounded parameters that are upper and lower bounded by $\simul$-bounded parameters. An overview of the most important results is given in \cref{fig:diagram}.

In \cref{sec:upper}, we consider the question in which cases parameters form upper bounds on the simultaneous $\C$-number. In particular, we characterize the graph classes $\cC$ for which the parameters treewidth, pathwidth, bandwidth, treebandwidth, and treedepth form upper bounds on the simultaneous $\C$-number. Furthermore, we present sufficient conditions on a class $\cC$, such that $\cP$-modular cardinality upper bounds the simultaneous $\cC$-number, where $\cP$ is replaced by the complete graphs, the edgeless graphs, cographs or the class $\cC$ itself. On the contrary, we show that modular width and, thus, its lower bounds cliquewidth and twin-width can never form upper bounds on a non-trivial simultaneous $\cC$-number.

Besides these structural studies, we also present some algorithmic results in \cref{sec:algo} that fit into our general approach. We show that, given a $d$-simultaneous $\cC$-representation, the clique problem can be solved in \FPT{} time parameterized by $d$ if and only if the clique problem can be solved in polynomial time on $\cC$. Furthermore, we consider the computation of simultaneous $\cC$-numbers. In particular, we study the complexity of that problem when parameterized by some $\cP$-modular cardinalities (plus the solution size) where $\cP$ is another graph class. This enhances the results of Bonomo-Braberman, Brandwein and Sau~\cite{bonomobraberman2025computing}, who presented an \FPT{} algorithm for computing the simultaneous interval number parameterized by $\mathsf{cluster}$-modular cardinality plus solution size.

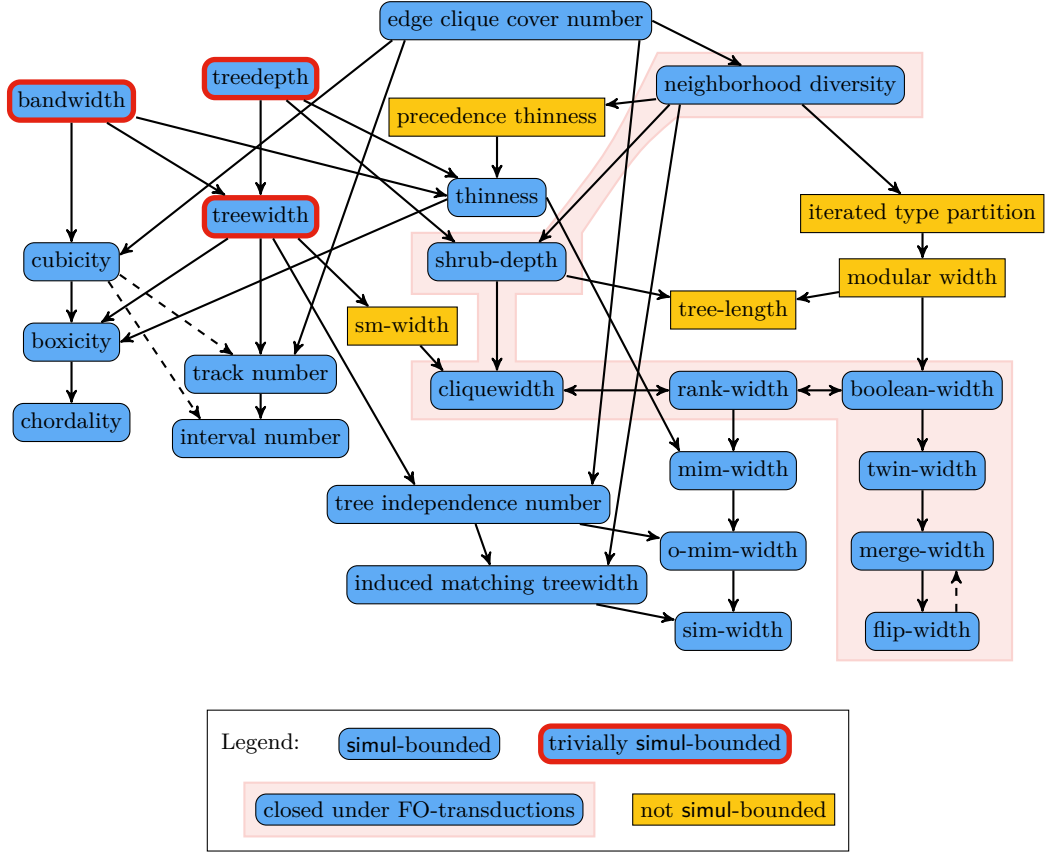
\begin{figure}[t]
\begin{center}
\begin{tikzpicture}[xscale=2.5,yscale=0.85]
    \tikzset{box/.style={draw,rectangle}}
    \tikzset{roundedbox/.style={draw,rectangle,rounded corners}}
    \tikzset{markedbox/.style={draw,rectangle,very thick, fill=lipicsYellow}}
    \tikzset{simul/.style={draw,rectangle,rounded corners, fill=coolblue}}
    \tikzset{trivial/.style={draw=BTUred,rectangle,rounded corners,fill=coolblue,line width=2pt}}
    \tikzset{notsimul/.style={draw,rectangle, fill=lipicsYellow}}
    \footnotesize

     \draw [thick,BTUred!20!white, fill=BTUred!10!white] (3.1,3.75) .. controls (2.85,2.25)  .. (2.5,0.95) -- ++(-0.7,0) -- ++(0,-0.95) -- ++(0.35,0) -- ++(0,-1.05) -- ++(-0.35,0) -- ++(0,-0.9) -- ++(2.25,0) -- ++(0,-3.75) -- ++(0.925,0) -- ++(0,4.65) -- ++(-2.625,0) -- ++(0,1.05) -- ++ (0.35,0) -- ++(0,0.95) .. controls (3,2.25) .. ++(0.5,1.8) -- ++ (1.3,0) -- ++(0,1) -- cycle;
   
    \node[simul] (cub) at (0,0.5) {\mathstrut cubicity};
    \node[simul] (box) at (0,-0.75) {\mathstrut boxicity};
    \node[simul] (chor) at (0,-2) {\mathstrut chordality};
    \node[trivial] (tw) at (1,1.2) {\mathstrut treewidth};
    \node[trivial] (bw) at (0,3) {\mathstrut bandwidth};
    \node[trivial] (td) at (1,3.35) {\mathstrut treedepth};
    \node[simul] (sd) at (2.25,0.5) {\mathstrut shrub-depth};
    \node[notsimul] (sm) at (1.75,-0.5) {\mathstrut sm-width};
    \node[simul] (cw) at (2.25,-1.5) {\mathstrut cliquewidth};
    \node[simul] (rw) at (3.5,-1.5) {\mathstrut rank-width};
    \node[simul] (boolw) at (4.5,-1.5) {\mathstrut boolean-width};
    \node[notsimul] (tl) at (3.5,-0.25) {\mathstrut tree-length};
    \node[notsimul] (mw) at (4.5,0.25) {\mathstrut modular width};
    \node[notsimul] (itp) at (4.5,1.25) {\mathstrut iterated type partition};
    \node[simul] (nd) at (3.75,3.25) {\mathstrut neighborhood diversity};
    \node[simul] (ecc) at (2.35,4.25) {\mathstrut edge clique cover number};
    \node[notsimul] (pre-thin) at (2.25, 2.75) {\mathstrut precedence thinness};
    \node[simul] (thin) at (2.25, 1.5) {\mathstrut thinness};
    \node[simul] (mim) at (3.5,-2.75) {\mathstrut mim-width};
    \node[simul] (omim) at (3.5,-4) {\mathstrut o-mim-width};
    \node[simul] (sim) at (3.5,-5.25) {\mathstrut sim-width};
    \node[simul] (twin) at (4.5,-2.75) {\mathstrut twin-width};
    \node[simul] (merge) at (4.5,-4) {\mathstrut merge-width};
    \node[simul] (flip) at (4.5,-5.25) {\mathstrut flip-width};
    \node[simul] (talpha) at (2.1,-3.275) {\mathstrut tree independence number};
    \node[simul] (imtw) at (2.25,-4.525) {\mathstrut induced matching treewidth};
    \node[simul] (tn) at (1,-1.25) {\mathstrut track number};
    \node[simul] (in) at (1,-2.25) {\mathstrut interval number};

    \draw[thick,-stealth',font=\sffamily] (bw) to (tw);
    \draw[thick,-stealth',font=\sffamily] (td) to (tw);
    \draw[thick,-stealth',font=\sffamily] (td) to (sd.155);
    \draw[thick,-stealth',font=\sffamily] (td) to (thin);
    \draw[thick,-stealth',font=\sffamily] (bw) to (cub);
    \draw[thick,-stealth',font=\sffamily] (bw) to (thin.178);
    \draw[thick,-stealth',font=\sffamily] (cub) to (box);
    \draw[thick,-stealth',font=\sffamily] (box) to (chor);
    \draw[thick,-stealth',font=\sffamily] (ecc.0) to (nd);
    \draw[thick,-stealth',font=\sffamily] (ecc.-9) to (talpha.9);
    \draw[thick,-stealth',font=\sffamily] (ecc.-170) to (tn.30);
    \draw[thick,-stealth',font=\sffamily] (ecc.-171) to (cub.8);
    \draw[thick,-stealth',font=\sffamily] (nd.-170) to (sd.25);
    \draw[thick,-stealth',font=\sffamily] (nd) to (itp);
    \draw[thick,-stealth',font=\sffamily] (nd.-169) to (imtw.10);
    \draw[thick,-stealth',font=\sffamily] (nd) to (pre-thin);
    \draw[thick,-stealth',font=\sffamily] (pre-thin) to (thin);
    \draw[thick,-stealth',font=\sffamily] (thin.-178) to (box.0);
    \draw[thick,-stealth',font=\sffamily] (thin.0) to (mim.160);
    \draw[thick,-stealth',font=\sffamily] (itp) to (mw);
    \draw[thick,-stealth',font=\sffamily] (mw) to (tl);
    \draw[thick,-stealth',font=\sffamily] (mw) to (boolw);
    \draw[thick,-stealth',font=\sffamily] (sd) to (cw);    
    \draw[thick,-stealth',font=\sffamily] (sd) to (tl);
    \draw[thick,-stealth',font=\sffamily] (tw) to (box);
    \draw[thick,-stealth',font=\sffamily] (tw.-30) to (sm.150);
    \draw[thick,-stealth',font=\sffamily] (sm) to (cw.160);
    \draw[thick,stealth'-stealth',font=\sffamily] (cw) to (rw);
    \draw[thick,stealth'-stealth',font=\sffamily] (rw) to (boolw);
    \draw[thick,-stealth',font=\sffamily] (rw) to (mim);
    \draw[thick,-stealth',font=\sffamily] (mim) to (omim);
    \draw[thick,-stealth',font=\sffamily] (omim) to (sim);
    \draw[thick,-stealth',font=\sffamily] (boolw) to (twin);
    \draw[thick,-stealth',font=\sffamily] (twin) to (merge);
    \draw[thick,-stealth',font=\sffamily] (merge) to (flip);
    \draw[thick,-stealth',font=\sffamily] (tw) to (talpha.160);
    \draw[thick,-stealth',font=\sffamily] (talpha) to (omim);
    \draw[thick,-stealth',font=\sffamily] (talpha) to (imtw);
    \draw[thick,-stealth',font=\sffamily] (imtw) to (sim);
    \draw[thick,-stealth',font=\sffamily] (tw) to (tn);
    \draw[thick,-stealth',font=\sffamily] (tn) to (in);

    \draw[thick,-stealth',font=\sffamily,dashed] (flip.30) to (merge.-30);
    \draw[thick,-stealth',font=\sffamily,dashed] (cub.-15) to (tn);
    \draw[thick,-stealth',font=\sffamily,dashed] (cub.-28) to (in.162);
    
    \begin{scope}[xshift=1cm,yshift=-7cm]
    \node[inner sep=0pt] (Legend) at (0,0) {Legend:};
    \node[simul] (Simul) [right=0.5cm of Legend] {$\simul$-bounded};
    \node[trivial] (Trivial) [right=0.5cm of Simul] {trivially $\simul$-bounded};
    \node[simul, outer sep=0.15cm] (FO2) [below=0.3cm of Simul] {\phantom{closed under FO-transductions}};
    \draw[thick,BTUred!20!white, fill=BTUred!10!white] (FO2.south west) rectangle (FO2.north east);
    \node[inner sep=0, outer sep=0] (test) at (FO2.south) {};
     \node[simul, outer sep=0.15cm] (FO) [below=0.3cm of Simul] {closed under FO-transductions};
    \node[notsimul] (Not) [right=0.5cm of FO] {not $\simul$-bounded};
    \node (border) [draw, rectangle, fit={(Legend) (Simul) (Trivial) (Not) (test)}, inner sep=5pt] {};
    \end{scope}
\end{tikzpicture}
\end{center}
  \caption{Diagram illustrating the relations between considered parameters.
  A directed edge from $p$ to $q$ means that $p$ upper bounds $q$.
  If there is no directed path from $p$ to $q$, then $p$ does not upper bound $q$.
  A dashed edge from $p$ to $q$ means that, to the best of our knowledge, it is open whether $p$ upper bounds $q$.
  We give a justification of this figure in Appendix \ref{sec:justification}.}
  \label{fig:diagram}
\end{figure}

\subparagraph{Definitions and Notation}

All graphs considered are finite, simple, and undirected.
A vertex subset $A$ of a graph $G$ is called a \emph{clique} if all vertices of $A$ are pairwise adjacent and an \emph{independent set} if all vertices of $A$ are pairwise not adjacent.
The maximum size of an independent set or a clique of a graph $G$ is called \emph{independence number} or \emph{clique number} of $G$ and denoted by $\alpha(G)$ or $\omega(G)$, respectively.
We denote by $K_n$ the complete graph on $n$ vertices, by $C_n$ the cycle on $n$ vertices, and by $P_n$ the path on $n$ vertices.

A graph class is \emph{non-trivial} if it not empty and does not contain all graphs. We call a graph class \emph{completable} if it contains all complete graphs. A graph class $\cC$ is \emph{hereditary} if for every $G \in \cC$ and every induced subgraph $H$ of $G$ it holds that~$H \in \cC$.

We define $\bbG$ to be the class of all graphs. Then a \emph{graph parameter}~$p$ is a function $p: \bbG \to \bbR$.
We say that a parameter $p$ is \emph{upper bounded} by some parameter $p'$ if there is a function $g : \bbR \to \bbR$ such that for all graphs $G \in \bbG$ it holds that $p(G) \leq g(p'(G))$. We also say that $p'$ \emph{upper bounds} $p$ or is a \emph{upper bound} of $p$ and that $p'$ is \emph{lower bounded} by $p$.
If a parameter $p'$ upper and lower bounds a parameter $p$, we say that $p$ and $p'$ are \emph{functionally equivalent}.
For a graph class $\cC$ we define $p(\cC) := \sup_{G \in \cC} p(G)$. If $p(\cC) < \infty$ we say that $p$ is \emph{bounded on $\cC$} with upper bound $p(\cC)$, otherwise $p$ is \emph{unbounded on $\cC$}.
We will use many parameters throughout this paper. Therefore, we do not define all of them here, but rather only give the definition of a parameter the first time we use it.

We denote by $\mathsf{complete}$ the class of complete graphs and by $\mathsf{edgeless}$ the class of edgeless graphs. 
The class of clusters consists of disjoint unions of cliques and is denoted by $\mathsf{cluster}$. The class containing complements of clusters is denoted by $\cocluster$. We denote interval graphs, the intersection graphs of intervals on the real line, by $\mathsf{interval}$ and cographs, the  $P_4$-free graphs, by $\mathsf{cograph}$.

\section{Simultaneous Representations} \label{sec:simrep}

We start by defining simultaneous representations. Note that they have been introduced in~\cite{beisegel2024simultaneousinterval} for classes of intersection graphs. Here, we define them for all completable classes.

\begin{definition}
    Let $G$ be a graph, $\cC$ be a completable graph class, $d \in \N$, and $L : V(G) \to \cP\!\br{\set{1, 2, \dots, d}}$ be a labeling of the vertices of $G$.
    Furthermore, let $H \in \cC$ with $V(H) = V(G)$ and $E(G) \subseteq E(H)$.
    We say that $(H,L)$ is a \emph{$d$-simultaneous $\cC$-representation} of $G$ if for all $u,v \in V(G)$ it holds that $uv \in E(G)$ if and only if $uv \in E(H)$ and $L(u) \cap L(v) \neq \emptyset$.
\end{definition}

\cref{fig:example1,fig:example2} show examples of what a simultaneous representation might look like. \cref{fig:example1} shows two simultaneous representations of the cycle $C_6$, one using a cograph and one using an interval graph. \cref{fig:example2} shows two simultaneous representations of the graph $G$ obtained by removing a single edge from the complete bipartite graph $K_{3,3}$, again using a cograph and an interval graph. All pictured representations are minimal with respect to the number of labels used for these graph classes.

The class $\cC$ has to be completable to ensure that every graph $G$ has some $d$-simultaneous $\cC$-representation. In fact, we can define $H$ as the complete graph on the vertex set $V(G)$ and consider $m = \abs{E(G)}$ many labels, one for every edge $e_i \in E(G)$. Then $L$ is defined by placing label $i$ in the label sets of the two vertices incident to $e_i$. It can easily be verified that $(H,L)$ is an $m$-simultaneous $\cC$-representation of $G$ (see also~\cite{beisegel2024simultaneousinterval}). This motivates the following definition.

\begin{definition}
    Let $\cC$ be a hereditary completable graph class.
    We define the \emph{simultaneous $\cC$-number} of a graph $G$, denoted by $\simul[\cC][G]$, as the smallest integer $d$ such that $G$ admits a $d$-simultaneous $\cC$-representation.
\end{definition}

Requiring $\cC$ to be hereditary makes the parameters monotone, i.e., the simultaneous $\cC$-number does not increase when induced subgraphs are considered. In fact, for hereditary classes, the restriction of a $d$-simultaneous $\cC$-representation $(H,L)$ of $G$ to some induced subgraph $G'$ of $G$ is a $d$-simultaneous $\cC$-representation of $G'$.

By this definition, the graphs with simultaneous $\cC$-number $0$ are exactly the edgeless graphs and the class of graphs with simultaneous $\cC$-number at most $1$ contains $\cC$. This containment may be strict, as can for example be witnessed by graph classes $\cC$ not containing all edgeless graphs.
In some contexts it may be useful to forbid the empty label set in simultaneous representations.
This can drastically change the number of labels needed to represent a given graph, e.g., consider the class $\mathsf{complete}$ of complete graphs and an edgeless graph $G$ on $n$ vertices. Then, by the above comment, $G$ admits a $0$-simultaneous $\mathsf{complete}$-representation.
If, however, every label set must be non-empty, then $n$ pairwise disjoint non-empty label sets are needed, which is only achievable using $n$ different labels.
This motivates the following definition.

\begin{definition}
    Let $\cC$ be a hereditary completable graph class.
    We define the \emph{simultaneous $\cC$-number with non-empty label sets} of a graph $G$, denoted by $\simulnonempty[\cC][G]$, as the smallest integer $d$ such that $G$ admits a $d$-simultaneous $\cC$-representation with non-empty label sets.\footnote{The notation $\simulnonempty[\cC]$ is chosen to emphasize the fact that there \emph{exists} some element in any label set.}
\end{definition}

Note that, given a simultaneous $\cC$-representation of some graph $G$, we can construct a simultaneous $\cC$-representation of $G$ with non-empty label sets by adding a new label for every vertex $v \in V(G)$ and by placing this new label only in the label set of the corresponding vertex.
Thus, we have for any hereditary completable graph class $\cC$ and any graph $G$ that $\simulnonempty[\cC][G] \leq \abs{V(G)} + \simul[\cC][G] \leq \abs{V(G)} + \abs{E(G)} \leq \abs{V(G)} + \binom{\abs{V(G)}}{2} \leq \abs{V(G)}^2$.  

Obviously, vertices with empty label sets must be isolated. This implies the following.

\begin{observation} \label{obsv:isolatedvertices2}
    Let $\cC$ be a hereditary completable graph class and let $G$ be a graph without isolated vertices. Then $\simulnonempty[\cC][G] = \simul[\cC][G]$.
\end{observation}

Furthermore, we do not need to use empty label sets if isolated vertices can already be added to graphs of $\cC$.

\begin{observation} \label{obsv:isolatedvertices}
    Let $\cC$ be a hereditary completable graph class that is closed under the addition of isolated vertices. Then, for every graph $G$ it holds that $\simulnonempty[\cC][G] = \max\{\simul[\cC][G]\!,1\}$.
\end{observation}

Thus, on most of the relevant graph classes $\cC$, the parameters $\simul[\cC]$ and $\simulnonempty[\cC]$ are functionally equivalent. 
Moreover, we have that a hereditary completable graph class $\cC$ is exactly equal to the class of graphs with $\simulnonempty[\cC]$ equal to $1$.

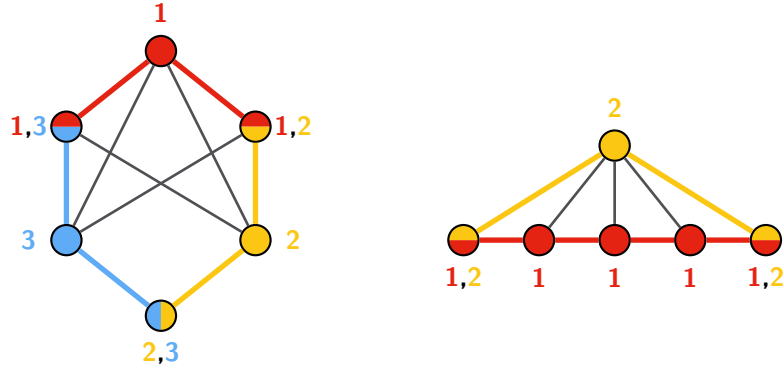
\begin{figure}
    \centering
    \begin{tikzpicture}
        \def\a{.4} 
        \def\labelone{BTUred}
        \def\labeltwo{lipicsYellow}
        \def\labelthree{coolblue}
        \def\nolabel{lipicsGray}
        \tikzset{c/.style={circle,draw,thick,minimum size=\a cm}} 
        \tikzset{one color/.style={c,fill=#1}}  
        \tikzset{pics/two colors/.style args=
            {#1|#2|rotate=#3}{code={%
        \fill[#1,rotate=#3] (0,\a/2) arc(90:270:\a/2)--cycle;           
        \fill[#2,rotate=#3] (0,\a/2) arc(90:-90:\a/2)--cycle;
        \path (0,0) node[c] (-boundary) {};
        }}}
        \tikzset{pics/three colors/.style args=
            {#1|#2|#3|rotate=#4}{code={%
        \fill[#1,rotate=#4] (0,\a/2) arc(90:210:\a/2)--(0,0)--cycle;            
        \fill[#2,rotate=#4] (0,\a/2) arc(90:-30:\a/2)--(0,0)--cycle;
        \fill[#3,rotate=#4] (210:\a/2) arc(210:330:\a/2)--(0,0)--cycle;
        \path (0,0) node[c] (-boundary) {};
        }}}
        \begin{scope}
            \path
            (-1.25,1) node[one color = \labelthree] (B) {}
            (1.25,1) node[one color = \labeltwo] (Y) {}
            (0,3.5) node[one color = \labelone] (R) {}
            ;
            \path 
            (0,0)  pic (BY) {two colors={\labelthree|\labeltwo|rotate=0}}
            (-1.25,2.5) pic (BR) {two colors={\labelthree|\labelone|rotate=90}}
            (1.25,2.5) pic (RY) {two colors={\labelone|\labeltwo|rotate=-90}}
            ;
            \draw[line width=2pt,\labelthree] (BY-boundary) -- (B) -- (BR-boundary);
            \draw[line width=2pt,\labeltwo] (BY-boundary) -- (Y) -- (RY-boundary);
            \draw[line width=2pt,\labelone] (BR-boundary) -- (R) -- (RY-boundary);
            \draw[line width=1pt,\nolabel] (R) -- (B) -- (RY-boundary);
            \draw[line width=1pt,\nolabel] (R) -- (Y) -- (BR-boundary);
            \path
            (0,4) node[font=\sffamily\bfseries] {\textcolor{\labelone}{1}}
            (1.75,2.5) node[font=\sffamily\bfseries] {\textcolor{\labelone}{1},\textcolor{\labeltwo}{2}}
            (1.75,1) node[font=\sffamily\bfseries] {\textcolor{\labeltwo}{2}}
            (0,-0.5) node[font=\sffamily\bfseries] {\textcolor{\labeltwo}{2},\textcolor{\labelthree}{3}}
            (-1.75,1) node[font=\sffamily\bfseries] {\textcolor{\labelthree}{3}}
            (-1.75,2.5) node[font=\sffamily\bfseries] {\textcolor{\labelone}{1},\textcolor{\labelthree}{3}}
            ;
        \end{scope}
        \begin{scope}[xshift=4cm,yshift=1cm]
            \path
            (1,0) node[one color = \labelone] (R1) {}
            (2,0) node[one color = \labelone] (R2) {}
            (3,0) node[one color = \labelone] (R3) {}
            (2,1.25) node[one color = \labeltwo] (Y) {}
            ;
            \path
            (0,0) pic (YR1) {two colors={\labelone|\labeltwo|rotate=90}}
            (4,0) pic (YR2) {two colors={\labelone|\labeltwo|rotate=90}}
            ;
            \draw[line width=2pt,\labeltwo] (YR1-boundary) -- (Y) -- (YR2-boundary);
            \draw[line width=2pt,\labelone] (YR1-boundary) -- (R1) -- (R2) -- (R3) -- (YR2-boundary);
            \draw[line width=1pt,\nolabel] (R1) -- (Y) -- (R2);
            \draw[line width=1pt,\nolabel] (Y) -- (R3);
            \path
            (0,-0.5) node[font=\sffamily\bfseries] {\textcolor{\labelone}{1},\textcolor{\labeltwo}{2}}
            (1,-0.5) node[font=\sffamily\bfseries] {\textcolor{\labelone}{1}}
            (2,-0.5) node[font=\sffamily\bfseries] {\textcolor{\labelone}{1}}
            (3,-0.5) node[font=\sffamily\bfseries] {\textcolor{\labelone}{1}}
            (4,-0.5) node[font=\sffamily\bfseries] {\textcolor{\labelone}{1},\textcolor{\labeltwo}{2}}
            (2,1.75) node[font=\sffamily\bfseries] {\textcolor{\labeltwo}{2}}
            ;
        \end{scope}
\end{tikzpicture}
    \caption{A $3$-simultaneous $\mathsf{cograph}$-representation and a $2$-simultaneous $\mathsf{interval}$-representation of the cycle $C_6$.}
    \label{fig:example1}
\end{figure}

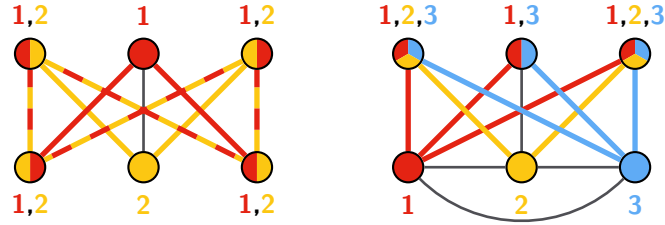
\begin{figure}
    \centering
        \begin{tikzpicture}
        \def\a{.4} 
        \def\d{8pt}
        \def\labelone{BTUred}
        \def\labeltwo{lipicsYellow}
        \def\labelthree{coolblue}
        \def\nolabel{lipicsGray}
        \tikzset{c/.style={circle,draw,thick,minimum size=\a cm}} 
        \tikzset{one color/.style={c,fill=#1}}  
        \tikzset{pics/two colors/.style args=
            {#1|#2|rotate=#3}{code={%
        \fill[#1,rotate=#3] (0,\a/2) arc(90:270:\a/2)--cycle;           
        \fill[#2,rotate=#3] (0,\a/2) arc(90:-90:\a/2)--cycle;
        \path (0,0) node[c] (-boundary) {};
        }}}
        \tikzset{pics/three colors/.style args=
            {#1|#2|#3|rotate=#4}{code={%
        \fill[#1,rotate=#4] (0,\a/2) arc(90:210:\a/2)--(0,0)--cycle;            
        \fill[#2,rotate=#4] (0,\a/2) arc(90:-30:\a/2)--(0,0)--cycle;
        \fill[#3,rotate=#4] (210:\a/2) arc(210:330:\a/2)--(0,0)--cycle;
        \path (0,0) node[c] (-boundary) {};
        }}}
        \begin{scope}
            \path
            (0,0) pic (YR1) {two colors={\labelone|\labeltwo|rotate=180}}
            (0,1.5) pic (YR2) {two colors={\labelone|\labeltwo|rotate=0}}
            (3,0) pic (YR3) {two colors={\labelone|\labeltwo|rotate=0}}
            (3,1.5) pic (YR4) {two colors={\labelone|\labeltwo|rotate=180}}
            ;
            \path
            (1.5,0) node[one color = \labeltwo] (R) {}
            (1.5,1.5) node[one color = \labelone] (Y) {}
            ;
            \draw[line width=1pt,\nolabel] (Y) -- (R);
            \draw[line width=2pt,\labeltwo] (YR2-boundary) -- (R) -- (YR4-boundary);
            \draw[line width=2pt,\labelone] (YR1-boundary) -- (Y) -- (YR3-boundary);
            \draw[line width=2pt,\labeltwo,dash pattern= on \d off \d] (YR1-boundary) -- (YR2-boundary) -- (YR3-boundary) -- (YR4-boundary) -- (YR1-boundary);
            \draw[line width=2pt,\labelone,dash pattern= on \d off \d,dash phase=\d] (YR1-boundary) -- (YR2-boundary) -- (YR3-boundary) -- (YR4-boundary) -- (YR1-boundary);
            \path
            (0,2) node[font=\sffamily\bfseries] {\textcolor{\labelone}{1},\textcolor{\labeltwo}{2}}
            (1.5,2) node[font=\sffamily\bfseries] {\textcolor{\labelone}{1}}
            (3,2) node[font=\sffamily\bfseries] {\textcolor{\labelone}{1},\textcolor{\labeltwo}{2}}
            (0,-0.5) node[font=\sffamily\bfseries] {\textcolor{\labelone}{1},\textcolor{\labeltwo}{2}}
            (1.5,-0.5) node[font=\sffamily\bfseries] {\textcolor{\labeltwo}{2}}
            (3,-0.5) node[font=\sffamily\bfseries] {\textcolor{\labelone}{1},\textcolor{\labeltwo}{2}}
            ;
        \end{scope}
        \begin{scope}[xshift=5cm]
            \path
            (0,0) node[one color = \labelone] (R) {}
            (1.5,0) node[one color = \labeltwo] (B) {}
            (3,0) node[one color = \labelthree] (Y) {}
            ;
            \path
            (0,1.5) pic (RYB1) {three colors={\labelone|\labelthree|\labeltwo|rotate=0}}
            (1.5,1.5) pic (RY) {two colors={\labelone|\labelthree|rotate=0}}
            (3,1.5) pic (RYB2) {three colors={\labelone|\labelthree|\labeltwo|rotate=0}}
            ;
            \draw[line width=1pt,\nolabel] (R) to (B) to (Y) to [bend left = 45] (R);
            \draw[line width=1pt,\nolabel] (B) -- (RY-boundary);
            \draw[line width=2pt,\labelone] (R) -- (RYB1-boundary);
            \draw[line width=2pt,\labelone] (R) -- (RY-boundary);
            \draw[line width=2pt,\labelone] (R) -- (RYB2-boundary);
            \draw[line width=2pt,\labeltwo] (B) -- (RYB1-boundary);
            \draw[line width=2pt,\labeltwo] (B) -- (RYB2-boundary);
            \draw[line width=2pt,\labelthree] (Y) -- (RYB1-boundary);
            \draw[line width=2pt,\labelthree] (Y) -- (RY-boundary);
            \draw[line width=2pt,\labelthree] (Y) -- (RYB2-boundary);
            \path
            (0,2) node[font=\sffamily\bfseries] {\textcolor{\labelone}{1},\textcolor{\labeltwo}{2},\textcolor{\labelthree}{3}}
            (1.5,2) node[font=\sffamily\bfseries] {\textcolor{\labelone}{1},\textcolor{\labelthree}{3}}
            (3,2) node[font=\sffamily\bfseries] {\textcolor{\labelone}{1},\textcolor{\labeltwo}{2},\textcolor{\labelthree}{3}}
            (0,-0.5) node[font=\sffamily\bfseries] {\textcolor{\labelone}{1}}
            (1.5,-0.5) node[font=\sffamily\bfseries] {\textcolor{\labeltwo}{2}}
            (3,-0.5) node[font=\sffamily\bfseries] {\textcolor{\labelthree}{3}}
            ;
        \end{scope}
    \end{tikzpicture}
    \caption{A $2$-simultaneous $\mathsf{cograph}$-representation and a $3$-simultaneous $\mathsf{interval}$-representation of the graph $K_{3,3} - K_2$.}
    \label{fig:example2}
\end{figure}

\section{Characterizing Lower Bounds: Simul-Boundedness} \label{sec:simulbounded}

We want to investigate when a given parameter is a lower bound to a simultaneous $\cC$-number.
Clearly, if a parameter $p$ lower bounds $\simul[\cC]$ for some class $\cC$, then $p$ is bounded on $\cC$.
We investigate when the converse is also true, i.e., for which parameters $p$ it holds that $p$ being bounded on $\cC$ implies that $p$ lower bounds $\simul[\cC]$.
In other words, we are looking for parameters where being bounded on a class $\cC$ is equivalent to lower bounding $\simul[\cC]$. We formalize this property as follows.

\begin{definition}
    We call a graph parameter $p$ \emph{$\simul$-bounded} if for every hereditary completable graph class $\cC$, the parameter $p$ is upper bounded by $\simul[\cC]$ if and only if $p$ is bounded on $\cC$.
    The term \emph{$\simulnonempty$-bounded} is defined analogously. 
\end{definition}

Thus, to show $\simul$-boundedness of a graph parameter $p$, we have to show that for any hereditary completable $\cC$ the following holds: if $p$ is bounded on $\cC$, then $p(G) \leq f(\simul[\cC][G])$ for any graph $G$ and some function $f$.
This in particular shows that every parameter which is unbounded on complete graphs is trivially $\simul$-bounded.
To show that a parameter $p$ is not $\simul$-bounded, we have to give a hereditary completable graph class $\cC$ and a family of graphs $(G_n)_{n \in \N}$ such that the following holds: $p$ is bounded on $\cC$, there is a constant $c$  with $\simul[\cC][G_n] \leq c$ for every $n \in \N$, and $p$ is unbounded on $(G_n)_{n \in \N}$.

In the following subsections we give an exhaustive overview on $\simul$-boundedness of graph parameters. In \cref{subsec:obsv} we give some basic results and intuition on $\simul$-boundedness. In particular, we present simple conditions which we can use to quickly decide whether some parameter is $\simul$-bounded or not. Using these conditions, we obtain $\simul$-boundedness of cliquewidth, twin-width, shrub-depth, flip-width, merge-width, and neighborhood diversity. In \cref{subsec:classes}, we consider families of parameters which are defined using graph classes. These families include the parameters clique cover number, edge clique cover number, interval number, boxicity, and $\distanceto{\cP}$. We focus on the parameter thinness and some of its variants in \cref{subsec:thinness}. In \cref{subsec:binary}, we consider parameters that are defined via binary decomposition trees, in particular mim-width, o-mim-width, sim-width, and sm-width. Similarly, we investigate parameters defined via tree- and path-decompositions in \cref{subsec:tree}, giving results on induced matching treewidth, tree independence number, and tree-length. Finally, in \cref{subsec:modules}, we study how parameters related to modules behave in terms of $\simul$-boundedness. This includes, e.g., modular-width and iterated type partition.

\subsection{Basic Observations} \label{subsec:obsv}

Note that $\simul$-boundedness of some parameter implies $\simulnonempty$-boundedness, as $\simul[\cC][G] \leq \simulnonempty[\cC][G]$ holds for every hereditary completable $\cC$ and every graph $G$.
In fact, $\simul$-boundedness is strictly stronger than $\simulnonempty$-boundedness, as is illustrated by the independence number.

\begin{theorem} \label{thm:indsimulbounded}
    The independence number is $\simulnonempty$-bounded but not $\simul$-bounded.
    In particular, $\alpha(G) \leq \alpha(\cC) \cdot \simulnonempty[\cC][G]$ for any hereditary completable $\cC$ and any graph $G$.
\end{theorem}

\begin{proof}
    First, we show that $\alpha$ is $\simulnonempty$-bounded.
    Consider some $k \in \N$ and a graph class $\cC$ such that $\alpha(H) \leq k$ holds for every $H \in \cC$.
    Furthermore, let $G$ be a graph with $d = \simulnonempty[\cC][G]$ and $(H,L)$ be a $d$-simultaneous $\cC$-representation of $G$ with non-empty label sets.
    Consider some maximum independent set $S \subseteq V(G)$ in $G$.
    If $\abs{S} \leq k$, then nothing is to be shown. Thus, assume that $\abs{S} > k$.
    Let $v_1, v_2, \dots, v_{k+1} \in S$ be chosen arbitrarily.
    As $\alpha(H) \leq k$, the vertices $v_1, \dots, v_{k+1}$ do not form an independent set in $H$.
    Thus, at least two vertices $v_i$ and $v_j$ are adjacent in $H$ but not in $G$, implying $L(v_i) \cap L(v_j) = \emptyset$.
    Therefore, in every selection of $k+1$ vertices of $S$ at least two vertices have disjoint label sets.
    Hence, for every label $1 \leq \ell \leq d$, the label sets of at most $k$ vertices in $S$ contain $\ell$, implying $\abs{S} \leq kd$.

    This proof fails for simultaneous representations with possibly empty label sets.
    In this case, $S$ contains all vertices with empty label sets and the number of these vertices cannot be bounded by $d$.
    This is best illustrated by edgeless graphs, which have unbounded independence number but simultaneous $\cC$-number~$0$ for all hereditary completable $\cC$.
\end{proof}

Because of this relationship between $\simul$-boundedness and $\simulnonempty$-boundedness we mostly prefer to show $\simul$-boundedness.

We have already seen one property that trivially implies for some parameters that they are $\simul$-bounded, namely when they are unbounded on complete graphs.
There are several other parameters for which $\simul$-boundedness can simply be derived from known results, using first-order transductions. Such transductions consider so called vertex colored graphs which essentially are graphs equipped with a vertex labeling.
A \emph{first-order formula} on a vertex colored graph is a logical formula which only quantifies over vertices and uses the following atomic formulas: $x=y$, $E(x,y)$ (meaning $x$ is adjacent to $y$) and $M_i(x)$ with $i \in \N$ (meaning $x$ has label $i$).
We denote by $(G,c) \models \varphi(a_1, \dots, a_k)$ that the edge relation of $G$ and vertex labeling $c : V(G) \to \cP(\N)$ satisfy $\varphi$ when the free variables are fixed to vertices $a_1, \dots, a_k$.
For more information on first-order logic see e.g. \cite{ebbinghaus2021mathematical}.
A \emph{non-copying first-order transduction} (FO-transduction for short) $\sfT$ is defined by a first-order formula $\varphi(x,y)$ on vertex colored graphs with two free variables: for a graph $H$, the set $\sfT(H)$ contains all graphs $G$ with $V(G) \subseteq V(H)$, for which there is a vertex labeling $c$ of $V(H)$ such that $uv \in E(G)$ if and only if $(H,c) \models \varphi(u,v)$.
A graph class $\cD$ is a \emph{$\sfT$-transduction} of a graph class $\cC$ if $\cD \subseteq \bigcup_{H \in \cC} \sfT(H)$.
For an overview on FO-transductions see e.g. \cite{mendez2021firstorder}.

We can easily see that for any hereditary completable graph class $\cC$ the class of graphs with simultaneous $\cC$-number at most $d$ is an FO-transduction of $\cC$ with the formula 
\[\varphi(x,y) := E(x,y) \wedge \bigvee_{i=1}^d \!\br{M_i(x) \wedge M_i(y)}. \]

\begin{lemma}
    Parameters closed under non-copying FO-transductions are $\simul$-bounded.
\end{lemma}

Closure under FO-transductions has been shown for some parameters, giving us their $\simul$-boundedness.
Examples are cliquewidth~\cite{Colcombet07, Courcelle_Engelfriet_2012}, twin-width~\cite{BonnetKTW22}, shrub-depth~\cite{GanianHNOM19}, flip-width~\cite{Torunczyk23}, merge-width~\cite{DreierT25}, and neighborhood diversity\footnote{It seems that closure under non-copying first-order transductions of neighborhood diversity has never been explicitly stated. However, it is implied by the closure of shrub-depth under FO-transductions.}.

\begin{theorem} \label{thm:fo}
    Cliquewidth, twin-width, shrub-depth, flip-width, merge-width, and neighborhood diversity are $\simul$-bounded.
\end{theorem}

We note here that the bounds which are used for showing closure under first-order transductions are often not explicitly given or very large. Because of that, we have examined all of these parameters and have shown their $\simul$-boundedness directly. We were even able to show that our bounds for cliquewidth and twin-width are tight up to a constant factor. These proofs can be found in Appendix \ref{sec:FObounds}.

Further, note that the list of parameters given in \cref{thm:fo} is complete in the sense that all other parameters considered here are not closed under first-order transductions. Justification of this will be given in Appendix \ref{sec:justification}.

We now want to give simple conditions on when a parameter is \emph{not} $\simul$-bounded. Any graph parameter that is $\simul$-bounded and bounded on complete graphs is upper bounded by $\simul[\mathsf{complete}]$, with $\mathsf{complete}$ being the class of complete graphs. We will later see in \cref{cor:eccsimulclique} that $\simul[\mathsf{complete}]$ is equal to the {edge clique cover number}, which is defined as the minimum number of cliques necessary to cover all edges of a graph. This implies the following.

\begin{observation} \label{obsv:ecc1}
    If a parameter is $\simul$-bounded and bounded on complete graphs, then it is upper bounded by the edge clique cover number.
\end{observation}

This observation holds only for $\simul$-boundedness but not for $\simulnonempty$-boundedness since $\simulnonempty[\mathsf{complete}]$ may be unbounded in the edge clique cover number. The independence number is an example for this: it is incomparable to the edge clique cover number and bounded on complete graphs, but it still is $\simulnonempty$-bounded.
However, as discussed in the previous section, differences between $\simul[\cC]$ and $\simulnonempty[\cC]$ can only occur whenever a graph contains isolated vertices. This leads to the following observation.

\begin{observation} \label{obsv:ecc2}
    Let $p$ be a parameter that is bounded on complete graphs.
    If there is a family of graphs that do not contain isolated vertices, that have bounded edge clique cover number, and that have unbounded $p$, then $p$ is neither $\simul$-bounded nor $\simulnonempty$-bounded. 
\end{observation}

A third condition considers the behavior of parameters under disjoint union.

\begin{lemma} \label{lem:maxlower}
    Let $p$ be a parameter that is bounded on some hereditary completable graph class $\cP$ which is closed under disjoint union, but is unbounded on the class of all graphs.
    If there is a number $k \in \N$ such that for every two graphs $G_1, G_2$ with $p(G_1), p(G_2) \geq k$ it holds that $p(G_1 \cup G_2) > \max\!\set{p(G_1), p(G_2)}$, then $p$ is neither $\simul$-bounded nor $\simulnonempty$-bounded.
\end{lemma}

\begin{proof}
    If a graph class $\cP$ is closed under disjoint union, then the following holds for any two graphs $G_1, G_2$: $\simulnonempty[\cP][G_1 \cup G_2] \leq \max\!\set{\simulnonempty[\cP][G_1], \simulnonempty[\cP][G_2]}$.
    Thus, if we take some graph $G$ with $p(G) \geq k$ and consider the graph $nG$, which is obtained by taking $n$ disjoint unions of $G$, we have by assumption that $p(nG) \geq k+n$, while $\simulnonempty[\cP][nG] \leq \abs{V(G)}^2$ which is constant in $n$.
    A suitable graph $G$ always exists, because $p$ is unbounded on the class of all graphs.
    This illustrates that $\simulnonempty[\cP]$ does not upper bound $p$, which together with the fact that $p$ is bounded on $\cP$ implies that $p$ is neither $\simul$-bounded nor $\simulnonempty$-bounded.
\end{proof}

\cref{obsv:ecc1,obsv:ecc2,lem:maxlower} can be used as a first test to check if $\simul$-boundedness of a parameter is plausible. For example, the \emph{clique cover number}, the minimum number of cliques necessary to cover all vertices of a graph, is not comparable to the edge clique cover number and therefore by \cref{obsv:ecc1} not $\simul$-bounded. We will later use these results to show for certain parameters that they are not $\simul$-bounded.

Before starting the comprehensive study of $\simul$-boundedness, we recall that this property has no transitivity property. This means that there are parameters that are not $\simul$-bounded, but are both upper and lower bounded by $\simul$-bounded parameters (see \cref{fig:diagram}). We will list these examples at the end of \cref{sec:simulbounded}.

\subsection{Parameters Defined via Graph Classes} \label{subsec:classes}

There are many families of parameters that are defined using graph classes.
The first we consider are \emph{generalized colorings}, which have been considered as early as the 1970's \cite{jones1974pchromatic}.

\begin{definition}
    Let $\cP$ be a hereditary graph class and let $G$ be a graph.
    The \emph{$\cP$-chromatic number} of $G$, denoted by $\chi_{\cP}(G)$, is the minimum integer $k$ for which there is a partition $X_1, X_2, \dots, X_k$ of $V(G)$, called a \emph{$\cP$-coloring},
    where each $G[X_i]$ is in $\cP$.
\end{definition}

If $\cP$ is not completable, then $\chi_\cP$ is unbounded for complete graphs and, thus, it is trivially $\simul$-bounded. However, there are completable classes $\cP$ where $\chi_\cP$ is not $\simul$-bounded. An example is the class of complete graphs, whose generalized coloring number is equivalent to the before-mentioned clique cover number. Contrarily, we can show that for all completable~$\cP$, the generalized coloring number is $\simulnonempty$-bounded. Furthermore, we can characterize the completable classes $\cP$ where $\chi_\cP$ is also $\simul$-bounded. 

\begin{theorem} \label{thm:chromatic}
    Let $\cP$ be a hereditary completable graph class.
    Then $\chi_{\cP}$ is $\simulnonempty$-bounded.
    In particular, $\chi_{\cP}(G) \leq \chi_{\cP}(\cC) \cdot \simulnonempty[\cC][G]$ holds for any hereditary completable $\cC$ and any graph $G$.
    Furthermore, $\chi_{\cP}$ is $\simul$-bounded if and only if $\cP$ contains all edgeless graphs.
    In that case, $\chi_{\cP}(G) \leq \chi_{\cP}(\cC) \cdot \simul[\cC][G] +1$ holds for any hereditary completable $\cC$ and any graph $G$.
\end{theorem}

\begin{proof}
    We first consider the case where no restrictions are placed on $\cP$.
    Let $k \in \N$ and $\cC$ be a graph class such that $\chi_{\cP} (H) \leq k$ for all $H \in \cC$.
    Furthermore, let $G$ be a graph with $d = \simulnonempty[\cC][G]$ and $(H,L)$ be a $d$-simultaneous $\cC$-representation of $G$ with non-empty label sets.
    Consider a $\cP$-coloring $X_1, \dots, X_k$ of $H$ of size $k$.
    We construct a $\cP$-coloring of $G$ with size $k d$ as follows.
    For all $1 \leq i \leq d$, we define $V_i := \setcond{v \in V(G)}{i \in L(v)}$. Note that $G[V_i] = H[V_i]$ holds for all $i$.
    Thus, as $\cP$ is hereditary, it follows that $G[V_i \cap X_j] = H[V_i \cap X_j] \in \cP$ for every $1 \leq i \leq d$ and $1 \leq j \leq k$, showing that $\br{V_i \cap X_j}_{1 \leq i \leq d, 1 \leq j \leq k}$ is a $\cP$-coloring of $G$ of size~$kd$.

    We now consider the case that $\cP$ contains all edgeless graphs.
    Let $\cC$ and $G$ be chosen as above with a $\simul[\cC][G]$-simultaneous $\cC$-representation $(H,L)$ of $G$.
    We define $V_0 := \setcond{v \in V(G)}{L(v) = \emptyset}$. The graph $G[V_0]$ is in $\cP$ as it is edgeless.
    As above, we obtain a $\cP$-coloring of $G - V_0$ of size $kd$. Combined with $V_0$, it forms a $\cP$-coloring of $G$ of size~$kd +1$.

    Conversely, assume that $\cP$ does not contain all edgeless graphs. Since $\cP$ is hereditary, there is a constant $\ell$ such that $\cP$ does not contain any edgeless graph with more than $\ell$ vertices. Therefore, for the edgeless graph $G$ with $k \cdot \ell$ vertices $\chi_\cP(G)$ is at least $k$. However, $\simul[\cP][G] = 0$. Thus, $\simul[\cP]$ does not upper bound $\chi_\cP$ although $\chi_\cP$ is bounded on $\cP$.
\end{proof}

By choosing $\cP = \cC$ in the above theorem, we obtain the lower bound $\chi_{\cC}(G) \leq \simulnonempty[\cC][G]$ for any hereditary completable graph class $\cC$ and $\chi_{\cC}(G) - 1 \leq \simul[\cC][G]$ if $\cC$ furthermore contains all edgeless graphs.
It is known due to Scheinermann~\cite{scheinermann1992generalizedchromatic} that there are families of graphs with unbounded $\cP$-chromatic number for every non-trivial infinite hereditary graph class $\cP$.
Thus, we know that for any non-trivial hereditary completable graph class $\cC$ there are families of graphs with unbounded simultaneous $\cC$-number. 
The proof given by Scheinermann~\cite{scheinermann1992generalizedchromatic} is probabilistic and gives no construction of such graphs. 
We will give an explicit construction of families of graphs with unbounded $\cP$-chromatic number (and, thus, unbounded $\simul[\cP]$) for every non-trivial hereditary graph class $\cP$ in \cref{thm:lexprod}.

Another family of parameters associated with graph classes are \emph{intersection dimensions}, introduced by Cozzens and Roberts~\cite{cozzens1989dimensional} and Kratochvíl and Tuza~\cite{kratochvil1994intersectiondimension}.

\begin{definition}[Intersection Dimension]
    Given a graph class $\cP$ and a graph $G$ we define the \emph{intersection dimension of $G$ with respect to $\cP$}, denoted by $\intdim[\cP][G]$, as the minimum number of graphs in $\cP$ required such that their intersection results in $G$.
\end{definition}

This definition encompasses some well-known graph parameters such as \emph{boxicity}, \emph{cubicity}, and \emph{chordality} with $\cP$ being the class of interval graphs, proper interval graphs, or chordal graphs, respectively.

Kratochvíl and Tuza~\cite{kratochvil1994intersectiondimension} have given a characterization of the graph classes for which the intersection dimension is well-defined: That is, $\intdim[\cP]$ is well-defined if and only if $\cP$ contains all complete graphs and all complete graphs where one edge has been removed. When taking $\cP$ to be this minimal graph class for which $\intdim[\cP]$ is well-defined, we have $\intdim[\cP][G] = \max\!\set{1, \binom{\abs{V(G)}}{2} - \abs{E(G)}}$ for any graph $G$. This parameter is obviously bounded on complete graphs but unbounded on graphs with edge clique cover number two, therefore, by \cref{obsv:ecc2}, it is neither $\simul$-bounded nor $\simulnonempty$-bounded. 

We will now characterize the graph classes whose intersection dimension is $\simul$-bounded.
To this end, let $G$ be a graph with some $d$-simultaneous $\cC$-representation $(H,L)$ for some hereditary completable $\cC$.
We define a graph $F$ based on the label function $L$ with $V(F) = V(G)$ and $E(F) = \setcond{uv}{L(u) \cap L(v) \neq \emptyset}$.
Note that $G$ is equal to $H \cap F$. Thus, the intersection dimension of $G$ with respect to some graph class $\cP$ is upper bounded by the sum of the intersection dimensions of $H$ and $F$.

We will now give a construction of $F$ as the intersection of $2^d$ graphs and identify the corresponding graph class.
Let $A \subseteq \set{1,\dots,d}$ be non-empty.
We define a graph $F_A$ with $V(F_A) = V(G)$ and $E(F_A) = \setcond{uv}{L(u) \cap L(v) \neq \emptyset} \cup \setcond{uv}{L(u) \neq A \text{ and } L(v) \neq A}$.
That is, $F_A$ is a complete graph where all edges between vertices $x,y$ with $L(x) = A$ and $L(y) \cap A = \emptyset$ are removed, i.e., a complete graph where the edge set of some complete bipartite subgraph is removed.
These graphs are exactly the $(3K_1, C_4, P_4)$-free graphs, as we show in Appendix \ref{sec:proofs}.
If all label sets used in $L$ are non-empty, we are done and can easily see that the intersection of all graphs $F_A$ is exactly $F$.
Otherwise, we define the graph $F_\emptyset$ on $V(G)$ with $E(G) = \setcond{uv}{L(u) \neq \emptyset \text{ and } L(v) \neq \emptyset}$.
That is, the vertices with non-empty label sets form a clique in $F_\emptyset$ and the vertices with empty label sets are isolated, i.e., $F_\emptyset$ is the disjoint union of a complete and an edgeless graph.
These graphs are exactly the $(2K_2, P_3)$-free graphs, as we show in Appendix \ref{sec:proofs}.
Then, the intersection of all graphs $F_A$ and $F_\emptyset$ is exactly $F$.
We use this construction to show the following theorem.

\begin{theorem} \label{thm:intersectiondimension}
    Let $\cP$ be a graph class.
    Then $\intdim[\cP]$ is $\simulnonempty$-bounded if and only if $\cP$ is bounded on the class of $\br{3K_1, C_4, P_4}$-free graphs.
    In particular, if $\intdim[\cP]$ is bounded by $c$ on $\br{3K_1, C_4, P_4}$-free graphs, then $\intdim[\cP][G] \leq c \cdot 2^{\simulnonemptybr{\cC}{G}} + \intdim[\cP][\cC]$ for any hereditary completable $\cC$ and any graph $G$.
    
    Furthermore, $\intdim[\cP]$ is $\simul$-bounded if and only if $\cP$ is bounded on the classes of $\br{3K_1, C_4, P_4}$-free graphs and $(2K_2, P_3)$-free graphs.
    In particular, if $\intdim[\cP]$ is bounded by $c$ on $\br{3K_1, C_4, P_4}$-free graphs and $(2K_2, P_3)$-free graphs, then $\intdim[\cP][G] \leq c\cdot 2^{\simulbr{\cC}{G}} + \intdim[\cP][\cC]$ for any hereditary completable $\cC$ and any graph $G$.
\end{theorem}

\begin{proof}
    Let $\cP$ be a graph class such that $\intdim[\cP]$ is bounded by $c$ on $\br{3K_1, C_4, P_4}$-free graphs.
    Let $\cC$ be a hereditary completable graph class such that $\intdim[\cP][H] \leq k$ for all $H \in \cC$ and let $G$ be some graph with $\simulnonempty[\cC][G] = d$ and $(H,L)$ a $d$-simultaneous $\cC$-representation of $G$ with non-empty label sets.
    Then we can use the construction above to show that $G$ is the intersection of $k$ graphs in $\cP$ and $2^d - 1$ $\br{3K_1, C_4, P_4}$-free graphs. We can substitute each $\br{3K_1, C_4, P_4}$-free graph by $c$ graphs of $\cP$ in the intersection, giving us the desired bound.

    On the other hand, every $\br{3K_1, C_4, P_4}$-free graphs admits a $2$-simultaneous $\mathsf{complete}$-representation using only non-empty label sets. Thus, if $\intdim[\cP]$ is $\simulnonempty$-bounded, then $\intdim[\cP]$ must also be bounded on $\br{3K_1, C_4, P_4}$-free graphs.

    Let $\cP$ be a graph class such that $\intdim[\cP]$ is bounded by $c$ on $\br{3K_1, C_4, P_4}$-free graphs and $(2K_2, P_3)$-free graphs. We can show analogously to the above that $\intdim[\cP]$ is $\simul$-bounded.

    Conversely, every $(2K_2, P_3)$-free graph admits a $1$-simultaneous $\mathsf{complete}$-representation that uses empty label sets. Thus, if $\intdim[\cP]$ is $\simul$-bounded, then $\intdim[\cP]$ must also be bounded on $(2K_2, P_3)$-free graphs. Using the fact that $\simul$-boundedness implies $\simulnonempty$-boundedness, $\intdim[\cP]$ must also be bounded on $\br{3K_1, C_4, P_4}$-free graphs.
\end{proof}

We can give another characterization of when an intersection dimension is $\simul$-bounded, namely when it is upper bounded by the edge clique cover number. This follows easily from the facts that all $\simul$-bounded parameters are upper bounded by the edge clique cover number (see \cref{obsv:ecc1}) and that being upper bounded by the edge clique cover number implies boundedness on $\br{3K_1, C_4, P_4}$-free graphs and $(2K_2, P_3)$-free graphs.

\begin{corollary}
    Let $\cP$ be a graph class. Then $\intdim[\cP]$ is $\simul$-bounded if and only if $\intdim[\cP]$ is upper bounded by the edge clique cover number.
\end{corollary}

Note that both restrictions on the class $\cP$ are satisfied by (proper) interval graphs and chordal graphs as these graph classes contain $\br{3K_1, C_4, P_4}$-free graphs and $(2K_2, P_3)$-free graphs. Therefore cubicity, boxicity, and chordality are $\simul$-bounded. 
The exponential bound in \cref{thm:intersectiondimension} is necessary.
To see this, consider boxicity and the class of interval graphs which are exactly the graphs of boxicity $1$. 
The boxicity of complements of matchings is linear in the number of vertices~\cite{roberts1969boxicity} while the simultaneous interval number of complements of matchings is logarithmic in the number of vertices~\cite{beisegel2024simultaneousinterval}, yielding an exponential relationship between the two parameters.

Cozzens and Roberts~\cite{cozzens1989dimensional} also considered the dual to intersection dimensions, where intersections are replaced by unions. However, they required all graphs used in graph unions to be on the same vertex set, which restricts the graph classes for which such parameters can be defined.
A less restrictive definition, which differs at most by one from these parameters, is offered by \emph{covering numbers}~\cite{harary1970covering,KnauerU16,pyber1992covering}.

\begin{definition}[Covering Number]
    Let $\cP$ be a graph class and $G$ be a graph.
    Graphs $G_1, \dots, G_k \in \cP$ are called a \emph{$\cP$-cover} of $G$ if each $G_i$ is a subgraph of $G$ and every edge of $G$ is contained in at least one $G_i$. 
    The minimum cardinality of a $\cP$-cover of $G$ is called the \emph{covering number of $G$ with respect to $\cP$} and denoted by $\cov[\cP][G]$.
\end{definition}

Note that covering numbers are well-defined for exactly those graph classes $\cP$ which contain the graph $K_2$. 
Obviously, every graph admits a $\set{K_2}$-cover. On the other hand, $\cov[\cP][K_2]$ is well-defined only if $\cP$ contains $K_2$.

In contrast to the intersection dimension, we can show that all well-defined covering numbers are $\simul$-bounded.

\begin{theorem} \label{thm:covering}
    Let $\cP$ be a graph class containing $K_2$.
    Then $\cov[\cP]$ is $\simul$-bounded.
    In particular, $\cov[\cP][G] \leq \cov[\cP][\cC] \cdot \simul[\cC][G]$ for any hereditary completable $\cC$ and any graph~$G$.
\end{theorem}

\begin{proof}
    Let $k \in \N$ and $\cC$ be a graph class such that $\cov[\cP][H] \leq k$ for all $H \in \cC$.
    Furthermore, let $G$ be a graph with $d = \simul[\cC][G]$ and $(H,L)$ be a $d$-simultaneous $\cC$-representation of $G$.
    For every label $1 \leq i \leq d$, we define a vertex set $V_i := \setcond{v \in V(G)}{i \in L(v)}$. 
    Note that the subgraphs $G[V_1], \dots, G[V_d]$ cover all edges of $G$.
    Furthermore, $G[V_i] = H[V_i] \in \cC$ for all $i$, implying that $G[V_i]$ admits a $\cP$-cover of size at most $k$.
    The collection of $\cP$-covers of the graphs $G[V_i]$ gives a $\cP$-cover of $G$ with size at most $kd$.
\end{proof}

A variant of the covering numbers, called \emph{local covering numbers}, have been introduced by Knauer and Ueckerdt~\cite{KnauerU16}.

\begin{definition}[Local Covering Number]
    Let $\cP$ be a graph class and $G$ be a graph.
    The \emph{frequency} of a $\cP$-cover $G_1, \dots, G_k$ is defined by $\max_{v \in V(G)} \abs{\setcond{i}{v \in V(G_i)}}$.
    The \emph{local covering number of $G$ with respect to $\cP$}, denoted by $\covl[\cP][G]$, is the minimum frequency of any $\cP$-cover of $G$.
\end{definition}

Again, the local covering number is well-defined for exactly those graph classes which contain the graph $K_2$.

Using the same arguments as the proof of \cref{thm:covering}, we obtain the following result.

\begin{theorem}
    Let $\cP$ be a graph class containing $K_2$.
    Then $\covl[\cP]$ is $\simul$-bounded.
    In particular, $\covl[\cP][G] \leq \covl[\cP][\cC] \cdot \simul[\cC][G]$ for any hereditary completable $\cC$ and any graph~$G$.
\end{theorem}

By a finer analysis of the local frequency in the proof, the given upper bound on the local covering number can be improved to use the maximum cardinality of any used label set rather than the number of labels. Thus, for $d$-simultaneous $\cC$-representations where $d$ is large but each vertex only uses few labels, the bound can be much improved.

Knauer and Ueckerdt~\cite{KnauerU16} also introduced the \emph{folded covering numbers}. 
Their definition requires homomorphisms.
A \emph{homomorphism} from a graph $F$ to a graph $G$ is a map $\varphi : V(F) \to V(G)$ such that $uv \in E(F)$ implies $\varphi(u) \varphi(v) \in E(G)$.
A homomorphism $\varphi$ is \emph{edge-surjective} if for every edge $u'v' \in E(G)$ there exists an edge $uv \in E(F)$ with $\varphi(u) = u'$ and $\varphi(v) = v'$.

\begin{definition}[Folded Covering Number]
    Let $\cP$ be a graph class and $G$ be a graph.
    A \emph{folded $\cP$-cover} of $G$ is an edge-surjective homomorphism from some graph of $\cP$ to $G$.
    The \emph{folded covering number of $G$ with respect to $\cP$}, denoted by $\covf[\cP][G]$, is defined by
    \[ \covf[\cP][G] := \min \! \setcond{ \max_{v \in V(G)} \abs{\varphi^{-1} (v)} }{\varphi \text{ is a folded } \cP \text{-cover of } G}. \]
\end{definition}

We again want to mention the graph classes $\cP$ for which the folded covering numbers are well-defined.
Consider a graph $G$ that consists of $k$ non-trivial bipartite connected components.
A graph used in a folded $\cP$-cover of $G$ must therefore consist of at least $k$ non-trivial bipartite connected components.
Thus, it is necessary for $\cP$ to contain graphs consisting of arbitrarily many non-trivial bipartite connected components.
This condition is actually sufficient: For a graph $G$ with $m$ edges, take a graph $F$ of $\cP$ with $\geq m$ non-trivial bipartite connected components. We assign each edge $uv$ of $G$ a non-trivial bipartite component $F_{uv}$ of $F$. For each edge $uv \in E(G)$, $\varphi$ maps one set of the bipartition of $F_{uv}$ to $u$ and the other to $v$. If there are unmapped bipartite components left in $F$ after this process has been done for every edge in $G$, they can similarly be mapped to some edge. Thus, $\varphi$ is a folded $\cP$-cover of $G$.

Not every graph class satisfying this condition is $\simul$-bounded. Take, for example, the graph class $\cP = \mathsf{complete} \cup \setcond{nK_2 \cup mK_1}{n,m \in \N}$ which consists of complete graphs and graphs that contain matchings plus isolated vertices. The parameter $\covf[\cP]$ is obviously bounded on complete graphs. However, it can easily be seen that $\covf[\cP][2K_n] = n-1$ for all $n \in \N$, where $2K_n$ denotes the disjoint union of two complete graphs $K_n$. Thus, by \cref{obsv:ecc2}, the parameter $\covf[\cP]$ is neither $\simul$-bounded nor $\simulnonempty$-bounded. Furthermore, $\cP$ was chosen to be hereditary, thus adding this as a requirement to $\cP$ does not help with $\simul$-boundedness. Hence, we need to require some condition on $\cP$ to show $\simul$-boundedness of folded covering numbers.

\begin{theorem}
    Let $\cP$ be a graph class that is closed under disjoint union.
    Then $\covf[\cP]$ is $\simul$-bounded.
    In particular, $\covf[\cP][G] \leq \covf[\cP][\cC] \cdot \simul[\cC][G]$ for any hereditary completable $\cC$ and any graph $G$.
\end{theorem}

\begin{proof}
    Let $k \in \N$ and $\cC$ be a graph class such that $\covf[\cP][H] \leq k$ for all $H \in \cC$.
    Furthermore, let $G$ be a graph with $d = \simul[\cC][G]$ and $(H,L)$ be a $d$-simultaneous $\cC$-representation of $G$.
    We define sets $V_i$ for $1 \leq i \leq d$ as in the proof of \cref{thm:covering} and note that each graph $G[V_i]$ admits a folded $\cP$-cover where the preimage of every vertex has size at most $k$.
    For every $V_i$, let $\varphi_i$ be the folded $\cP$-cover of $G[V_i]$ which maps from graph $F_i$ to $G[V_i]$.
    We define the graph $F$ as the disjoint union of all graphs $F_i$ and a folded $\cP$-cover $\varphi : V(F) \to V(G)$ by $\varphi(v) = \varphi_i(v)$ for all $v \in V(F_i)$ for $1 \leq i \leq d$.
    Then, the preimage of every vertex in $\varphi$ has size at most $kd$.
\end{proof}

Note that the condition on $\cP$ to be closed under disjoint union seems to be rather mild since Knauer and Ueckerdt~\cite{KnauerU16} only considered covering numbers on such classes.

As a direct consequence of these results on covering numbers, we have that the \emph{track number}, which is the covering number with respect to interval graphs, the \emph{interval number}, which is the folded covering number with respect to interval graphs, and the \emph{edge clique cover number}, which is the covering number with respect to complete graphs, are $\simul$-bounded.
Together with the fact that the edge clique cover number on complete graphs is bounded by one and the result by Beisegel et al.~\cite{beisegel2024simultaneousinterval} that the edge clique cover number is always an upper bound on $\simul[\cC]$, we obtain the following.

\begin{corollary} \label{cor:eccsimulclique}
    The edge clique cover number is equal to the simultaneous $\mathsf{complete}$-number.
\end{corollary}

Interestingly, simultaneous $\cC$-numbers are $\simul$-bounded themselves, which follows from the next result.

\begin{lemma} \label{lem:c1c2}
    Let $G$ be a graph and $\cC_1, \cC_2$ be two hereditary completable graph classes.
    If $G$ has a $d_1$-simultaneous $\cC_1$-representation $(H_1,L_1)$ and $H_1$ has a $d_2$-simultaneous $\cC_2$-representation $(H_2, L_2)$, then $G$ admits a $(d_1 \cdot d_2)$-simultaneous $\cC_2$-representation.
\end{lemma}

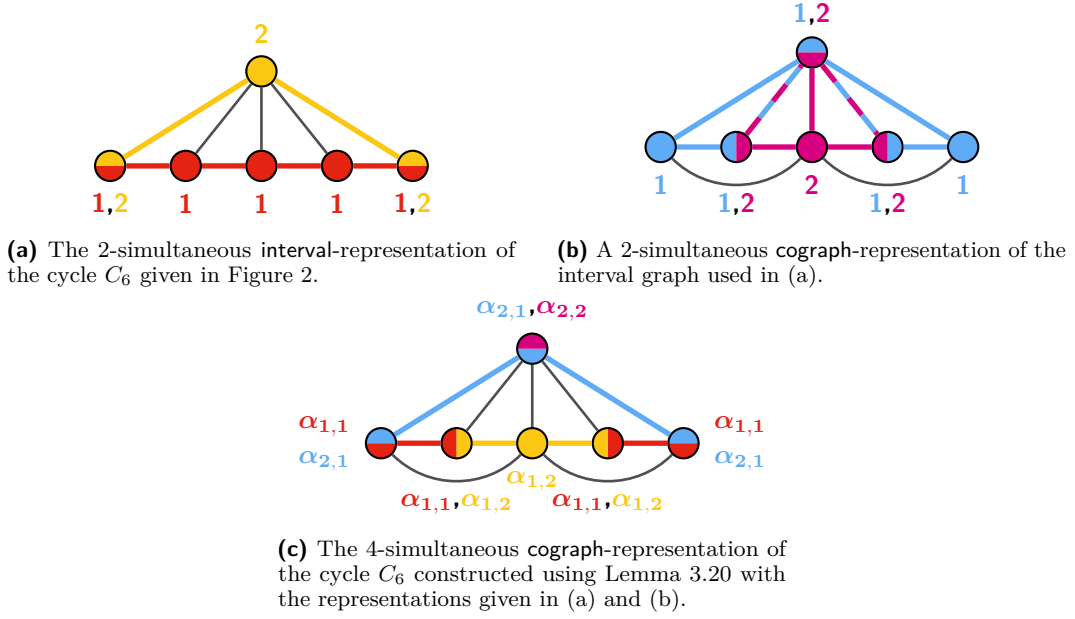
\begin{figure}
    \centering
        \begin{subfigure}[t]{0.48\textwidth}
        \centering
        \begin{tikzpicture}
            \def\a{.4} 
            \def\d{8pt}
            \def\labelone{BTUred}
            \def\labeltwo{lipicsYellow}
            \def\labelthree{coolblue}
            \def\nolabel{lipicsGray}
            \tikzset{c/.style={circle,draw,thick,minimum size=\a cm}} 
            \tikzset{one color/.style={c,fill=#1}}  
            \tikzset{pics/two colors/.style args=
                {#1|#2|rotate=#3}{code={%
            \fill[#1,rotate=#3] (0,\a/2) arc(90:270:\a/2)--cycle;           
            \fill[#2,rotate=#3] (0,\a/2) arc(90:-90:\a/2)--cycle;
            \path (0,0) node[c] (-boundary) {};
            }}}
            \tikzset{pics/three colors/.style args=
                {#1|#2|#3|rotate=#4}{code={%
            \fill[#1,rotate=#4] (0,\a/2) arc(90:210:\a/2)--(0,0)--cycle;            
            \fill[#2,rotate=#4] (0,\a/2) arc(90:-30:\a/2)--(0,0)--cycle;
            \fill[#3,rotate=#4] (210:\a/2) arc(210:330:\a/2)--(0,0)--cycle;
            \path (0,0) node[c] (-boundary) {};
            }}}
            \path
            (1,0) node[one color = \labelone] (R1) {}
            (2,0) node[one color = \labelone] (R2) {}
            (3,0) node[one color = \labelone] (R3) {}
            (2,1.25) node[one color = \labeltwo] (Y) {}
            ;
            \path
            (0,0) pic (YR1) {two colors={\labelone|\labeltwo|rotate=90}}
            (4,0) pic (YR2) {two colors={\labelone|\labeltwo|rotate=90}}
            ;
            \draw[line width=2pt,\labeltwo] (YR1-boundary) -- (Y) -- (YR2-boundary);
            \draw[line width=2pt,\labelone] (YR1-boundary) -- (R1) -- (R2) -- (R3) -- (YR2-boundary);
            \draw[line width=1pt,\nolabel] (R1) -- (Y) -- (R2);
            \draw[line width=1pt,\nolabel] (Y) -- (R3);
            \path
            (0,-0.5) node[font=\sffamily\bfseries] {\textcolor{\labelone}{1},\textcolor{\labeltwo}{2}}
            (1,-0.5) node[font=\sffamily\bfseries] {\textcolor{\labelone}{1}}
            (2,-0.5) node[font=\sffamily\bfseries] {\textcolor{\labelone}{1}}
            (3,-0.5) node[font=\sffamily\bfseries] {\textcolor{\labelone}{1}}
            (4,-0.5) node[font=\sffamily\bfseries] {\textcolor{\labelone}{1},\textcolor{\labeltwo}{2}}
            (2,1.75) node[font=\sffamily\bfseries] {\textcolor{\labeltwo}{2}}
            ;
        \end{tikzpicture}
        \caption{The $2$-simultaneous $\mathsf{interval}$-representation of the cycle $C_6$ given in \cref{fig:example1}.}
        \label{fig:subfig1}
    \end{subfigure}
    \hfill
    \begin{subfigure}[t]{0.48\textwidth}
        \centering
        \begin{tikzpicture}
            \def\a{.4} 
            \def\d{8pt}
            \def\labelone{BTUred}
            \def\labeltwo{lipicsYellow}
            \def\labelthree{coolblue}
            \def\labelfour{BTUmagenta}
            \def\nolabel{lipicsGray}
            \tikzset{c/.style={circle,draw,thick,minimum size=\a cm}} 
            \tikzset{one color/.style={c,fill=#1}}  
            \tikzset{pics/two colors/.style args=
                {#1|#2|rotate=#3}{code={%
            \fill[#1,rotate=#3] (0,\a/2) arc(90:270:\a/2)--cycle;           
            \fill[#2,rotate=#3] (0,\a/2) arc(90:-90:\a/2)--cycle;
            \path (0,0) node[c] (-boundary) {};
            }}}
            \tikzset{pics/three colors/.style args=
                {#1|#2|#3|rotate=#4}{code={%
            \fill[#1,rotate=#4] (0,\a/2) arc(90:210:\a/2)--(0,0)--cycle;            
            \fill[#2,rotate=#4] (0,\a/2) arc(90:-30:\a/2)--(0,0)--cycle;
            \fill[#3,rotate=#4] (210:\a/2) arc(210:330:\a/2)--(0,0)--cycle;
            \path (0,0) node[c] (-boundary) {};
            }}}
            \path
            (0,0) node[one color=\labelthree] (1) {}
            (2,0) node[one color=\labelfour] (3) {}
            (4,0) node[one color=\labelthree] (5) {}
            ;
            \path
            (1,0) pic (2) {two colors={\labelthree|\labelfour|rotate=0}}
            (3,0) pic (4) {two colors={\labelthree|\labelfour|rotate=180}}
            (2,1.25) pic (6) {two colors={\labelthree|\labelfour|rotate=-90}}
            ;
            \draw[line width=2pt,\labelthree] (4-boundary) -- (5) -- (6-boundary) -- (1) -- (2-boundary);
            \draw[line width=2pt,\labelfour] (4-boundary) -- (3) -- (2-boundary);
            \draw[line width=2pt,\labelthree,dash pattern= on \d off \d,dash phase=\d] (2-boundary) -- (6-boundary) -- (4-boundary);
            \draw[line width=2pt,\labelfour,dash pattern= on \d off \d] (2-boundary) -- (6-boundary) -- (4-boundary);
            \draw[line width=2pt,\labelfour] (3) -- (6-boundary);
            \draw[line width=1pt,\nolabel] (1) to[bend right=45] (3);
            \draw[line width=1pt,\nolabel] (3) to[bend right=45] (5);
            \path
            (0,-0.5) node[font=\sffamily\bfseries] {\textcolor{\labelthree}{1}}
            (1,-0.75) node[font=\sffamily\bfseries] {\textcolor{\labelthree}{1},\textcolor{\labelfour}{2}}
            (2,-0.5) node[font=\sffamily\bfseries] {\textcolor{\labelfour}{2}}
            (3,-0.75) node[font=\sffamily\bfseries] {\textcolor{\labelthree}{1},\textcolor{\labelfour}{2}}
            (4,-0.5) node[font=\sffamily\bfseries] {\textcolor{\labelthree}{1}}
            (2,1.75) node[font=\sffamily\bfseries] {\textcolor{\labelthree}{1},\textcolor{\labelfour}{2}}
            ;
        \end{tikzpicture}
        \caption{A $2$-simultaneous $\mathsf{cograph}$-representation of the interval graph used in (a).}
        \label{fig:subfig2}
    \end{subfigure}
    \begin{subfigure}[b]{0.48\textwidth}
        \centering
        \begin{tikzpicture}
            \def\a{.4} 
            \def\d{8pt}
            \def\labelone{BTUred}
            \def\labeltwo{lipicsYellow}
            \def\labelthree{coolblue}
            \def\labelfour{BTUmagenta}
            \def\nolabel{lipicsGray}
            \tikzset{c/.style={circle,draw,thick,minimum size=\a cm}} 
            \tikzset{one color/.style={c,fill=#1}}  
            \tikzset{pics/two colors/.style args=
                {#1|#2|rotate=#3}{code={%
            \fill[#1,rotate=#3] (0,\a/2) arc(90:270:\a/2)--cycle;           
            \fill[#2,rotate=#3] (0,\a/2) arc(90:-90:\a/2)--cycle;
            \path (0,0) node[c] (-boundary) {};
            }}}
            \tikzset{pics/three colors/.style args=
                {#1|#2|#3|rotate=#4}{code={%
            \fill[#1,rotate=#4] (0,\a/2) arc(90:210:\a/2)--(0,0)--cycle;            
            \fill[#2,rotate=#4] (0,\a/2) arc(90:-30:\a/2)--(0,0)--cycle;
            \fill[#3,rotate=#4] (210:\a/2) arc(210:330:\a/2)--(0,0)--cycle;
            \path (0,0) node[c] (-boundary) {};
            }}}
            \path
            (0,0) pic (1) {two colors={\labelone|\labelthree|rotate=90}}
            (1,0) pic (2) {two colors={\labelone|\labeltwo|rotate=0}}
            (2,0) node[one color=\labeltwo] (3) {}
            (3,0) pic (4) {two colors={\labelone|\labeltwo|rotate=180}}
            (4,0) pic (5) {two colors={\labelone|\labelthree|rotate=90}}
            (2,1.25) pic (6) {two colors={\labelthree|\labelfour|rotate=90}}
            ;
            \draw[line width=2pt,\labelthree] (1-boundary) -- (6-boundary) -- (5-boundary);
            \draw[line width=2pt,\labeltwo] (2-boundary) -- (3) -- (4-boundary);
            \draw[line width=2pt,\labelone] (1-boundary) -- (2-boundary);
            \draw[line width=2pt,\labelone] (4-boundary) -- (5-boundary);
            \draw[line width=1pt,\nolabel] (2-boundary) -- (6-boundary) -- (4-boundary);
            \draw[line width=1pt,\nolabel] (3) -- (6-boundary);
            \draw[line width=1pt,\nolabel] (1-boundary) to[bend right=45] (3);
            \draw[line width=1pt,\nolabel] (3) to[bend right=45] (5-boundary);
            \path
            (-0.75,0) node[font=\sffamily\bfseries,align=left] {$\textcolor{\labelone}{\boldsymbol{\alpha}_{\mathsf{\mathbf{1,1}}}}$\\$\textcolor{\labelthree}{\boldsymbol{\alpha}_{\mathsf{\mathbf{2,1}}}}$}
            (1,-0.8) node[font=\sffamily\bfseries] {$\textcolor{\labelone}{\boldsymbol{\alpha}_{\mathsf{\mathbf{1,1}}}}$,$\textcolor{\labeltwo}{\boldsymbol{\alpha}_{\mathsf{\mathbf{1,2}}}}$}
            (2,-0.5) node[font=\sffamily\bfseries] {$\textcolor{\labeltwo}{\boldsymbol{\alpha}_{\mathsf{\mathbf{1,2}}}}$}
            (3,-0.8) node[font=\sffamily\bfseries] {$\textcolor{\labelone}{\boldsymbol{\alpha}_{\mathsf{\mathbf{1,1}}}}$,$\textcolor{\labeltwo}{\boldsymbol{\alpha}_{\mathsf{\mathbf{1,2}}}}$}
            (4.75,0) node[font=\sffamily\bfseries,align=left] {$\textcolor{\labelone}{\boldsymbol{\alpha}_{\mathsf{\mathbf{1,1}}}}$\\$\textcolor{\labelthree}{\boldsymbol{\alpha}_{\mathsf{\mathbf{2,1}}}}$}
            (2,1.75) node[font=\sffamily\bfseries] {$\textcolor{\labelthree}{\boldsymbol{\alpha}_{\mathsf{\mathbf{2,1}}}}$,$\textcolor{\labelfour}{\boldsymbol{\alpha}_{\mathsf{\mathbf{2,2}}}}$}
            ;
        \end{tikzpicture}
        \caption{The $4$-simultaneous $\mathsf{cograph}$-representation of the cycle $C_6$ constructed using \cref{lem:c1c2} with the representations given in (a) and (b).}
        \label{fig:subfig3}
    \end{subfigure}
    \caption{An example of \cref{lem:c1c2} applied to the cycle $C_6$ using interval graphs and cographs. Note that the label $\alpha_{2,2}$ used in (c) is redundant.}
    \label{fig:c1c2}
\end{figure}

\begin{proof}
    We introduce labels $\alpha_{ij}$ with $1 \leq i \leq d_1$ and $1 \leq j \leq d_2$ and define
    $L(v) := \setcond{\alpha_{ij}}{i \in L_1(v) \text{ and } j \in L_2(v)}$ (see \cref{fig:c1c2} for an example).
    We claim that $(H_2, L)$ is a $(d_1 \cdot d_2)$-simultaneous $\cC_2$-representation of $G$.

    Let $uv \in E(G) \subseteq E(H_1) \subseteq E(H_2)$.
    As we have $uv \in E(G)$, we know that $L_1(u) \cap L_1(v) \neq \emptyset$.
    Furthermore, as $uv \in E(H_1)$, we have $L_2(u) \cap L_2(v) \neq \emptyset$.
    Let $i \in L_1(u) \cap L_1(v)$ and $j \in L_2(u) \cap L_2(v)$.
    Then $\alpha_{ij} \in L(u) \cap L(v)$, showing $L(u) \cap L(v) \neq \emptyset$.

    Now, let $u,v \in V(G)$ be given with $uv \in E(H_2)$ and $L(u) \cap L(v) \neq \emptyset$.
    Choose some $\alpha_{ij} \in L(u) \cap L(v)$.
    Then we have $j \in L_2(u) \cap L_2(v)$, which together with $uv \in E(H_2)$ implies $uv \in E(H_1)$.
    Thus, $i \in L_1(u) \cap L_1(v) \neq \emptyset$ implies $uv \in E(G)$, as $(H_1,L_1)$ is a simultaneous $\cC_1$-representation of $G$.
\end{proof}

If one of the simultaneous representations in the proof above contains empty label sets, then the constructed representation also contains empty label sets. Thus, if one wants to construct a simultaneous representation with non-empty label sets, both representations used need to have non-empty label sets.

\begin{corollary}
    Let $\cP$ be a hereditary completable graph class. Then $\simul[\cP]$ is $\simul$-bounded.
    In particular, $\simul[\cP][G] \leq \simul[\cP][\cC] \cdot \simul[\cC][G]$ for any hereditary completable $\cC$ and any graph $G$.
    Furthermore, $\simulnonempty[\cP]$ is $\simulnonempty$-bounded.
    In particular, $\simulnonempty[\cP][G] \leq \simulnonempty[\cP][\cC] \cdot \simulnonempty[\cC][G]$ for any hereditary completable $\cC$ and any graph $G$.
\end{corollary}

Simultaneous $\cC$-numbers with non-empty label sets are not necessarily $\simul$-bounded. For example, the simultaneous $\mathsf{complete}$-number with non-empty label sets is not upper bounded by the edge clique cover number, as witnessed by edgeless graphs. Therefore, the simultaneous $\mathsf{complete}$-number with non-empty label sets is $\simulnonempty$-bounded but not $\simul$-bounded.

Finally we consider the parameters $\distanceto{\cP}$, which for a graph class $\cP$ count the number of vertices which have to be removed from a given graph to obtain a graph in $\cP$.
We have trivial cases in which these parameters are $\simul$-bounded.
These are whenever $\distanceto{\cP}$ is unbounded on complete graphs and when $\distanceto{\cP}$ is bounded on the class of all graphs.
However, we can show that for most interesting graph classes these parameters are not $\simul$-bounded by making use of \cref{obsv:ecc2}.

\begin{lemma}
    Let $\cP$ be a graph class such that $\cP$ admits at least one forbidden induced subgraph $G$.
    Then $\distanceto{\cP}$ is not upper bounded by the edge clique cover number.
\end{lemma}

\begin{proof}
    Let $G$ be a forbidden induced subgraph for $\cP$.
    Let $n = \abs{V(G)}$ and $m = \abs{E(G)}$.
    For $k \in \N$ we construct a graph $G_k$ by replacing every vertex of $G$ with complete graph of size $k$.
    Note that $\distanceto{\cP}(G_k) \geq k$ holds for all $k \in \N$ as $G_k$ contains $k$ disjoint copies of the forbidden subgraph $G$.
    Furthermore, the edge clique cover number of $G_k$ is bounded by $n+m$.
    An appropriate edge clique cover can be constructed by adding for every vertex of $G$ the clique that replaced it and for every edge of $G$ the two cliques that replaced the vertices incident to that edge.
    Note that the summand $n$ is only needed in the case that $G$ contains isolated vertices.
    As $n+m$ is constant in $k$, this shows that $\distanceto{\cP}$ is not upper bounded by the edge clique cover number.
\end{proof}

This together with \cref{obsv:ecc2} implies the following.

\begin{theorem}
    Let $\cP$ be a graph class such that $\distanceto{\cP}$ is bounded on complete graphs and $\cP$ admits at least one forbidden induced subgraph.
    Then $\distanceto{\cP}$ is neither $\simul$-bounded nor $\simulnonempty$-bounded.
\end{theorem}

The conditions of the theorem hold for every non-trivial completable hereditary graph class $\cC$.
Thus, for each such class $\cC$, we have an example of a parameter which is bounded on that class, but not a lower bound on $\simul[\cC]$.

\subsection{Thinness and Its Variants} \label{subsec:thinness}

Another family of parameters that is related to some graph classes is given by thinness~\cite{mannino2007thinness} and its variants.
These parameters are derived from the structure of interval graphs, i.e., intersection graphs of intervals on the real line.
The following characterization of these graphs is well-known~\cite{olariu1991optimal}: a graph $G$ is an interval graph if and only if there exists an ordering $v_1, \dots, v_n$ of the vertices $V(G)$ such that for each triple $(r,s,t)$ with $r < s < t$ if $v_r v_t \in E(G)$ then $v_s v_t \in E(G)$.
Thinness generalizes this by introducing a partition of $V(G)$ and forcing this property only for triples where the first two vertices belong to the same class.

Formally, a graph $G$ is \emph{$k$-thin}, if there exist an ordering $\sigma = v_1, v_2, \dots, v_n$ of $V(G)$ and a partition of $V(G)$ into $k$ classes $\set{V^1, V^2, \dots, V^k}$ such that, for each triple $(r,s,t)$ with $r < s < t$, if $v_r$ and $v_s$ belong to the same class $V^i$ and $v_t v_r \in E(G)$, then $v_t v_s \in E(G)$. Such an ordering is \emph{consistent} with the partition. 

A well-known subclass of interval graphs is given by proper interval graphs. This class can also be characterized by means of vertex orderings, where the above property for interval graphs must hold for the ordering and its reverse~\cite{looges1993optimal}. This can again be generalized to the parameter proper thinness~\cite{bonomo2019properthinness}.

Formally, a graph $G$ is called \emph{proper $k$-thin} if there exists an ordering $\sigma$ of $V(G)$ and a partition $\set{V^1, \dots, V^k}$ of $V(G)$ such that both $\sigma$ and its reversal are consistent with $\set{V^1, \dots, V^k}$. Such an ordering is called \emph{strongly consistent} with the partition.

Recently, the variants independent thinness and complete thinness have been introduced~\cite{bonomobraberman2022thinnessproduct}, where the sets of a partition are required to induce an independent or complete subgraph. We generalize this notion to $\cP$-thinness for some graph class $\cP$.
A graph $G$ is called \emph{(proper) $k$-$\cP$-thin} if there exist an ordering $\sigma$ of $V(G)$ and a partition $\set{V^1, \dots, V^k}$ such that $\sigma$ is (strongly) consistent with $\set{V^1, \dots, V^k}$ and each $V^i$ induces a graph of $\cP$ in $G$.
By the definition of (proper) $k$-thin graphs, every set of such a partition necessarily induces a (proper) interval graph. Therefore, only subclasses of (proper) interval graphs have to be considered in this definition.

\begin{definition}[Thinness]
    Let $G$ be a graph.
    The minimum $k$ such that $G$ is $k$-thin is called the \emph{thinness} of $G$, denoted $\thin(G)$.
    The minimum $k$ such that $G$ is proper $k$-thin is called the \emph{proper thinness} of $G$, denoted $\propthin(G)$.
    Let $\cP$ be a subclass of (proper) interval graphs. The minimum $k$ such that $G$ is (proper) $k$-$\cP$-thin is called the \emph{(proper) $\cP$-thinness} of $G$, denoted $\thin_\cP(G)$ respectively $\propthin_\cP(G)$.
\end{definition}

We note that (proper) thinness is exactly the (proper) $\cP$-thinness where $\cP$ is the class of (proper) interval graphs. Therefore, we only consider (proper) $\cP$-thinness moving forward.
Furthermore, we only consider hereditary classes $\cP$. Otherwise, for a partition and ordering that are consistent, refinements of that partition may not yield a valid partition. Moreover, requiring $\cP$ to be hereditary makes (proper) $\cP$-thinness monotone.
Using this, we again have the following: if $\cP$ is not completable, then (proper) $\cP$-thinness is not bounded on complete graphs and therefore trivially $\simul$-bounded. We now give a result similar to \cref{thm:chromatic} for (proper) $\cP$-thinness.

\begin{theorem}
    Let $\cP$ be a hereditary completable subclass of (proper) interval graphs.
    Then (proper) $\cP$-thinness is $\simulnonempty$-bounded.
    In particular, $\thin_\cP(G) \leq \thin_\cP (\cC) \cdot 2^{\simulnonemptybr{\cC}{G}}$ and $\propthin_\cP(G) \leq \propthin_\cP (\cC) \cdot 2^{\simulnonemptybr{\cC}{G}}$ for any hereditary completable $\cC$ and any graph $G$.
    
    Moreover, (proper) $\cP$-thinness is $\simul$-bounded if and only if $\cP$ contains all edgeless graphs.
    In this case, $\thin_\cP(G) \leq \thin_\cP (\cC) \cdot 2^{\simulbr{\cC}{G}}$ and $\propthin_\cP(G) \leq \propthin_\cP (\cC) \cdot 2^{\simulbr{\cC}{G}}$ for any hereditary completable $\cC$ and any graph $G$.
\end{theorem}

\begin{proof}
    Let $k \in \N$ and $\cC$ be a graph class such that $\thin_\cP(H) \leq k$ for all $H \in \cC$.
    Furthermore, let $G$ be a graph with $d = \simulnonempty[\cC][G]$ and $(H,L)$ a $d$-simultaneous $\cC$-representation of $G$ with non-empty label sets.
    Let $\sigma = v_1, \dots, v_n$ and $\set{V^1, \dots, V^k}$ be an ordering and a partition of $V(H) = V(G)$ that verifies $\thin_\cP(H) \leq k$.
    We define for every $1 \leq i \leq k$ and every $\emptyset \neq A \subseteq \set{1,2,\dots,d}$ a set $V_A^i := \setcond{v}{v \in V^i \text{ and } L(v) = A}$ and claim that the ordering $v_1, v_2, \dots, v_n$ together with the partition $\mathcal{V} :=\setcond{V_A^i}{1 \leq i \leq k, \ \emptyset \neq A \subseteq \set{1,\dots,d}}$ verify $\thin_\cP(G) \leq k 2^d$.

    First, observe that $G[V_A^i] = H[V_A^i]$ is an induced subgraph of $H[V^i]$ for every $A$. Furthermore, as $\cP$ is hereditary and $H[V^i]$ induces a graph of $\cP$, so does $G[V_A^i]$.
    Thus, it remains to show that $\sigma$ is consistent with $\mathcal{V}$.

    To this end, let $v_r, v_s, v_t$ be three vertices with $r < s < t$, $v_t v_r \in E(G)$ and $v_r, v_s$ belong to the same class in $\mathcal{V}$.
    By definition, we have that $v_r v_t \in E(H)$ and $v_r$ and $v_s$ are in the same class of the partition $\set{V^1, \dots, V^k}$,
    which implies that $v_t v_s \in E(H)$ holds.
    Furthermore, as $v_r v_t \in E(G)$, we have $L(v_r) \cap L(v_t) \neq \emptyset$.
    As $v_r$ and $v_s$ are in the same class of $\mathcal{V}$ they have the same label set, which implies $L(v_t) \cap L(v_s) \neq \emptyset$.
    Altogether, we obtain $v_t v_s \in E(G)$, which was to be shown.

    Note that this proof works exactly the same for strongly consistent orderings. Furthermore, the proof for $\simul$-boundedness under the condition that $\cP$ contains all edgeless graphs works almost the same. The only difference lies in the definition of $\mathcal{V}$, which now must take into account the case that $A = \emptyset$. Then we have that $G[V_\emptyset^i]$ induces an edgeless graph for every $i$, which by assumption is in $\cP$.

    Any consistent partition and ordering of a graph define a $\cP$-coloring. Therefore, the fact that $\cP$ containing all edgeless graphs being necessary for (proper) $\cP$-thinness to be $\simul$-bounded follows analogously to \cref{thm:chromatic}.
\end{proof}

By choosing $\cP$ as the class of (proper) interval graphs in the above theorem, we obtain the $\simul$-boundedness of both thinness and proper thinness. The bound given in the theorem is tight for thinness and proper thinness with the following argument:
Thinness and proper thinness are bounded by $1$ on complete graphs. By \cref{cor:eccsimulclique}, the edge clique cover number is equal to the simultaneous $\mathsf{complete}$-number, which is logarithmic on complements of matchings~\cite{gregory1982onclique}. On the other hand, thinness, and thus proper thinness, is known to be at least linear on complements of matchings~\cite{mannino2002solving}, yielding an exponential relationship.

There are other interesting variations of thinness, namely precedence (proper) thinness~\cite{bonomobraberman2022precedence} and (proper) mixed-thinness~\cite{balaban2024twin}.
The ideas used to show $\simul$-boundedness of (proper) $\cP$-thinness also work for (proper) mixed-thinness. We omit the details here.
In contrast, we now show that precedence (proper) thinness is not $\simul$-bounded.

A graph is \emph{precedence $k$-thin}, if there exist an ordering $\sigma = v_1, \dots, v_n$ of $V(G)$ and a partition $\set{V^1, \dots, V^k}$ of $V(G)$ such that $\sigma$ is consistent with $\set{V^1, \dots, V^k}$ and the vertices of each set $V^i$ appear consecutively in $\sigma$. The minimum $k$ such that $G$ is precedence $k$-thin is called the \emph{precedence thinness} of $G$, denoted $\precthin(G)$. One can define \emph{precedence $\cP$-thinness} and \emph{precedence proper $\cP$-thinness}, denoted $\precthin_\cP(G)$ and $\precpropthin_\cP(G)$ respectively, analogously to the standard definition. 

We make use of \cref{lem:maxlower} to show that precedence $\cP$-thinness is not $\simul$-bounded for suitable graph classes $\cP$.
Bonomo-Braberman et al.~\cite{bonomobraberman2024thinness} have given a formula for the precedence thinness of disjoint unions of graphs: $\precthin(G_1 \cup G_2) = \precthin(G_1) + \precthin(G_2) - 1$. The upper bound follows easily by concatenating the orders of $G_1$ and $G_2$ and taking the union of the last set in the partition related to $G_1$ with the first set in the partition related to $G_2$. For the lower bound, they show that for every consistent partition and ordering of $G_1 \cup G_2$ that satisfies the precedence constraint there is another consistent partition and ordering that satisfies the precedence constraint, such that at most one set of the partition is \emph{mixed}. That is, at most one set of the partition contains vertices of both $G_1$ and $G_2$. This immediately implies the lower bound. They show this by transforming a partition and ordering by splitting apart the left-most mixed set $V^i$ according to the ordering, taking one part of the set as a new set in the partition replacing $V^i$ and adding the other part to $V^{i+1}$. They show that the part added to $V^{i+1}$ has no edges to that set, therefore the graph induced by it is the disjoint union of the graphs induced by the two sets. During that process the left-most mixed set moves to the right and, therefore, at most one set can be mixed at the end of the process. This proof works for precedence $\cP$-thinness under certain conditions on $\cP$, namely that the process described above does not impact membership in $\cP$. We need $\cP$ to be hereditary and closed under disjoint union. This observation gives us a slight strengthening of the formula given in~\cite{bonomobraberman2024thinness}. 

\begin{lemma} \label{lem:precthinunion}
    Let $\cP$ be a hereditary subclass of interval graphs that is closed under disjoint union.
    Then $\precthin_\cP(G_1 \cup G_2) = \precthin_\cP(G_1) + \precthin_\cP(G_2) - 1$ for any two graphs $G_1$ and $G_2$.
\end{lemma}

This result shows that precedence $\cP$-thinness for suitable $\cP$ satisfies the conditions of \cref{lem:maxlower} for a parameter value of $k = 2$.

\begin{theorem}
    Let $\cP$ be a hereditary completable subclass of interval graphs that is closed under disjoint union. Then $\precthin_\cP$ is neither $\simul$-bounded nor $\simulnonempty$-bounded.
\end{theorem}

We note that the proof of \cref{lem:precthinunion} can not easily be adapted to precedence proper thinness. However, it seems reasonable that a similar formula would also hold for precedence proper thinness. 
To show that this parameter is not $\simul$-bounded, we therefore directly give a suitable family of graphs with bounded simultaneous $\mathsf{proper}\text{-}\mathsf{interval}$-number and unbounded precedence proper thinness.
This family of graphs is the family $(nK_{1,3})_{n \in \N}$, where $nK_{1,3}$ is the disjoint union of $n$ \emph{claws}. Obviously, a single claw has bounded simultaneous $\mathsf{proper}\text{-}\mathsf{interval}$-number and it does not increase when taking disjoint unions. Thus, it suffices to show that the precedence proper thinness of this family is unbounded.

\begin{lemma}
    Let $n \in \N$. Then $\precpropthin(nK_{1,3}) = n+1$.
\end{lemma}

\begin{proof}
    Let $n \in \N$ and $nK_{1,3}$ be given.
    We denote the vertices of the $k$-th copy of $K_{1,3}$ by $v_1^{(k)}, v_2^{(k)}, v_3^{(k)}, v_4^{(k)}$, where $v_1^{(k)}$ has degree $3$ and the other three vertices have degree $1$.
    We give an ordering and a partition verifying that $\precpropthin(K_{1,3}) \leq n+1$.
    We define the ordering as $\sigma := v_1^{(1)}, v_2^{(1)}, v_3^{(1)}, v_4^{(1)}, v_1^{(2)}, \dots, v_4^{(n)}$ and the sets of the partition as follows: $V^1 := \set{ v_1^{(1)}}$, $V^{n+1} := \set{v_2^{(n)}, v_3^{(n)}, v_4^{(n)}}$ and $V^k := \set{ v_2^{(k-1)}, v_3^{(k-1)}, v_4^{(k-1)}, v_1^{(k)}}$ for $2 \leq k \leq n$. It can easily be seen that the ordering is strongly consistent with the partition and the sets of the partition are consecutive in the ordering.

    It remains to show that $\precpropthin(nK_{1,3}) \geq n+1$.
    Suppose that we have a strongly consistent ordering $\sigma$ and partition $\mathcal{V}$ that satisfies the precedence constraint.
    First, we note that the vertex sets of each copy of $K_{1,3}$ have to be placed into at least two sets of the partition, as a $K_{1,3}$ is not a proper interval graph.
    Next, we show that any set of $\mathcal{V}$ can contain vertices of at most two claws. Assume to the contrary that there is a set $V^p$ which contains vertices of three claws. Without loss of generality, we have that $V^p$ contains $v_i^{(1)}, v_j^{(2)}, v_k^{(3)}$ for some $i,j,k$ and we have $v_i^{(1)} < v_j^{(2)} < v_k^{(3)}$ in $\sigma$. This can be achieved by renaming the copies of $K_{1,3}$. Furthermore, $v_j^{(2)}$ must have an edge going out of $V^p$. Without loss of generality, this edge goes to a vertex in set $V^q$, which lies to the right of $V^p$ in $\sigma$. This can be achieved by possibly considering the reverse of $\sigma$.
    Thus, we have a vertex $z$ in $V^q$ with $v_j^{(2)}z \in E(nK_{1,3})$. Furthermore, we have $v_j^{(2)} < v_k^{(3)} < z$ in $\sigma$ and $v_j^{(2)}, v_k^{(3)} \in V^p$, which by consistency implies that $v_k^{(3)}z \in E(nK_{1,3})$, a contradiction.
    Thus, each set of $\mathcal{V}$ only contains vertices of at most two claws.
    This together with the fact that each claw must have vertices in at least two sets imply that there are at least $n$ sets in $\mathcal{V}$ by the pigeonhole principle.
    We furthermore claim that the first and last set of $\mathcal{V}$ ordered according to $\sigma$ can only contain vertices of one claw, which implies the lower bound $n+1$. We show this only for the first set $V^1$, as it follows by symmetry for the last set.
    Assume there are vertices of two claws contained in $V^1$. Without loss of generality, let $v_i^{(1)}, v_j^{(2)}$ with some $i,j$ be contained in $V^1$ such that $v_i^{(1)} < v_j^{(2)}$ in $\sigma$. Then $v_i^{(1)}$ must have a neighbor in a set $V^k$ which is to the right of $V^1$ in $\sigma$. Again, by consistency, $v_j^{(2)}$ must have the same neighbor, which is impossible. Thus, $V^1$ contains vertices of only one claw, which gives the desired bound.
\end{proof}

As mentioned above, this implies the following.

\begin{theorem}
    Precedence proper thinness is neither $\simul$-bounded nor $\simulnonempty$-bounded.
\end{theorem}

\subsection{Binary Decomposition Trees} \label{subsec:binary}

Several graph parameters can be defined by using binary decomposition trees.
A \emph{binary decomposition tree} of a graph $G$ is a pair $\br{\br{T,r}\!, \delta}$ with a binary rooted tree $\br{T,r}$ and a bijection $\delta : V(G) \to \Lambda(T)$, where $\Lambda(T)$ denotes the set of leaves of $T$.
For every $a \in V(T)$, let $\Lambda_a \subseteq \Lambda(T)$ be the leaves with ancestor $a$ and define $V_a := \setcond{\delta^{-1}(v)}{v \in \Lambda_a}$. A \emph{caterpillar decomposition} is a binary decomposition tree $((T,r), \delta)$ where $T$ is a caterpillar, that is, a tree where the removal of all leaves results in a path.

\begin{definition} \label{def:dectree}
    Let $f = (f_G)_{G \in \bbG}$ be a family of functions where $f_G : \cP(V(G)) \to \bbR$ is a set function on $V(G)$ for a graph $G \in \bbG$.
    The \emph{$f_G$-width} of a binary decomposition tree $\br{\br{T,r}\!, \delta}$ of a graph $G$ is the maximum value of $f_G\!\br{V_a}$ over all $a \in V(T)$ and the \emph{$f$-width} of $G$ is the minimum $f_G$-width over all binary decomposition trees of $G$.
    The \emph{linear $f$-width} of $G$ is defined analogously using caterpillar decompositions.
\end{definition}

Binary decomposition trees are closely related to \emph{branch decompositions}. These are defined using unrooted ternary trees where the vertex set which is given as input to $f_G$ is obtained by removing any edge $e$ from $T$ and taking the preimage of the leaves of a connected component of $T-e$. Thus, to be well defined, functions used with branch decompositions must be symmetric. This restriction is not necessary for binary decomposition trees as the ancestor relation directly defines the input set of the function $f_G$.
However, if a given family of functions $(f_G)_{G \in \bbG}$ is symmetric, it defines the same parameter using binary decomposition trees and branch decompositions.

We show a sufficient condition for $\simul$-boundedness of parameters defined via binary decomposition trees, which only depends on the used family of functions.

\begin{lemma} \label{lem:bintree}
    Consider a family of functions $f = (f_G)_{G \in \bbG}$.
    If for every graph $G$, every $d$-simultaneous $\cC$-representation $(H,L)$ of $G$ for some hereditary completable graph class $\cC$, and every vertex set $A \subseteq V(G) = V(H)$, it holds that $f_G(A) \leq g\!\br{d, f_H(A)}$ for some non-decreasing function $g : \N \times \bbR \to \bbR$, then (linear) $f$-width is $\simul$-bounded.
\end{lemma}

\begin{proof}
    Consider some hereditary completable graph class $\cC$, such that (linear) $f$-width is bounded by $k$ on graphs in $\cC$.
    Let $G$ be some graph with $d = \simul[\cC][G]$ and let $(H,L)$ be a $d$-simultaneous $\cC$-representation.
    Let $\br{(T,r),\delta}$ be a binary decomposition tree (caterpillar decomposition) of $f_H$-width $\leq k$.
    Then for every $a \in V(T)$ we have $f_G(V_a) \leq g\!\br{d, f_H(V_a)} \leq g(d, k)$, which shows that the $f_G$-width of $\br{(T,r), \delta}$ and, thus, the $f$-width of $G$ are bounded by $g(d,k)$.
\end{proof}

Eiben et al.~\cite{eiben2022unifying} have given a framework to define parameters via binary decomposition trees using obstructions. That is, a family of obstructions is given and $f_G(A)$ outputs the size of the largest such obstruction found in the cut defined by $A$.
Bergougnoux et al.~\cite{bergougnoux2025algorithmic} have shown that this framework produces, up to functional equivalence, eleven non-trivial parameters defined by five families of obstructions and their combinations.
We show for four of these five families of obstructions that the parameters defined using them are $\simul$-bounded. This will also lead to the $\simul$-boundedness of the combination of these obstructions. We show that the fifth obstruction defines a parameter that is not $\simul$-bounded.
Note that each obstruction is a family of \emph{balanced} bipartite graphs, that are bipartite graphs where the two sets of the bipartition have the same size. Furthermore, for each $A \subseteq V(G)$ we consider the obstructions only in the bipartite subgraph $G[A,\overline{A}]$ of $G$ which contains only the edges between $A$ and $\overline{A} := V(G) \setminus A$ in $G$.

We start with the obstruction of matchings, which gives the well-known parameter mim-width. A \emph{matching} of size $n$ is the disjoint union of $n$ edges, in other words, the graph~$nK_2$.

\begin{definition}[Mim-width]
    For a graph $G$ let $mim_G : \cP(V(G)) \to \N$ be the function with $mim_G(A)$ for $A \subseteq V(G)$ being the size of a maximum induced matching in the subgraph $G\!\ekbr{A, \overline{A}}$.
    The \emph{mim-width} and \emph{linear mim-width} of $G$ are defined using \cref{def:dectree} with $f = (mim_G)_{G \in \bbG}$.
\end{definition}

\begin{theorem} \label{thm:mim}
    Let $d \in \N$, $G$ be a graph, and $(H,L)$ be a $d$-simultaneous $\cC$-representation of~$G$ for some hereditary completable graph class $\cC$.
    Then $mim_G(A) \leq d \cdot mim_H(A)$ for any vertex set $A \subseteq V(G)$. 
    Thus, (linear) mim-width is $\simul$-bounded.
\end{theorem}

\begin{proof}
    Let $G$ be a graph with a $d$-simultaneous $\cC$-representation $(H,L)$.
    If $d = 0$, then $G$ is edgeless and the inequality holds trivially.
    Hence, we may assume that $d \geq 1$.
    Consider some vertex set $A \subseteq V(G)$.
    Let $k = mim_H(A)$ be the size of a maximum induced matching in $H[A, \overline{A}]$.
    We need to show that $mim_G(A) \leq kd$, i.e., the size of a maximum induced matching in $G[A, \overline{A}]$ is bounded by $kd$.
    Let $M$ be a maximum induced matching in $G[A,\overline{A}]$.
    For every edge $e = uv \in M$ we define $\lambda(e) := L(u) \cap L(v)$, which is non-empty for every edge $e$.
    If $\abs{M} \leq k$, we have nothing to show. Thus, we may assume that $\abs{M} > k$.
    Let $e_1 = u_1 v_1, e_2 = u_2 v_2, \dots, e_{k+1} = u_{k+1} v_{k+1} \in M$ be arbitrary where $u_i \in A$ and $v_i \in \overline{A}$ for $1 \leq i \leq k+1$.
    As $mim_H(A) = k$, we have that the edges $e_1, \dots, e_{k+1}$ are no induced matching in $H[A, \overline{A}]$, i.e., there are at least two edges $e_i, e_j$ such that one of the edges $u_i v_j$ or $u_j v_i$ exists in $H$ but not in $G$.
    If $u_i v_j$ exists in $H$, then $L(u_i) \cap L(v_j) = \emptyset$ and if $u_j v_i \in E(H)$, then $L(u_j) \cap L(v_i) = \emptyset$.
    In both cases we have that $\lambda(e_i) \cap \lambda(e_j) = \emptyset$.
    Therefore, in every arbitrary selection of $k+1$ edges of $M$ there are at least two edges with disjoint $\lambda$-sets.
    This implies that for every label $1 \leq \ell \leq d$ there are at most $k$ edges in $M$ whose $\lambda$-sets contain $\ell$.
    Thus, $\abs{M} \leq k d$.
\end{proof}

We continue with the obstruction of anti-matchings. An \emph{anti-matching} is a complete bipartite graph where the edge set of a perfect matching is removed. The size of the anti-matching is the number of edges removed from the complete bipartite graph or half the number of vertices.

\begin{definition}[Miam-width] 
    For a graph $G$ let $miam_G : \cP(V(G)) \to \N$ be the function with $miam_G(A)$ for $A \subseteq V(G)$ being the size of a maximum induced anti-matching in the subgraph $G\!\ekbr{A, \overline{A}}$.
    The \emph{miam-width} and \emph{linear miam-width} of $G$ are defined using \cref{def:dectree} with $f = (miam_G)_{G \in \bbG}$.
\end{definition}

\begin{theorem} \label{thm:miam}
    Let $d \in \N$, $G$ be a graph, and $(H,L)$ be a $d$-simultaneous $\cC$-representation of $G$ for some hereditary completable graph class $\cC$.
    Then $miam_G(A) \leq 2^d \cdot miam_H(A)$ for any vertex set $A \subseteq V(G)$.
    Thus, (linear) miam-width is $\simul$-bounded.
\end{theorem}

\begin{proof}
    Let $G$ be a graph with a $d$-simultaneous $\cC$-representation $(H,L)$.
    If $d = 0$, then $G$ is edgeless and the inequality holds trivially.
    Hence, we may assume that $d \geq 1$.
    Consider some vertex set $A \subseteq V(G)$.
    Let $k = miam_H(A)$ be the size of a maximum induced anti-matching in $H[A, \overline{A}]$.
    We need to show that the size of a maximum induced anti-matching in $G[A, \overline{A}]$ is bounded by $2^dk$.
    Let $A' \subseteq A$ and $B' \subseteq \overline{A}$ be vertex sets which induce a maximum induced anti-matching in $G[A, \overline{A}]$.
    If the anti-matching has size $\leq k$ we have nothing to show, so assume it has size $> k$, in particular size $>1$.
    We claim that for every label set $D \subseteq \set{1,\dots,d}$ at most $k$ vertices of $A'$ can have the label set $D$ in $L$.
    Note first that any vertex of any induced anti-matching of size greater than $1$ has at least one incident edge, so no vertex in $A'$ has an empty label set.
    Consider some arbitrary non-empty label set $D \subseteq \set{1,\dots,d}$ and assume that there are more than $k$ vertices in $A'$ with label set $D$.
    Denote these vertices by $A_D \subseteq A'$ and their corresponding vertices of the anti-matching by $B_D \subseteq B'$.
    As each vertex $u \in B_D$ is adjacent to all vertices of $A_D$ except for one, it holds that $L(u) \cap D \neq \emptyset$ for all $u \in B_D$. 
    This implies that every edge between $A_D$ and $B_D$ in $H$ also exists in $G$, which means that $G[A_D, B_D] = H[A_D, B_D]$.
    In other words, $A_D$ and $B_D$ define an induced anti-matching of size $> k$ in $H[A, \overline{A}]$, a contradiction to the fact that $miam_H(A) = k$.
    Thus, $\abs{A'} \leq 2^dk$.
\end{proof}

This exponential bound is tight up to a constant factor, which can be seen using complements of $n \times n$-grids. It is known that $n \times n$-grids have mim-width at least $\frac{n-1}{3}$~\cite{vatshelle2012mimwdith}. As the miam-width of a graph is exactly the mim-width of its complement, we have that complements of $n \times n$-grids have miam-width at least $\frac{n-1}{3}$. On the other hand, these graphs are complements of graphs of maximum degree $4$ and, thus, their edge clique cover number is at most logarithmic in $n^2$ due to Alon~\cite{alon1986cubicecc}. This also means that the edge clique cover number, which is equal to the simultaneous $\mathsf{complete}$-number due to \cref{cor:eccsimulclique}, is logarithmic in $n$, yielding an exponential relationship.

The next obstruction we consider are chains. A \emph{chain} on $2k$ vertices is a bipartite graph with bipartition $\set{a_1, \dots, a_k}$ and $\set{b_1, \dots, b_k}$ and edge set $E = \setcond{a_i b_j}{1 \leq i \leq j \leq k}$. We say that $k$ is the \emph{half-size} of a chain on $2k$ vertices.

\begin{definition}[Chain-width]
    For a graph $G$ let $chain_G : \cP(V(G)) \to \N$ be the function with $chain_G(A)$ for $A \subseteq V(G)$ being the half-size of a maximum induced chain in the subgraph $G[A, \overline{A}]$.
    The \emph{chain-width} and \emph{linear chain-width} of $G$ are defined using \cref{def:dectree} with $f = (chain_G)_{G \in \bbG}$.
\end{definition}

Using essentially the same proof as for \cref{thm:miam}, we obtain the following result.

\begin{theorem}
    Let $d \in \N$, $G$ be a graph, and $(H,L)$ be a $d$-simultaneous $\cC$-representation of $G$ for some hereditary completable graph class $\cC$.
    Then $chain_G(A) \leq 2^d \cdot chain_H(A)$ for any vertex set $A \subseteq V(G)$.
    Thus, (linear) chain-width is $\simul$-bounded.
\end{theorem}

We do not know whether the exponential bound given for chain-width is tight. In particular, we are only aware of one family of graphs with unbounded chain-width~\cite{bergougnoux2025algorithmic}.

The last two families of obstructions considered by Bergougnoux et al.~\cite{bergougnoux2025algorithmic} are balanced complete bipartite graphs and balanced edgeless bipartite graphs.

\begin{definition}
    For a graph $G$ let $complete_G : \cP(V(G)) \to \N$ be the function with $complete_G(A)$ for $A \subseteq V(G)$ being half of the number of vertices of a maximum induced balanced complete bipartite graph in $G[A, \overline{A}]$.
    Furthermore, let $empty_G : \cP(V(G)) \to \N$ be the function with $empty_G(A)$ for $A \subseteq V(G)$ being half of the number of vertices of a maximum induced balanced edgeless bipartite graph.
    The \emph{complete-width}, \emph{linear complete-width}, \emph{empty-width}, and \emph{linear empty-width} of $G$ are defined using \cref{def:dectree} with $f = (complete_G)_{G \in \bbG}$ and $f = (empty_G)_{G \in \bbG}$ respectively.
\end{definition}

Note that complete-width is trivially $\simul$-bounded as it is unbounded on complete graphs.
The $\simul$-boundedness of complete-width can also be seen using \cref{lem:bintree} as the size of an induced complete bipartite graph can not increase due to removal of edges.
On the other hand, empty-width is not $\simul$-bounded. To show that, we will make use of the following folklore lemma. It can be found for example in~\cite{lewis1965memory,lipton1979separator}.

\begin{lemma} \label{lem:balancedcut}
    Let $(T,r)$ be a rooted binary tree and let $\Lambda(T)$ denote its set of leaves. Denote by $\Lambda_v$ the leaves which are descendants of $v \in V(T)$ in $(T,r)$. 
    Then there exists a node $v \in V(T)$ with $\frac{1}{3} \abs{\Lambda(T)} \leq \abs{\Lambda_v} \leq \frac{2}{3} \abs{\Lambda(T)}$.
\end{lemma}

\begin{theorem} \label{thm:emptywidth}
    Empty-width is neither $\simul$-bounded nor $\simulnonempty$-bounded.
\end{theorem}

\begin{proof}
    Empty-width is bounded on complete graphs.
    We consider the family of graphs $(2K_n)_{n \in \N}$ consisting of the disjoint union of two complete graphs of size $n$.
    Obviously, the edge clique cover number is bounded on this family of graphs.
    We now show that the empty-width of this family is unbounded, which proves the theorem by \cref{obsv:ecc2}.

    Let $n \in \N$ be arbitrary and let $((T,r), \delta)$ be any binary decomposition tree of $2K_n$.
    Denote by $A$ the vertex set of one copy of $K_n$ and by $B$ the vertex set of the other copy of $K_n$.
    By \cref{lem:balancedcut} there is a node $a \in V(T)$ such that $\abs{V_a} \geq \frac{2n}{3}$ and $\abs{\overline{V_a}} \geq \frac{2n}{3}$.
    Without loss of generality, let $\abs{V_a \cap A} \geq \abs{\overline{V_a} \cap A}$. That is, $V_a$ contains at least $\frac{n}{2}$ vertices of $A$ and $\overline{V_a}$ contains at most $\frac{n}{2}$ vertices of $A$.
    Thus, $\overline{V_a}$ must contain at least $\frac{n}{6}$ vertices of $B$.
    These vertices together with $\frac{n}{6}$ vertices of $A$ contained in $V_a$ induce a balanced edgeless bipartite graph in $G[V_a, \overline{V_a}]$, implying that the empty-width of $2K_n$ is at least $\frac{n}{6}$. 
\end{proof}

Now we have dealt with the five families of obstructions given in \cite{bergougnoux2025algorithmic}. These five families of obstructions can be combined to yield six more parameters~\cite{bergougnoux2025algorithmic}. Three of these six parameters are functionally equivalent to well-known parameters, namely treewidth, co-treewidth, and cliquewidth. Note that only the parameter functionally equivalent to co-treewidth uses the empty obstruction, and therefore all other parameters are $\simul$-bounded by combining the suitable results above. On the other hand, co-treewidth is not $\simul$-bounded with the counterexample given in the proof of \cref{thm:emptywidth}.
Because we already have results for the three known parameters, we state the $\simul$-boundedness result only for the other three parameters.

\begin{definition}
    Let $G$ be a graph.
    Define the functions $mam_G$, $chim_G$, and $cham_G$ as follows
    \begin{align*}
        mam_G(A) := \max\{ mim_G(A), miam_G(A) \}, &\ chim_G(A) := \max\{ chain_G(A), mim_G(A) \}, \\
        cham_G(A) := \max\{ &chain_G(A), miam_G(A)\}.
    \end{align*}
    The \emph{(linear) mam-width}, \emph{(linear) chim-width}, and \emph{(linear) cham-width} are defined using \cref{def:dectree} with these functions.
\end{definition}

Because each of the families of functions $(mim_G)_{G \in \bbG}$, $(miam_G)_{G \in \bbG}$, and $(chain_G)_{G \in \bbG}$ satisfies the condition in \cref{lem:bintree}, we have that maximum functions of them also satisfy this condition. This immediately implies the following.

\begin{corollary}
    (Linear) mam-width, (linear) chim-width, and (linear) cham-width are $\simul$-bounded.
\end{corollary}

Thus, out of the eleven parameters definable using this framework, two are trivially $\simul$-bounded, seven are non-trivially $\simul$-bounded and two are not $\simul$-bounded. Note that for the parameters not $\simul$-bounded it holds that they are complementary to the parameters which are trivially $\simul$-bounded. Here we say that a parameter $p$ is \emph{complementary} to a parameter $q$ if for all $G$ it holds that $p(G) = q(\overline{G})$.

Not all parameters defined via binary decomposition trees follow the framework given in \cite{bergougnoux2025algorithmic,eiben2022unifying}. We examine two examples of such parameters, namely \emph{o-mim-width}~\cite{bergougnoux2023omim} and \emph{sim-width}~\cite{kang2017sim}, which are generalizations of mim-width.

\begin{definition}[O-mim-width]
    Let $G = (V,E)$ be a graph.
    For $A \subseteq V(G)$, we denote by $E[A]$ the edges contained in the induced subgraph $G[A]$.
    Define the function $umim_G : \cP\!\br{V(G)} \to \N$ with $umim_G(A)$ for $A \subseteq V(G)$ being the size of an induced matching between $A$ and $\overline{A}$ in $G - E[\overline{A}]$.
    Then define $omim_G : \cP\!\br{V(G)} \to \N$ with $omim_G(A) = \min\!\set{umim_G(A), umim_G(\overline{A})}$.
    Using \cref{def:dectree} with $f = (omim_G)_{G \in \cG}$, we define the \emph{o-mim-width} and \emph{linear o-mim-width} of $G$ as the $f$-width and linear $f$-width of $G$ respectively.
\end{definition}

\begin{definition}[Sim-width] \label{def:sim}
    For a graph $G$ define the function $sim_G : \cP\!\br{V(G)} \to \N$ with $sim(A)$ for $A \subseteq V(G)$ being the size of an induced matching between $A$ and $\overline{A}$ in $G$.
    Using \cref{def:dectree} with $f = (sim_G)_{G \in \cG}$ we define the \emph{sim-width} and \emph{linear sim-width} of $G$ as the $f$-width and linear $f$-width of $G$ respectively.
\end{definition}

Using the same arguments as in the proof of \cref{thm:mim}, we obtain the following result.

\begin{theorem}
    Let $d \in \N$, $G$ be a graph, and $(H,L)$ be a $d$-simultaneous $\cC$-representation of $G$ for some hereditary completable graph class $\cC$.
    Then $omim_G(A) \leq d \cdot omim_H(A)$ for any vertex set $A \subseteq V(G)$.
    Thus, (linear) o-mim-width is $\simul$-bounded.
    Furthermore, $sim_G(A) \leq d \cdot sim_H(A)$ for any vertex set $A \subseteq V(G)$.
    Thus, (linear) sim-width is $\simul$-bounded.
\end{theorem}

Binary decomposition trees can also be used to define the parameters \emph{boolean-width} and \emph{rank-width}, which are functionally equivalent to cliquewidth (see e.g. \cite{vatshelle2012mimwdith}).
Using the fact that cliquewidth is $\simul$-bounded, this immediately gives us the $\simul$-boundedness of these parameters.
However, once again, the bounds on these parameters can be improved by considering them directly, which we do in Appendix \ref{sec:FObounds}.

Finally, we present a parameter that looks rather similar to mim-width, o-mim-width, and sim-width, but is not $\simul$-bounded.
This parameter is the sm-width~\cite{saether2016between}.

\begin{definition}[Sm-width]
    For a graph $G$, let $sm_G : \cP(V(G)) \to \N$ be defined as follows:
    \[ sm_G(A) := \begin{cases}
        1, & \text{if } (A, \overline{A}) \text{ is  a split of } G \\
        \max\!\setcond{\abs{M}}{M \text{ is a matching in } G\!\ekbr{A, \overline{A}}}, & \text{otherwise}
    \end{cases}, \]
    where a cut $(A,B)$ is a split of $G$ if there are $A' \subseteq A$ and $B' \subseteq B$ such that the edges between $A$ and $B$ exactly induce a complete bipartite graph between $A'$ and $B'$.
    The \emph{(linear) sm-width} of $G$ is defined using \cref{def:dectree} with $f = (sm_G)_{G \in \bbG}$.
\end{definition}

We note that sm-width has been defined in a way such that it lies between treewidth and cliquewidth~\cite{saether2016between}. In fact, if we define a cut function using only the second case in the definition of $sm_G$, we obtain a parameter which is functionally equivalent to treewidth~\cite{vatshelle2012mimwdith}.

We make use of \cref{obsv:ecc2} to show that (linear) sm-width is neither $\simul$-bounded nor $\simulnonempty$-bounded.

\begin{lemma} \label{lem:smecc}
    (Linear) sm-width is not upper bounded by the edge clique cover number.
\end{lemma}

\begin{proof}
    We consider a family of graphs $G_n$ obtained by blowing up the vertices of the cycle $C_6$ to cliques of size $n$. This family of graphs obviously has edge clique cover number six.

    Consider a graph $G_n$ together with a binary decomposition tree.
    Let $V_1, V_2, \dots, V_6$ be the cliques in $G_n$ corresponding to the original vertices of the cycle, such that every vertex of $V_i$ is adjacent to every vertex in $V_{i-1}$ and $V_{i+1}$ where $6+1 \equiv 1$ and $1-1 \equiv 6$.
     \begin{claim}
         If $(A,B)$ is a split of $G_n$ with at least one edge, then one of $A$ or $B$ is completely contained in a clique $V_i$.
     \end{claim}
     \begin{claimproof}
         Let $A' \subseteq A$ and $B' \subseteq B$ be chosen in such a way that the edges in the cut $(A,B)$ induce a complete bipartite graph between $A'$ and $B'$.
         The neighborhood of a vertex of $V_i$ is contained in $N_i := V_{i-1} \cup V_i \cup V_{i+1}$.
         We first show that one of $A'$ or $B'$ is contained in a single clique.
         Assume that $A'$ contains vertices of at least two cliques. We consider multiple cases.
         First assume that $A'$ contains vertices of two cliques $V_i$ and $V_{i+3}$. Then by the above statement $B' \subseteq N_i \cap N_{i+3} = \emptyset$. Thus, there are no edges in the cut, a contradiction.
         
         If $A'$ contains vertices of two cliques $V_i$ and $V_{i+2}$, then we have that $B' \subseteq N_{i} \cap N_{i+2} = V_{i+1}$ and, thus, $B'$ is completely contained in a single clique.
         
         If $A'$ contains vertices of two cliques $V_i$ and $V_{i+1}$, then $B' \subseteq N_i \cap N_{i+1} = V_i \cup V_{i+1}$.
         At least one of the sets $A$ and $B$ has to contain vertices of $V_{i-1}$. Vertices of $V_{i-1}$ cannot be in $B$ since, otherwise, they would only be adjacent to the vertices of $V_i$ in $A'$ but not to the vertices of $V_{i+1}$ in $A'$. Thus, $V_{i-1}$ is completely contained in $A$. Now, it follows analogously that $B'$ cannot contain vertices of both $V_i$ and $V_{i+1}$. Hence, $B'$ is either completely contained in $V_i$ or in $V_{i+1}$.

         We have shown that one of $A'$ or $B'$ is fully contained in a single clique.
         Say without loss of generality that $B' \subseteq V_i$ for some $i$.
         Assume that $B$ contains a vertex $v$ in $V_k$ with $k \neq i$.
         As all edges between $A$ and $B$ are incident to $B' \subseteq V_i$, we must have that all vertices of $N_k = V_{k-1} \cup V_k \cup V_{k+1}$ are contained in $B$. It must hold that one of $k-1$ or $k+1$ is not equal to $i$, thus we can iterate this argument with either $N_{k-1}$ or $N_{k+1}$. In total, we have that $B$ must completely contain every clique except $V_i$, implying $A \subseteq V_i$. In this case, however, we have that $B'$ contains $V_{i-1}$ and $V_{i+1}$, a contradiction to $B' \subseteq V_i$.
         Thus, $B \subseteq V_i$ holds. 
     \end{claimproof}
     In particular, this implies that for any split $(A,B)$ of $G_n$ we have $\min\!\set{\abs{A}, \abs{B}} \leq n$.
     By \cref{lem:balancedcut} there is a node $a \in V(T)$ such that $\abs{V_a} \geq 2n$ and $\abs{\overline{V_a}} \geq 2n$.
     We use the notation $A := V_a$ and $B := \overline{V_a}$.
     By the above claim, the cut $(A,B)$ is not a split.
     We show that a maximum matching for $(A,B)$ has size at least $\frac{n}{6}$.
     
     If we assume that there is a clique $V_i$ such that both $A$ and $B$ contain at least $\frac{n}{6}$ vertices of $V_i$, then these vertices induce a complete bipartite subgraph in the cut, which contains a matching of size $\frac{n}{6}$.
     So we assume that there is no clique $V_i$ such that both $A$ and $B$ contain at least $\frac{n}{6}$ vertices of $V_i$, i.e., for any clique $V_i$ it holds that one of $A$ or $B$ contains less than $\frac{n}{6}$ vertices of $V_i$ and the other set contains at least $\frac{5n}{6}$ vertices of $V_i$.
     In particular, as $\abs{A} \geq 2n$, there must be at least two cliques such that more than $\frac{5n}{6}$ of their vertices are contained in $A$, analogously for $B$.
     As the cliques of $G_n$ form a cycle, there must be two cliques $V_i$ and $V_{i+1}$ such that $A$ contains more than $\frac{5n}{6}$ vertices of $V_i$ and $B$ contains more than $\frac{5n}{6}$ vertices of $V_{i+1}$. 
     These vertices induce a complete bipartite subgraph in the cut and, thus, a matching of size at least $\frac{5n}{6}$.
\end{proof}

Combining this result with \cref{obsv:ecc2} we obtain the following.

\begin{theorem}
    (Linear) sm-width is neither $\simul$-bounded nor $\simulnonempty$-bounded.
\end{theorem}

\subsection{Tree- and Path-Decompositions} \label{subsec:tree}

Tree- and path-decompositions are well known and used to define the parameters \emph{treewidth} and \emph{pathwidth}.

\begin{definition} \label{def:treedecomp}
    A \emph{tree-decomposition} of a graph $G$ is a pair $\cT = \br{T, \set{X_t}_{t \in V(T)}}$, where $T$ is a tree and $X_t \subseteq V(G)$ for all $t \in V(T)$ with the following properties:
    \begin{enumerate}
        \item $\bigcup_{t \in V(T)} X_t = V(G)$,
        \item for every edge $uv \in E(G)$ there exists a node $t \in V(T)$ such that $u,v \in X_t$, and
        \item for $t, t^\prime, t^{\prime\prime} \in V(T)$, if $t^\prime$ is on the path between $t$ and $t^{\prime\prime}$, then $X_t \cap X_{t^{\prime\prime}} \subseteq X_{t^\prime}$.
    \end{enumerate}
    The sets $X_t$ for $t \in V(T)$ are called \emph{bags}.
    A \emph{path-decomposition} of a graph $G$ is a tree decomposition $\cT = \br{P, \set{X_t}_{t \in V(P)}}$, where $P$ is a path.
\end{definition}

Recently, new parameters using tree- and path-decomposition have been introduced~\cite{DallardMS24treeindependence,DourisboureG07,LimaMMORS24treedecompositions}. 
We give a general definition on how to obtain a parameter from such a decomposition.

\begin{definition} \label{def:treewidth}
    Let $f = (f_G)_{G \in \bbG}$ be a family of functions where $f_G : \cP(V(G)) \to \bbR$ is a set function on $V(G)$ for $G \in \bbG$.
    Furthermore, let $G$ be a graph and let $\cT = \br{T, \set{X_t}_{t \in V(T)}}$ be a tree-decomposition of $G$.
    Then the maximum value of $f_G(X_t)$ over all bags $X_t$ is called the \emph{$f_G$-width} of $\cT$ and the minimum $f_G$-width over all tree-decompositions of $G$ is called the \emph{$f$-treewidth} of $G$. 
    The \emph{$f$-pathwidth} of a graph is defined analogously for path-decompositions.
\end{definition}

Regular tree- and pathwidth are defined using this definition with the functions $f_G(A) := \abs{A}-1$ for any graph $G$.

We will not strictly adhere to the naming scheme of the definition above, as there are parameters defined in this way that have established names differing from this scheme.
We now give a result similar to \cref{lem:bintree} for tree- and path-decompositions.

\begin{lemma} \label{lem:treewidth1}
    Consider a family of functions $f = (f_G)_{G \in \bbG}$.
    If there is a non-decreasing function $g : \N \times \bbR \to \bbR$ such that for every graph $G$ with some $d$-simultaneous $\cC$-representation $(H,L)$ and every vertex set $A \subseteq V(G)$ it holds that $f_G(A) \leq g\!\br{d,f_H(A)}$, then the $f$-treewidth ($f$-pathwidth) is $\simul$-bounded.
\end{lemma}

\begin{proof}
    Consider some hereditary completable graph class $\cC$ such that $f$-treewidth ($f$-pathwidth) is bounded by $k$ on graphs in $\cC$.
    Let $G$ be any graph with $d = \simul[\cC][G]$ and $(H,L)$ be a $d$-simultaneous $\cC$-representation of $G$.
    Let $\cT = \br{T, \set{X_t}_{t \in V(T)}}$ be a tree-decomposition (path-decomposition) of $H$ with $f_H$-width $\leq k$.
    Then for every bag $X_t$ we have $f_G(X_t) \leq g\!\br{d, f_H(X_t)} \leq g(d,k)$.
    Thus, the $f$-treewidth ($f$-pathwidth) of $G$ is bounded by $g(d,k)$.
\end{proof}

We apply this theorem to \emph{induced matching pathwidth} and \emph{induced matching treewidth}~\cite{LimaMMORS24treedecompositions}.

\begin{definition}[Induced matching treewidth]
    Let $G$ be a graph.
    We define the function $\mu_G : \cP(V(G)) \to \N$ with $\mu_G(A)$ for $A \subseteq V(G)$ being the size of a maximum induced matching in $G[A]$.
    Using \cref{def:treewidth} with $f = (\mu_G)_{G \in \bbG}$ we define the \emph{induced matching treewidth} (\emph{induced matching pathwidth}) of a graph $G$ as the $f$-treewidth ($f$-pathwidth) of $G$.
\end{definition}

The fact that the functions $(\mu_G)_{G \in \bbG}$ satisfy the condition of \cref{lem:treewidth1} can be shown using similar arguments to those used in the proof of \cref{thm:mim}.

\begin{theorem}
    Let $d \in \N$, $G$ be a graph and $(H,L)$ be a $d$-simultaneous $\cC$-representation of $G$ for some hereditary completable graph class $\cC$.
    Then $\mu_G(A) \leq d\cdot \mu_H(A)$ for any vertex set $A \subseteq V(G)$.
\end{theorem}

The next parameters we consider are \emph{path independence number}~\cite{beisegel2024simultaneousinterval} and \emph{tree independence number}~\cite{DallardMS24treeindependence}.

\begin{definition}[Tree independence number] \label{def:treealpha}
    Let $G$ be a graph.
    We define the function $\alpha_G : \cP(V(G)) \to \N$ with $\alpha_G(A)$ for $A \subseteq V(G)$ being the size of a maximum independent set in $G[A]$.
    Using \cref{def:treewidth} with $f = (\alpha_G)_{G \in \bbG}$ we define the \emph{tree independence number} (\emph{path independence number}) of a graph $G$ as the $f$-treewidth ($f$-pathwidth) of $G$.
\end{definition}

The condition in \cref{lem:treewidth1} fails for the parameters these parameters by \cref{thm:indsimulbounded}.
However, $\simul$-boundedness of these parameters can still be shown using the following.

\begin{lemma} \label{lem:treewidth2}
  Consider a family of functions $f = (f_G)_{G \in \bbG}$.
  Assume there exist a non-decreasing function $g : \N \times \bbR \to \bbR$ and an $\ell \in \N$ such that for every graph $G$ it holds that 
  \begin{itemize}
      \item the function $f_G$ is non-decreasing, i.e., $f_G(A) \leq f_G(B)$ if $A \subseteq B$,
      \item $f_G$ is bounded by $\ell$ on singleton sets, and
      \item for every $d$-simultaneous $\cC$-representation $(H,L)$ of $G$ with non-empty label sets and every vertex set $A \subseteq V(G)$, it holds that $f_G(A) \leq g\!\br{d, f_H(A)}$.
  \end{itemize}
  Then, the $f$-treewidth ($f$-pathwidth) is $\simul$-bounded. 
\end{lemma}

\begin{proof}
    Let $\cC$ be a hereditary completable graph class. Let $G$ be some graph with $d = \simul[\cC][G]$ and let $(H,L)$ be a $d$-simultaneous $\cC$-representation of $G$.
    Furthermore, let $\cT = \br{T, \set{X_t}_{t \in V(T)}}$ be a tree-decomposition  of $H$ with $f_H$-width $\leq k$.
    If the vertices have only non-empty label sets, then $f_G(X_t) \leq g\!\br{d, f_H(X_t)} \leq g(d,k)$ for every bag $X_t$.
    If there are vertices with empty label sets, then these can be placed in their own bags associated with leaves in the tree-decomposition, as they are isolated vertices in $G$.
    Removing vertices from a bag does not increase the $f_G$-width since $f_G$ is non-decreasing.
    Thus, $f_G(X_t)$ is bounded by $\ell$ for bags $X_t$ consisting of a singleton vertex and by $g(d,k)$ for all other bags, using the same argument as above.
    In total, we have that the $f$-treewidth of $G$ is bounded by $\max\{g(d,k), \ell\}$. The arguments work similarly for path-decompositions.
\end{proof}

\begin{lemma}
    Let $d \in \N$, $G$ be a graph and $(H,L)$ be a $d$-simultaneous $\cC$-representation of $G$ with non-empty label sets for some hereditary completable graph class $\cC$.
    Then the following hold:
    \begin{itemize}
        \item $\alpha_G$ is non-decreasing, i.e., $A \subseteq B$ implies that $\alpha_G(A) \leq \alpha_G(B)$ for $A, B \subseteq V(G)$,
        \item $\alpha_G$ is bounded by $1$ on singleton sets,
        \item $\alpha_G(A) \leq d \cdot \alpha_H(A)$ for any vertex set $A \subseteq V(G)$.
    \end{itemize}
    Thus, tree independence number and path independence number are $\simul$-bounded.
\end{lemma}

\begin{proof}
    The first two properties are obviously true.
    The third property follows directly from the proof of \cref{thm:indsimulbounded}.
\end{proof}

Finally, we consider graph parameters defined via tree-decompositions that are not $\simul$-bounded, namely tree-length and path-length~\cite{DourisboureG07,DraganKL17}.
For a graph $G$ and a set $A \subseteq V(G)$, define $diam_G(A) := \max_{u,v \in A} d_G(u,v)$, where $d_G(u,v)$ denotes the distance between $u$ and $v$ in $G$.
Using \cref{def:treewidth} with $f = (diam_G)_{G \in \bbG}$, we define the \emph{tree-length} and \emph{path-length}.
Both parameters are bounded on interval graphs~\cite{DraganKL17} but unbounded on cycles~\cite{DourisboureG07,DraganKL17}, while the simultaneous interval number is bounded on cycles~\cite{beisegel2024simultaneousinterval}. This implies the following.

\begin{lemma}
    Tree-length and path-length are neither $\simul$-bounded nor $\simulnonempty$-bounded.
\end{lemma}

\subsection{Parameters Related to Modules} \label{subsec:modules}

We again consider a family of parameters which is defined using graph classes, namely the so called $\cG$-modular cardinalities, which were introduced by Lafond and Luo~\cite{lafond2023parameterized} and which generalize the parameters neighborhood diversity~\cite{lampis2012meta} and iterated type partition~\cite{cordasco2024iterated}.

A \emph{module} in a graph $G$ is a vertex set $M$ such that for every vertex $v \in V(G) \setminus M$ it holds that either $M \subseteq N(v)$ or $M \cap N(v) = \emptyset$. We say that a module is \emph{trivial} if it is either a singleton set or the set $V(G)$.
For a graph class $\cG$ and a graph $G$ the \emph{$\cG$-modular cardinality} of $G$, denoted $\cG\Gmc(G)$, is defined as the cardinality of a minimum partition of $V(G)$ such that each set of the partition is a module in $G$ and induces a graph of $\cG$.
Such a partition is called a \emph{$\cG$-modular partition}.
For a graph $G$ with a modular partition $P$, we denote by $G/P$ the \emph{quotient graph of $G$ with respect to $P$}, that is, the graph obtained from $G$ by contracting every module of $P$ to a single vertex. See \cref{fig:cographmodules} for an example.

When taking $\cG$ to be the graph class containing all edgeless and complete graphs, the $\cG$-modular cardinality equals neighborhood diversity.
Moreover, when taking $\cG$ to be the class of cographs, the $\cG$-modular cardinality equals the parameter iterated type partition~\cite{cordasco2024iterated}.

\begin{figure}
    \centering
        \begin{tikzpicture}
        \tikzstyle{dashedcircle}=[circle, draw, dashed, inner sep=0pt, minimum width=2.5cm]
        \tikzstyle{vertex}=[draw,circle]
        \def\a{.4} 
        \begin{scope}[xshift=-2cm,scale=0.7, every node/.append style={transform shape}]
        \node[dashedcircle] (M) at (0,0) {};
        \node[vertex] (m1) at (-0.6125,0.375) {};
        \node[vertex] (m2) at (-0.6125,-0.375) {};
        \node[vertex] (m3) at (0.375,0) {};
        \node[vertex] (m4) at (0.375,0.75) {};
        \node[vertex] (m5) at (0.375,-0.75) {};
        \draw[thick] (m1) -- (m3);
        \draw[thick] (m1) -- (m4);
        \draw[thick] (m1) -- (m5);
        \draw[thick] (m2) -- (m3);
        \draw[thick] (m2) -- (m4);
        \draw[thick] (m2) -- (m5);
        \node[dashedcircle] (A) at (-3,0.75) {};
        \node[vertex] (a1) at (-3,0.75) {};
        \node[vertex] (a2) at (-3.5,0.25) {};
        \node[vertex] (a3) at (-3.5,1.25) {};
        \node[vertex] (a4) at (-2.5,0.25) {};
        \node[vertex] (a5) at (-2.5,1.25) {};
        \node[dashedcircle] (B) at (-3,-2.5) {};
        \node[vertex] (b1) at (-3,-2.5) {};
        \node[vertex] (b2) at (-3.75,-2.5) {};
        \node[vertex] (b3) at (-2.5,-1.875) {};
        \node[vertex] (b4) at (-2.5,-3.125) {};
        \draw[thick] (b2) -- (b3) -- (b4) -- (b2);
        \draw[thick] (b1) -- (b2);
        \draw[thick] (b1) -- (b3);
        \draw[thick] (b1) -- (b4);
        \node[dashedcircle] (C) at (0,-3.5) {};
        \node[vertex] (c1) at (0,-3.5) {};
        \node[vertex] (c2) at (0.5,-3) {};
        \node[vertex] (c3) at (0.5,-4) {};
        \node[vertex] (c4) at (-0.5,-4) {};
        \node[vertex] (c5) at (-0.5,-3) {};
        \draw[thick] (c2) -- (c3) -- (c4) -- (c5) -- (c2) -- (c1) -- (c3);
        \draw[thick] (c4) -- (c1) -- (c5);
        \node[dashedcircle] (D) at (3,-2.5) {};
        \node[vertex] (d1) at (3,-2.5) {};
        \node[vertex] (d2) at (2.25,-2.5) {};
        \node[vertex] (d3) at (3.5,-1.875) {};
        \node[vertex] (d4) at (3.5,-3.125) {};
        \draw[thick] (d2) -- (d1) -- (d3);
        \draw[thick] (d1) -- (d4);
        \node[dashedcircle] (E) at (3,0.75) {};
        \node[vertex] (e1) at (2.875,1.25) {};
        \node[vertex] (e2) at (2.875,0.25) {};
        \node[vertex] (e3) at (2.125,0.75) {};
        \node[vertex] (e4) at (3.75,1.25) {};
        \node[vertex] (e5) at (3.75,0.25) {};
        \draw[thick] (e1) -- (e2) -- (e3) -- (e1);
        \draw[thick] (e4) -- (e5);
        \draw[ultra thick] (A) -- (B) -- (C) -- (D) -- (E);
        \draw[ultra thick] (M) -- (A);
        \draw[ultra thick] (M) -- (B);
        \draw[ultra thick] (M) -- (C);
        \draw[ultra thick] (M) -- (D);
        \draw[ultra thick] (M) -- (E);
        \end{scope}
        \begin{scope}[xshift=2.5cm,yshift=-1.25cm]
            \node[vertex] (a) at (0,0) {};
            \node[vertex] (b) at (1,0) {};
            \node[vertex] (c) at (2,0) {};
            \node[vertex] (d) at (3,0) {};
            \node[vertex] (e) at (4,0) {};
            \node[vertex] (m) at (2,1.25) {};
            \draw[thick] (m) -- (a) -- (b) -- (c) -- (d) -- (e) -- (m) -- (b);
            \draw[thick] (c) -- (m) -- (d);
            \node at (2,-1) {$G/P$};
        \end{scope}
        \node at (-2,-3.75) {$G$};
    \end{tikzpicture}
    \caption{A graph $G$ with $\mathsf{cograph}$-modular cardinality six and its quotient graph $G/P$. A thick edge between two dashed circles indicates a complete bipartite graph between the nodes inside the circles. Each dashed circle induces a cograph module in $G$.}
    \label{fig:cographmodules}
\end{figure}
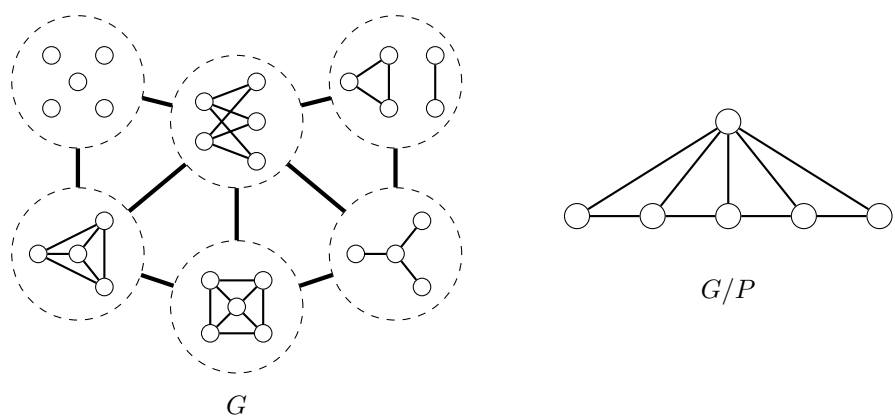

We make use of \cref{lem:maxlower} to show that many $\cG$-modular cardinalites are not $\simul$-bounded.

\begin{lemma} \label{lem:Gmclower}
    Let $\cG$ be a hereditary graph class that is closed under disjoint union.
    Then for any two graphs $G_1, G_2$ it holds that 
    \[ \cG\Gmc(G_1) + \cG\Gmc(G_2) -1 \leq \cG\Gmc(G_1 \cup G_2) \leq \cG\Gmc(G_1) + \cG\Gmc(G_2). \]
    In particular, equality for the lower bound holds if and only if both $G_1$ and $G_2$ contain a connected component that induces a graph of $\cG$.
\end{lemma}

\begin{proof}
    Let $G_1, G_2$ be any two graphs.
    The upper bounds follows trivially by taking the union of the minimum $\cG$-modular partitions of $G_1$ and $G_2$.
    We therefore focus on the lower bound.
    
    If $\cG\Gmc(G_1) = 1$ or $\cG\Gmc(G_2) = 1$, the statement holds trivially. 
    Thus, we assume that $\cG\Gmc(G_1), \cG\Gmc(G_2) \geq 2$.
    Let $M_1, \dots, M_k$ be a minimum $\cG$-modular partition of $G := G_1 \cup G_2$.
    It suffices to show that there is at most one module $M_i$ in the partition that intersects both $V(G_1)$ and $V(G_2)$, which we call a \emph{mixed module}.
    Assume that there is such a mixed module $M_i$.
    As there are no edges between $V(G_1)$ and $V(G_2)$ in $G$, it must hold that there are no edges between $M_i \cap V(G_1)$ and $V(G_1) \setminus M_i$, similarly for $M_i \cap V(G_2)$ and $V(G_2) \setminus M_i$.
    Thus, no vertex outside of $M_i$ is adjacent to $M_i$ in $G$.
    Now assume that there is another mixed module $M_j$. Thus, both $M_i$ and $M_j$ are isolated in $G$ and therefore $M_i \cup M_j$ is a module.
    Furthermore, as $\cG$ is closed under disjoint union, $G[M_i \cup M_j] \in \cG$ holds, which is a contradiction to the fact that $M_1, \dots, M_k$ is a minimum $\cG$-modular partition.

    Notice that a mixed module in any $\cG$-modular partition of $G$ consists of connected components of $G_1$ and $G_2$ which induce a graph of $\cG$. This implies that whenever there is a mixed module $G_1$ and $G_2$, both contain a connected component which induces a graph of $\cG$.

    Conversely, if both $G_1$ and $G_2$ contain such a connected component, then these connected components are fully contained in some module of minimum $\cG$-modular partitions of $G_1$ and $G_2$. Then the union of these modules is again a $\cG$-module in $G_1 \cup G_2$ since $\cG$ is closed under disjoint union. Combining this module with the remaining modules of minimum $\cG$-modular partitions of $G_1$ and $G_2$ results in a $\cG$-modular partition of $G_1 \cup G_2$ of size $\cG\Gmc(G_1) + \cG\Gmc(G_2) -1$.
\end{proof}

Because $\cG$-modular cardinality is trivially bounded on $\cG$ we obtain that, for $\cG$ fulfilling the conditions of \cref{lem:maxlower,lem:Gmclower}, $\cG$-modular cardinality is not $\simul$-bounded.

\begin{theorem}
    Let $\cG$ be a non-trivial hereditary completable graph class that is closed under disjoint union. Then $\cG$-modular cardinality is neither $\simul$-bounded nor $\simulnonempty$-bounded.
\end{theorem}

This result captures many relevant graph classes, in particular the class of cographs, showing that iterated type partition is not $\simul$-bounded.

Assuming that $\cG$ is hereditary, the other assumptions in the theorem are somewhat necessary: If one considers the trivial graph class $\cG$ of all graphs, then $\cG\Gmc(G) = 1$ holds for all graphs $G$ and the parameter is obviously $\simul$-bounded. If $\cG$ is not completable, then $\cG\Gmc$ is unbounded on complete graphs and, thus, trivially $\simul$-bounded. We can not omit the condition of $\cG$ being closed under disjoint union. This can be seen using neighborhood diversity, which is $\simul$-bounded due to \cref{thm:fo}. The defining class of complete and edgeless graphs fulfills the other conditions in the theorem but is not closed under disjoint union.

Another parameter which is stronlgy related to modules in a graph but not captured by $\cG$-modular cardinalities is the \emph{modular-width}~\cite{gajarsky2013modular}.

\begin{definition}[Modular-width]
    We consider algebraic expressions on graphs using the following four operations:
    \begin{enumerate}
        \item the creation of an isolated vertex,
        \item the disjoint union of two graphs,
        \item the complete join of two graphs, i.e., the disjoint union of the graph where every possible edge between the graphs is added,
        \item the substitution operation with respect to some graph $G$ with vertices $v_1, \dots, v_n$, i.e., given graphs $G_1, \dots, G_n$, each vertex $v_i$ in $G$ is substituted by the graph $G_i$.
    \end{enumerate}
    The \emph{width} of an algebraic expression using the above operations is the maximum number of operands used by the substitution operation.
    The \emph{modular-width} of a graph is the least integer $k$ such that the graph can be obtained by such an algebraic expression of width $k$.
\end{definition}

The operations used in the definition of modular-width correspond to modules in a graph in the sense that the operands are modules in the constructed graph.
Note that, if a graph $G$ only has trivial modules, then the modular-width of $G$ is equal to $\abs{V(G)}$ since using the operations 2 and 3 would lead to the existence of at least one non-trivial module. Such a graph is called \emph{prime}.

\begin{theorem} \label{thm:modw}
    Modular-width is neither $\simul$-bounded nor $\simulnonempty$-bounded.
\end{theorem}

\begin{proof}
    Consider the graph class $\cocluster$, consisting of complements of disjoint unions of cliques that has modular-width zero.
    For $d,k \in \N$ with $d \geq 2$, we define a family of graphs $G_{d,k}$ using simultaneous representations:
    Let $H$ be the co-cluster with $V(H) = \{v_A^{(i)} \mid \emptyset \neq A \subseteq \set{1,\dots,d}, 1 \leq i \leq k\}$ and $E(H) = \{v_A^{(i)} v_B^{(j)} \mid {i \neq j}\}$.
    That is, $H$ is the complete $k$-partite graph where the sets of the form $V_i = \{v_A^{(i)} \mid \emptyset \neq A \subseteq \set{1,\dots,d}\}$ are the independent sets of the $k$-partition.
    Then, we define $L$ by choosing $L(v_A^{(i)}) = A$ for $\emptyset \neq A \subseteq \set{1,\dots,d}$ and $1 \leq i \leq k$.
    Let $G_{d,k}$ be the graph obtained by the simultaneous representation $(H,L)$ (see \cref{fig:modw} for an example).
    By construction, $G_{d,k}$ has simultaneous $\cocluster$-number at most~$d$.
    We claim that $G_{d,k}$ only has trivial modules, which implies the statement.

    Let $G_{d,k}$ with some fixed $d \geq 2$ and some $k \in \N$ be given.
    Assume that there exists some non-trivial module $M$ of $G_{d,k}$.
    Thus, $M$ contains at least two vertices $u$ and $v$.
    We first consider the case that the two vertices are from different independent sets of the graph $H$.
    Then the neighborhood of $u$ restricted to the independent set of $v$ has to be contained in $M$, as otherwise these vertices would be adjacent to $u$ but not to $v$.
    Similarly, the neighborhood of $v$ restricted to the independent set of $u$ has to be contained in $M$.
    By iterating this argument, both independent sets of $u$ and $v$ have to be contained fully in $M$.
    Let $V_\ell$ be the independent set that contains $u$ and consider the vertices $v_{\set{1}}^{(\ell)}$ and $v_{\set{1,\dots,d}}^{(\ell)}$ which are in $M$.
    The set $N\br{v_{\set{1}}^{(\ell)}} \triangle N\br{v_{\set{1, \dots, d}}^{(\ell)}}$ must contain vertices of every independent set other than $V_\ell$.
    Again, all these independent sets have to be fully contained in $M$, which implies $M = V(G_{d,k})$, a contradiction to $M$ being non-trivial.
    
    Now consider the case that $u$ and $v$ are from the same independent set of $H$.
    Then again, the symmetric difference of their neighborhoods has to be contained in $M$, implying that at least one vertex of another independent set of $H$ is in $M$, which brings us back to the first case.
\end{proof}

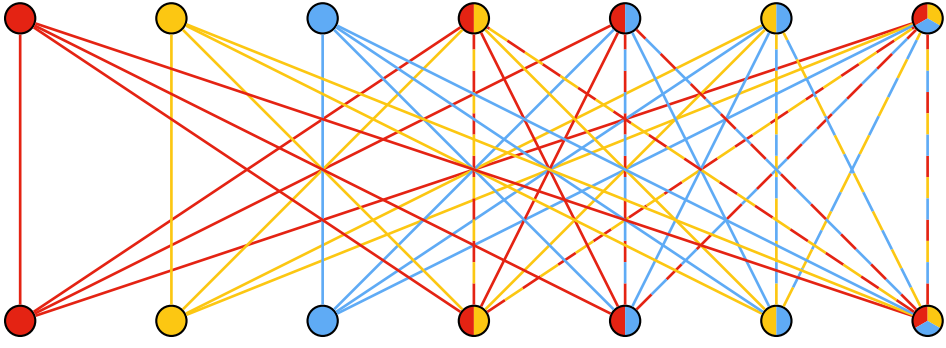
\begin{figure}
    \centering
        \begin{tikzpicture}
        \def\a{.4} 
        \def\d{8pt}
        \def\labelone{BTUred}
        \def\labeltwo{lipicsYellow}
        \def\labelthree{coolblue}
        \def\nolabel{lipicsGray}
        \tikzset{c/.style={circle,draw,thick,minimum size=\a cm}} 
        \tikzset{one color/.style={c,fill=#1}}  
        \tikzset{pics/two colors/.style args=
            {#1|#2|rotate=#3}{code={%
        \fill[#1,rotate=#3] (0,\a/2) arc(90:270:\a/2)--cycle;           
        \fill[#2,rotate=#3] (0,\a/2) arc(90:-90:\a/2)--cycle;
        \path (0,0) node[c] (-boundary) {};
        }}}
        \tikzset{pics/three colors/.style args=
            {#1|#2|#3|rotate=#4}{code={%
        \fill[#1,rotate=#4] (0,\a/2) arc(90:210:\a/2)--(0,0)--cycle;            
        \fill[#2,rotate=#4] (0,\a/2) arc(90:-30:\a/2)--(0,0)--cycle;
        \fill[#3,rotate=#4] (210:\a/2) arc(210:330:\a/2)--(0,0)--cycle;
        \path (0,0) node[c] (-boundary) {};
        }}}
        \path
        (0,0) node[one color=\labelone] (1a) {}
        (2,0) node[one color=\labeltwo] (2a) {}
        (4,0) node[one color=\labelthree] (3a) {}
        (0,4) node[one color=\labelone] (1b) {}
        (2,4) node[one color=\labeltwo] (2b) {}
        (4,4) node[one color=\labelthree] (3b) {}
        ;
        \path
        (6,0) pic (12a) {two colors={\labelone|\labeltwo|rotate=0}}
        (8,0) pic (13a) {two colors={\labelone|\labelthree|rotate=0}}
        (10,0) pic (23a) {two colors={\labeltwo|\labelthree|rotate=0}}
        (12,0) pic (123a) {three colors={\labelone|\labeltwo|\labelthree|rotate=0}}
        (6,4) pic (12b) {two colors={\labelone|\labeltwo|rotate=0}}
        (8,4) pic (13b) {two colors={\labelone|\labelthree|rotate=0}}
        (10,4) pic (23b) {two colors={\labeltwo|\labelthree|rotate=0}}
        (12,4) pic (123b) {three colors={\labelone|\labeltwo|\labelthree|rotate=0}}
        ;
		\draw[line width=1pt,\labelone] (1a) -- (1b);
        \draw[line width=1pt,\labelone] (1a) -- (12b-boundary);
        \draw[line width=1pt,\labelone] (1a) -- (13b-boundary);
        \draw[line width=1pt,\labelone] (1a) -- (123b-boundary);
        \draw[line width=1pt,\labeltwo] (2a) -- (2b);
        \draw[line width=1pt,\labeltwo] (2a) -- (12b-boundary);
        \draw[line width=1pt,\labeltwo] (2a) -- (23b-boundary);
        \draw[line width=1pt,\labeltwo] (2a) -- (123b-boundary);
        \draw[line width=1pt,\labelthree] (3a) -- (3b);
        \draw[line width=1pt,\labelthree] (3a) -- (13b-boundary);
        \draw[line width=1pt,\labelthree] (3a) -- (23b-boundary);
        \draw[line width=1pt,\labelthree] (3a) -- (123b-boundary);
        \draw[line width=1pt,\labelone] (12a-boundary) -- (1b);
        \draw[line width=1pt,\labeltwo] (12a-boundary) -- (2b);
        \draw[line width=1pt,\labelone] (12a-boundary) -- (13b-boundary);
        \draw[line width=1pt,\labeltwo] (12a-boundary) -- (23b-boundary);
        \draw[line width=1pt,\labelone,dash pattern= on \d off \d] (12a-boundary) -- (12b-boundary);
        \draw[line width=1pt,\labeltwo,dash pattern= on \d off \d,dash phase=\d] (12a-boundary) -- (12b-boundary);
        \draw[line width=1pt,\labelone,dash pattern= on \d off \d] (12a-boundary) -- (123b-boundary);
        \draw[line width=1pt,\labeltwo,dash pattern= on \d off \d,dash phase=\d] (12a-boundary) -- (123b-boundary);
        \draw[line width=1pt,\labelone] (13a-boundary) -- (1b);
        \draw[line width=1pt,\labelthree] (13a-boundary) -- (3b);
        \draw[line width=1pt,\labelone] (13a-boundary) -- (12b-boundary);
        \draw[line width=1pt,\labelthree] (13a-boundary) -- (23b-boundary);
        \draw[line width=1pt,\labelone,dash pattern= on \d off \d] (13a-boundary) -- (13b-boundary);
        \draw[line width=1pt,\labelthree,dash pattern= on \d off \d,dash phase=\d] (13a-boundary) -- (13b-boundary);
        \draw[line width=1pt,\labelone,dash pattern= on \d off \d] (13a-boundary) -- (123b-boundary);
        \draw[line width=1pt,\labelthree,dash pattern= on \d off \d,dash phase=\d] (13a-boundary) -- (123b-boundary);
        \draw[line width=1pt,\labeltwo] (23a-boundary) -- (2b);
        \draw[line width=1pt,\labelthree] (23a-boundary) -- (3b);
        \draw[line width=1pt,\labeltwo] (23a-boundary) -- (12b-boundary);
        \draw[line width=1pt,\labelthree] (23a-boundary) -- (13b-boundary);
        \draw[line width=1pt,\labeltwo,dash pattern= on \d off \d] (23a-boundary) -- (23b-boundary);
        \draw[line width=1pt,\labelthree,dash pattern= on \d off \d,dash phase=\d] (23a-boundary) -- (23b-boundary);
        \draw[line width=1pt,\labeltwo,dash pattern= on \d off \d] (23a-boundary) -- (123b-boundary);
        \draw[line width=1pt,\labelthree,dash pattern= on \d off \d,dash phase=\d] (23a-boundary) -- (123b-boundary);
        \draw[line width=1pt,\labelone] (123a-boundary) -- (1b);
        \draw[line width=1pt,\labeltwo] (123a-boundary) -- (2b);
        \draw[line width=1pt,\labelthree] (123a-boundary) -- (3b);
        \draw[line width=1pt,\labelone,dash pattern= on \d off \d] (123a-boundary) -- (12b-boundary);
        \draw[line width=1pt,\labeltwo,dash pattern= on \d off \d,dash phase=\d] (123a-boundary) -- (12b-boundary);
        \draw[line width=1pt,\labelone,dash pattern= on \d off \d] (123a-boundary) -- (13b-boundary);
        \draw[line width=1pt,\labelthree,dash pattern= on \d off \d,dash phase=\d] (123a-boundary) -- (13b-boundary);
        \draw[line width=1pt,\labeltwo,dash pattern= on \d off \d] (123a-boundary) -- (23b-boundary);
        \draw[line width=1pt,\labelthree,dash pattern= on \d off \d,dash phase=\d] (123a-boundary) -- (23b-boundary);
        \draw[line width=1pt,\labelone,dash pattern= on \d off 2*\d] (123a-boundary) -- (123b-boundary);
        \draw[line width=1pt,\labeltwo,dash pattern= on \d off 2*\d,dash phase=\d] (123a-boundary) -- (123b-boundary);
        \draw[line width=1pt,\labelthree,dash pattern= on \d off 2*\d,dash phase=2*\d] (123a-boundary) -- (123b-boundary);
    \end{tikzpicture}
    \caption{The graph $G_{3,2}$ as given in the proof of \cref{thm:modw}.}
    \label{fig:modw}
\end{figure}

We note that the parameters iterated type partition, modular-width, precedence thinness, and sm-width are sandwiched between $\simul$-bounded parameters while not being $\simul$-bounded themselves (see \cref{fig:diagram}). This demonstrates that $\simul$-boundedness of a parameter can not be deduced by comparison with other $\simul$-bounded parameters.

\section{Upper Bounds} \label{sec:upper}

While the previous section focused on lower bounds for simultaneous $\C$-numbers, we will now investigate which graph parameters can form upper bounds on these numbers. For lower bounds, we preferred to show $\simul$-boundedness instead of $\simulnonempty$-boundedness since that is a strictly stronger property.
In the case of upper bounds the reverse is true, i.e., a parameter being an upper bound on $\simulnonempty[\cC]$ is a stronger property than being an upper bound on $\simul[\cC]$.
This can be seen using the \emph{edge clique cover number}, which was shown to upper bound $\simul[\cC]$ for every hereditary completable $\cC$ by Beisegel et al.~\cite{beisegel2024simultaneousinterval}, while $\simulnonempty[\mathsf{complete}]$ is unbounded in the edge clique cover number, as can be seen via edgeless graphs.

\subsection{Subquadratic Upper Bounds}

We have seen in \cref{sec:simrep} that simultaneous $\cC$-numbers can be at most quadratic in the number of vertices.
Beisegel et al.~\cite{beisegel2024simultaneousinterval} showed that this bound is tight for the simultaneous interval number up to a constant factor.
Using a result of Erdős, Gimbel, and Kratsch~\cite{erdös1991cochromatic} on $\cP$-chromatic numbers, we can show that this is does not hold for all graph classes $\C$.

\begin{theorem} \label{thm:on^2}
    Let $G$ be a graph on $n$ vertices and $\cC$ be a hereditary completable graph class that contains all edgeless graphs and that is closed under the join operation. 
    Then $\simulnonempty[\cC][G] \in o\!\br{n^2}$.
    In particular, there is a constant $c > 0$ such that $\simulnonempty[\cC][G] \leq c \frac{n^2}{\log n}$.
\end{theorem}

To prove this theorem, we need the following three results.

\begin{restatable}{lemma}{ineq} \label{lem:partition}
    Let $V$ be a set with $n$ elements and $X_1, \dots, X_k$ be a partition of $V$.
    Furthermore, let $\abs{X_1} \leq \abs{X_2} \leq \dots \leq \abs{X_k}$.
    Then $\sum_{i=1}^{k-1} \br{(k-i) \abs{X_i}} \leq n \br{k - \ln k}$.
\end{restatable}

The proof of this lemma is given in Appendix \ref{sec:proofs}.

\begin{lemma} \label{lem:chromaticbound}
    Let $\mathcal{C}$ be a hereditary completable graph class that is closed under the join operation and let $G$ be a graph on $n$ vertices with $k = \chi_\cC(G)$.
    Then $\simul[\cC][G] \leq \simulnonempty[\cC][G] \leq n \br{k - \ln k} + k$.
\end{lemma}

\begin{proof}
    Let $X_1, \dots, X_k$ be a $\mathcal{C}$-coloring of $G$ and $ H := G[X_1] + G[X_2] + G[X_3] + \dots + G[X_k] \in \mathcal{C}$ be the join of all those induced subgraphs.
    W.l.o.g. let $\abs{X_1} \leq \abs{X_2} \leq \dots \leq \abs{X_k}$ and $X_i = \set{x_1^{(i)}, \dots, x_{\abs{X_i}}^{(i)}}$.
    We introduce the labels
    \[ \alpha_\ell^{(i,j)}, \ i \in \set{1,2,\dots,k-1}, \ j \in \set{i+1,i+2,\dots,k}, \ \ell \in \set{1,2,\dots,\abs{X_i}} \]
    and
    \[ \beta^{(i)}, \ i \in \set{1,2,\dots,k}. \]
    The number of labels is $k + \sum_{i = 1}^{k-1} \br{(k-i) \abs{X_i}} \leq k + n \br{k - \ln k}$ by \cref{lem:partition}.
    The idea here is that the labels $\beta^{(i)}$ model the edges inside the sets $X_i$ while the labels $\alpha_\ell^{(i,j)}$ model the restriction of the neighborhood of the $\ell$-th vertex in $X_i$ to $X_j$.
    We define $L$ via
    \[ L\!\br{x_\ell^{(i)}} := \set{\beta^{(i)}} \cup \setcond{\alpha_\ell^{(i,j)}}{j > i} \cup \setcond{\alpha_h^{(j,i)}}{j < i \text{ and } x_\ell^{(i)} x_h^{(j)} \in E(G)} \]
    for $i \in \set{1,2,\dots,k}$ and $\ell \in \set{1,2,\dots,\abs{X_i}}$ and claim that $(H,L)$ is a simultaneous $\mathcal{C}$-representation of $G$.

    Consider two vertices $x_\ell^{(i)}$ and $x_h^{(j)}$ which are adjacent in $G$.
    If $i = j$ holds, then $x_\ell^{(i)} x_h^{(j)} \in E(G[X_i]) \subseteq E(H)$ holds and both their label sets contain $\beta^{(i)}$.
    So, let $i \neq j$. Without loss of generality, we may assume that $i < j$.
    Then $H$ contains every edge between $X_i$ and $X_j$ and both of their label sets contain $\alpha_\ell^{(i,j)}$ by definition.

    Now consider two vertices $x_\ell^{(i)}$ and $x_h^{(j)}$ with $x_\ell^{(i)} x_h^{(j)} \in E(H)$ and $L(x_\ell^{(i)}) \cap L(x_h^{(j)}) \neq \emptyset$.
    If $i = j$ holds, then the edge ${x_\ell^{(i)} x_h^{(j)}}$ exists in $G$ since $H[X_i] = G[X_i]$.
    So, let $i \neq j$. Again, we may assue without loss of generality that $i < j$.
    Then the only label that can possibly be contained in the intersection of their label sets is $\alpha_\ell^{(i,j)}$, which is only contained in the label set of $x_h^{(j)}$ if the edge ${x_\ell^{(i)} x_h^{(j)}}$ exists in $G$.
\end{proof}

\begin{theorem}[Erdős, Gimbel, and Kratsch \cite{erdös1991cochromatic}] \label{thm:erdös}
    Let $\cP$ be the graph class containing all complete and all edgeless graphs.
    Then there exists a constant $c > 0$ such that for any graph $G$ on $n$ vertices we have $\chi_{\mathcal{P}}(G) \leq c \frac{n}{\log n}$.
    In particular, $c = 2 + o(1)$.
\end{theorem}

\cref{lem:chromaticbound,thm:erdös} imply the correctness of \cref{thm:on^2}.

\subsection{Modular Cardinalities}

We consider some special $\cG$-modular cardinalities again and show sufficient conditions on when they give upper bounds on simultaneous $\cC$-numbers with non-empty label sets. 
Two vertices $u,v$ of a graph $G$ are called \emph{false twins} if $N(u)=N(v)$ and \emph{true twins} if $N[u] = N[v]$.
We remind the reader that we denote by $G/P$ the quotient graph of $G$ with respect to a modular partition $P$.

\begin{theorem} \label{thm:quotientgraph}
    Let $\cC$ be a hereditary completable graph class and $G$ be a graph. It holds that $\simulnonempty[\cC][G] = \simulnonempty[\cC][G/P]$
    \begin{itemize}
        \item if $\cC$ is closed under addition of false twins and $P$ is an $\mathsf{edgeless}$-modular partition of $G$,~or
        \item if $\cC$ is closed under addition of true twins and $P$ is a $\mathsf{complete}$-modular partition of $G$, or
        \item if $\cC$ is closed under addition of true and false twins and $P$ is a $\mathsf{cograph}$-modular partition of $G$, or
        \item if $\cC$ is closed under vertex substitution and $P$ is a $\cC$-modular partition of $G$. 
    \end{itemize}
\end{theorem}

\begin{proof}
    Let $G$ be given.
    First note that for any modular partition $P$ of $G$ it holds that $G/P$ is an induced subgraph of $G$ and therefore $\simulnonempty[\cC][G/P] \leq \simulnonempty[\cC][G]$ holds.

    Now consider the first case, i.e., let $P$ be an $\mathsf{edgeless}$-modular partition of $G$.
    Consider a simultaneous $\cC$-representation $(H,L)$ of $G/P$ with non-empty label sets.
    Consider for a module $M \in P$ the corresponding vertex $v_M$ in $G/P$.
    To blow up $v_M$ to an edgeless module of size $\abs{M}$ it suffices to add $\abs{M}-1$ false twins of $v_M$ to $G/P$.
    This can be done without increasing the number of labels by adding false twins of $v_M$ to $H$ and choosing their label sets as $L(v_M)$.
    Thus, $\simulnonempty[\cC][G] \leq \simulnonempty[\cC][G/P]$.
    The case for complete modules and true twins follows analogously.
    
    For the case of cograph modules, we point out the following characterizations of cographs: A graph is a cograph if and only if it can be constructed from a single vertex, only using the operations of adding a true or false twin of an existing vertex.
    Using this characterization, this case follows analogously to the above, where true and false twins are added to create the cograph induced by $M$ (see \cref{fig:cographmodulerepr} for an example).

    For the last case, consider a $\cC$-modular partition $P$ of $G$ and the graph $G/P$ together with a simultaneous $\cC$-representation $(H,L)$ with non-empty label sets.
    A simultaneous $\cC$-representation of $G$ is obtained by substituting each vertex $v_M$ in $H$ with $G[M]$ and giving the new vertices in $M$ the label set $L(v_M)$.
\end{proof}

\begin{figure}
    \centering
        \begin{tikzpicture}[scale=0.85, every node/.append style={transform shape}]
        \tikzstyle{dashedcircle}=[circle, draw, dashed, inner sep=0pt, minimum width=2.5cm]
        \tikzstyle{vertex}=[draw,circle,minimum size=\a cm]
        \def\a{.4} 
        \def\d{8pt}
        \def\labelone{BTUred}
        \def\labeltwo{lipicsYellow}
        \def\labelthree{coolblue}
        \def\nolabel{lipicsGray}
        \tikzset{c/.style={circle,draw,thick,minimum size=\a cm}} 
        \tikzset{one color/.style={c,fill=#1}}  
        \tikzset{pics/two colors/.style args=
            {#1|#2|rotate=#3}{code={%
        \fill[#1,rotate=#3] (0,\a/2) arc(90:270:\a/2)--cycle;           
        \fill[#2,rotate=#3] (0,\a/2) arc(90:-90:\a/2)--cycle;
        \path (0,0) node[c] (-boundary) {};
        }}}
        \tikzset{pics/three colors/.style args=
            {#1|#2|#3|rotate=#4}{code={%
        \fill[#1,rotate=#4] (0,\a/2) arc(90:210:\a/2)--(0,0)--cycle;            
        \fill[#2,rotate=#4] (0,\a/2) arc(90:-30:\a/2)--(0,0)--cycle;
        \fill[#3,rotate=#4] (210:\a/2) arc(210:330:\a/2)--(0,0)--cycle;
        \path (0,0) node[c] (-boundary) {};
        }}}
        \node[dashedcircle] (M) at (0,0) {};
        \node[font=\bfseries\sffamily] at (0,1.75) {\textcolor{\labelone}{1},\textcolor{\labeltwo}{2}};
        \path
        (-0.6125,0.375) pic (m1) {two colors={\labelone|\labeltwo|rotate=-90}}
        (-0.6125,-0.375) pic (m2) {two colors={\labelone|\labeltwo|rotate=-90}}
        (0.375,0) pic (m3) {two colors={\labelone|\labeltwo|rotate=-90}}
        (0.375,0.75) pic (m4) {two colors={\labelone|\labeltwo|rotate=-90}}
        (0.375,-0.75) pic (m5) {two colors={\labelone|\labeltwo|rotate=-90}}
        ;
        \draw[line width=1.2pt,\labelone,dash pattern= on \d off \d] (m1-boundary) -- (m3-boundary);
        \draw[line width=1.2pt,\labelone,dash pattern= on \d off \d] (m1-boundary) -- (m4-boundary);
        \draw[line width=1.2pt,\labelone,dash pattern= on \d off \d] (m1-boundary) -- (m5-boundary);
        \draw[line width=1.2pt,\labelone,dash pattern= on \d off \d] (m2-boundary) -- (m3-boundary);
        \draw[line width=1.2pt,\labelone,dash pattern= on \d off \d] (m2-boundary) -- (m4-boundary);
        \draw[line width=1.2pt,\labelone,dash pattern= on \d off \d] (m2-boundary) -- (m5-boundary);
        \draw[line width=1.2pt,\labeltwo,dash pattern= on \d off \d,dash phase=\d] (m1-boundary) -- (m3-boundary);
        \draw[line width=1.2pt,\labeltwo,dash pattern= on \d off \d,dash phase=\d] (m1-boundary) -- (m4-boundary);
        \draw[line width=1.2pt,\labeltwo,dash pattern= on \d off \d,dash phase=\d] (m1-boundary) -- (m5-boundary);
        \draw[line width=1.2pt,\labeltwo,dash pattern= on \d off \d,dash phase=\d] (m2-boundary) -- (m3-boundary);
        \draw[line width=1.2pt,\labeltwo,dash pattern= on \d off \d,dash phase=\d] (m2-boundary) -- (m4-boundary);
        \draw[line width=1.2pt,\labeltwo,dash pattern= on \d off \d,dash phase=\d] (m2-boundary) -- (m5-boundary);
        \node[dashedcircle] (A) at (-3,0.75) {};
        \node[font=\bfseries\sffamily] at (-3,2.5) {\textcolor{\labelone}{1}};
        \node[one color=\labelone] (a1) at (-3,0.75) {};
        \node[one color=\labelone] (a2) at (-3.5,0.25) {};
        \node[one color=\labelone] (a3) at (-3.5,1.25) {};
        \node[one color=\labelone] (a4) at (-2.5,0.25) {};
        \node[one color=\labelone] (a5) at (-2.5,1.25) {};
        \node[dashedcircle] (B) at (-3,-2.5) {};
        \node[font=\bfseries\sffamily] at (-4,-1.25) {\textcolor{\labelone}{1},\textcolor{\labeltwo}{2}};
        \path 
        (-3,-2.5) pic (b1) {two colors={\labelone|\labeltwo|rotate=-45}}
        (-3.75,-2.5) pic (b2) {two colors={\labelone|\labeltwo|rotate=-45}}
        (-2.5,-1.875) pic (b3) {two colors={\labelone|\labeltwo|rotate=-45}}
        (-2.5,-3.125) pic (b4) {two colors={\labelone|\labeltwo|rotate=-45}}
        ;
        \def\d{4pt}
        \draw[line width=1.2pt,\labelone,dash pattern= on \d off \d] (b2-boundary) -- (b3-boundary) -- (b4-boundary) -- (b2-boundary);
        \draw[line width=1.2pt,\labelone,dash pattern= on \d off \d] (b1-boundary) -- (b2-boundary);
        \draw[line width=1.2pt,\labelone,dash pattern= on \d off \d] (b1-boundary) -- (b3-boundary);
        \draw[line width=1.2pt,\labelone,dash pattern= on \d off \d] (b1-boundary) -- (b4-boundary);
        \draw[line width=1.2pt,\labeltwo,dash pattern= on \d off \d,dash phase=\d] (b2-boundary) -- (b3-boundary) -- (b4-boundary) -- (b2-boundary);
        \draw[line width=1.2pt,\labeltwo,dash pattern= on \d off \d,dash phase=\d] (b1-boundary) -- (b2-boundary);
        \draw[line width=1.2pt,\labeltwo,dash pattern= on \d off \d,dash phase=\d] (b1-boundary) -- (b3-boundary);
        \draw[line width=1.2pt,\labeltwo,dash pattern= on \d off \d,dash phase=\d] (b1-boundary) -- (b4-boundary);
        \node[dashedcircle] (C) at (0,-3.5) {};
        \node[font=\bfseries\sffamily] at (0,-5.25) {\textcolor{\labeltwo}{2}};
        \node[one color=\labeltwo] (c1) at (0,-3.5) {};
        \node[one color=\labeltwo] (c2) at (0.5,-3) {};
        \node[one color=\labeltwo] (c3) at (0.5,-4) {};
        \node[one color=\labeltwo] (c4) at (-0.5,-4) {};
        \node[one color=\labeltwo] (c5) at (-0.5,-3) {};
        \draw[line width=1.2pt,\labeltwo] (c2) -- (c3) -- (c4) -- (c5) -- (c2) -- (c1) -- (c3);
        \draw[line width=1.2pt,\labeltwo] (c4) -- (c1) -- (c5);
        \node[dashedcircle] (D) at (3,-2.5) {};
        \node[font=\bfseries\sffamily] at (4,-1.25) {\textcolor{\labelone}{1},\textcolor{\labeltwo}{2}};
        \path 
        (3,-2.5) pic (d1) {two colors={\labelone|\labeltwo|rotate=-135}}
        (2.25,-2.5) pic (d2) {two colors={\labelone|\labeltwo|rotate=-135}}
        (3.5,-1.875) pic (d3) {two colors={\labelone|\labeltwo|rotate=-135}}
        (3.5,-3.125) pic (d4) {two colors={\labelone|\labeltwo|rotate=-135}}
        ;
        \draw[line width=1.2pt,\labelone,dash pattern= on \d off \d] (d2-boundary) -- (d1-boundary) -- (d3-boundary);
        \draw[line width=1.2pt,\labelone,dash pattern= on \d off \d] (d1-boundary) -- (d4-boundary);
        \draw[line width=1.2pt,\labeltwo,dash pattern= on \d off \d,dash phase=\d] (d2-boundary) -- (d1-boundary) -- (d3-boundary);
        \draw[line width=1.2pt,\labeltwo,dash pattern= on \d off \d,dash phase=\d] (d1-boundary) -- (d4-boundary);
        \node[dashedcircle] (E) at (3,0.75) {};
        \node[font=\bfseries\sffamily] at (3,2.5) {\textcolor{\labelone}{1}};
        \node[one color=\labelone] (e1) at (2.875,1.25) {};
        \node[one color=\labelone] (e2) at (2.875,0.25) {};
        \node[one color=\labelone] (e3) at (2.125,0.75) {};
        \node[one color=\labelone] (e4) at (3.75,1.25) {};
        \node[one color=\labelone] (e5) at (3.75,0.25) {};
        \draw[line width=1.2pt,\labelone] (e1) -- (e2) -- (e3) -- (e1);
        \draw[line width=1.2pt,\labelone] (e4) -- (e5);
        \def\d{8pt}
        \draw[line width=3pt,\labelone] (A) -- (B);
        \draw[line width=3pt,\labeltwo] (B) -- (C) -- (D);
        \draw[line width=3pt,\labelone] (D) -- (E);
        \draw[line width=3pt,\labelone] (M) -- (A);
        \draw[line width=3pt,\labelone,dash pattern= on \d off \d] (M) -- (B);
        \draw[line width=3pt,\labeltwo,dash pattern= on \d off \d,dash phase=\d] (M) -- (B);
        \draw[line width=3pt,\labeltwo] (M) -- (C);
        \draw[line width=3pt,\labelone,dash pattern= on \d off \d] (M) -- (D);
        \draw[line width=3pt,\labeltwo,dash pattern= on \d off \d,dash phase=\d] (M) -- (D);
        \draw[line width=3pt,\labelone] (M) -- (E);
        \draw[line width=1.5pt,\nolabel] (A) to[bend right = 90,looseness=1.5] (C);
        \draw[line width=1.5pt,\nolabel] (C) to[bend right = 90,looseness=1.5] (E);
    \end{tikzpicture}
    \caption{A $2$-simultaneous $\mathsf{cograph}$-representation of the graph $G$ of \cref{fig:cographmodules} using the $2$-simultaneous $\mathsf{cograph}$-representation of the graph $G/P$ given in \cref{fig:subfig2} and the construction of \cref{thm:quotientgraph}. A thick edge between two dashed circles again indicates a complete bipartite graph between the nodes inside the circles}
    \label{fig:cographmodulerepr}
\end{figure}
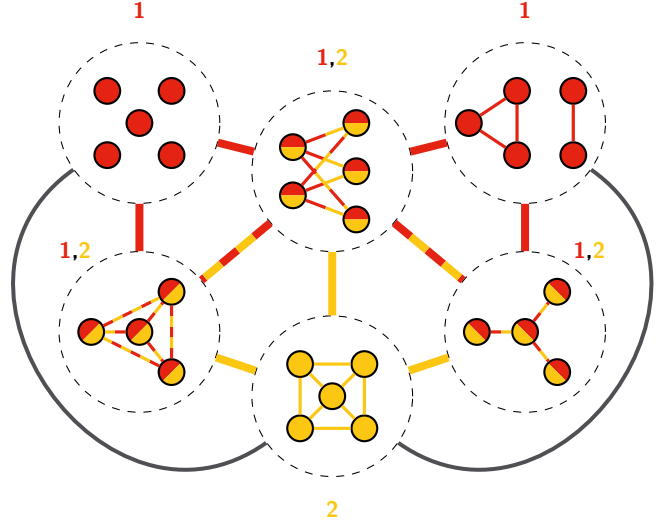

The $\cG$-modular cardinality of any graph $G$ is the minimum number of vertices of $G/P$ under all $\cG$-modular partitions $P$ of $G$. Using the fact that $\simulnonempty[\cC][G/P]$ is at most $|V(G/P)|^2$, we obtain the following upper bounds.

\begin{corollary}
    Let $\cC$ be a hereditary completable graph class.
    \begin{itemize}
        \item If $\cC$ is closed under addition of false twins, then $\mathsf{edgeless}$-modular cardinality upper bounds~$\simulnonempty[\cC]$. In particular, $\simulnonempty[\cC][G] \leq (\mathsf{edgeless}\Gmc(G))^2$ for any graph $G$.
        \item If $\cC$ is closed under addition of true twins, then $\mathsf{complete}$-modular cardinality upper bounds~$\simulnonempty[\cC]$. In particular, $\simulnonempty[\cC][G] \leq (\mathsf{complete}\Gmc(G))^2$ for any graph $G$.
        \item If $\cC$ is closed under addition of true and false twins, then $\mathsf{cograph}$-modular cardinality upper bounds $\simulnonempty[\cC]$. In particular, $\simulnonempty[\cC][G] \leq (\mathsf{cograph}\Gmc(G))^2$ for any graph $G$.
        \item If $\cC$ is closed under vertex substitution, then $\cC$-modular cardinality upper bounds $\simulnonempty[\cC]$.
        In particular, $\simulnonempty[\cC][G] \leq (\cC\Gmc(G))^2$ for any graph $G$.
    \end{itemize}
\end{corollary}

For simultaneous $\cC$-numbers with possibly empty label sets we have an even stronger statement for $\mathsf{complete}$-modular cardinality. It holds that $\mathsf{complete}$-modular cardinality upper bounds $\simul[\cC]$ for every hereditary completable $\cC$, because $\mathsf{complete}$-modular cardinality upper bounds the edge clique cover number. In particular, it holds that $\simul[\cC][G] \leq \ecc{G} \leq (\mathsf{complete}\Gmc(G))^2$ for every graph $G$.

\subsection{Treedepth, Bandwidth, Pathwidth, and Treewidth}

We now consider the parameters \emph{treedepth}, \emph{bandwidth}, \emph{treebandwidth}, \emph{pathwidth}, and \emph{treewidth} with the goal to characterize for which classes $\cC$ these parameters are upper bounds on the simultaneous $\cC$-number.
All these parameters share a common property: the value of treedepth, bandwidth, treebandwidth, pathwidth, or treewidth on a given graph $G$ is $\leq k$ if and only if there is a supergraph $H \in \cK$ of $G$ with clique number $\leq k+1$, where $\cK$ is the class of trivially perfect graphs~\cite{Golumbic78,nesetril2006treedepth}, proper interval graphs~\cite{kaplan1996bandwidth}, proper chordal graphs~\cite{jacob2025treebandwidth}, interval graphs~\cite{bodlaender1998treewidth}, or chordal graphs~\cite{bodlaender1998treewidth}, respectively.

The characterization we will obtain states that one of these parameters upper bounds a simultaneous $\cC$-number if and only if the simultaneous $\cC$-number is upper bounded by the clique number on the class corresponding to the parameter. In order to show this, we first have to bound the simultaneous $\cK$-number for each of the five $\cK$ in terms of their corresponding parameters.
The proofs for pathwidth and interval graphs, and treewidth and chordal graphs are given in \cite{beisegel2024simultaneousinterval,simchord}.
We will show the statements for the other classes now.
Note that all these bounds are quadratic. This is tight up to a constant, as can be seen by complete $3$-partite graphs. These graphs have quadratic simultaneous $\mathsf{chordal}$-number~\cite{simchord} in the number of vertices and therefore the same holds for the other four simultaneous $\cK$-numbers. However, their corresponding parameters are at most linear in the number of vertices on any graph. 

We start with the parameter treedepth.
A \emph{rooted forest} $\br{F, r_1, \dots, r_\ell}$ is the disjoint union of rooted trees $\br{T_1, r_1}, \dots, \br{T_k, r_\ell}$.
The \emph{height} of $\br{F, r_1, \dots, r_\ell}$ is the maximum height of any of the trees $\br{T_1, r_1}, \dots, \br{T_k, r_\ell}$.
For a rooted forest $\br{F, r_1, \dots, r_\ell}$, we define its \emph{closure} $\clos[F, r_1, \dots, r_\ell] := (V(F), E)$ with
$E := \setcond{xy}{x \text{ is an ancestor of } y \text{ in } F \text{ and } x \neq y}$.

\begin{definition}[Treedepth]
  The \emph{treedepth} of a graph $G$, denoted by $\td{G}$, is the minimum height of a rooted forest $\br{F, r_1, \dots, r_\ell}$ such that $G$ is a subgraph of its closure $\clos[F, r_1, \dots, r_\ell]$.
\end{definition}

The height of the closure of a rooted forest is exactly its clique number minus one, thus treedepth can be interpreted in terms of the minimum maximum clique size of some supergraph of a graph that is a closure of a rooted forest.
Closures of rooted forests are also called \emph{trivially perfect graphs}~\cite{Golumbic78}.

\begin{theorem} \label{thm:td}
    Let $\cC$ be the class of trivially perfect graphs.
    Then treedepth upper bounds the simultaneous $\cC$-number (with non-empty label sets).
    In particular, for any graph $G$ we have $\simul[\cC][G] \leq \binom{\tdbr{G} + 1}{2}$ and $\simulnonempty[\cC][G] \leq \binom{\tdbr{G} + 1}{2} + 1$.
\end{theorem}

\begin{proof}
    The closures of rooted forests are closed under addition of isolated vertices and, thus, by \cref{obsv:isolatedvertices}, $\simul[\cC]$ and $\simulnonempty[\cC]$ are equal on all graphs containing at least one edge.
    Hence, it suffices to show the bound for $\simul[\cC]$ as the bound for $\simulnonempty[\cC]$ follows directly.
    Let $G$ be a graph and $k = \td{G}$.
    Furthermore, let $\br{F, r_1, \dots, r_\ell}$ be a rooted forest with height $k$ such that $G$ is a subgraph of $\clos[F, r_1, \dots, r_\ell]$ and $V(G) = V(F)$.
    We choose $H = \clos[F, r_1, \dots, r_\ell]$ and introduce $\binom{k+1}{2}$ labels $\alpha_{ij} = \alpha_{ji}, \ 0 \leq i < j \leq k$.
    We define $L$ as follows:
    \[ L(v) := \setcond{\alpha_{ij}}{\genfrac{}{}{0pt}{}{v \text{ has height } i \text{ in } F \text{ and there exists some } u \in V(G)}{\text{with height } j \text{ in } F \text{ such that } uv \in E(G)}}. \]
    We claim that $(H, L)$ is a simultaneous $\cC$-representation of $G$.
    Let $uv \in E(G)$.
    Therefore, $u$ is either an ancestor or descendant of $v$ in $\br{F, r_1, \dots, r_\ell}$.
    Without loss of generality let $u$ be an ancestor of $v$ and $i < j$ be their heights.
    Since $H = \clos[F, r_1, \dots, r_\ell]$, we have that $uv \in E(H)$.
    Furthermore, by definition of $L$ we have $\alpha_{ij} = \alpha_{ji} \in L(u) \cap L(v)$.
    On the other hand, let $u,v \in V(G)$ with $uv \in E(H)$ and $L(u) \cap L(v) \neq \emptyset$.
    Thus, $u$ is either an ancestor or descendant of $v$ in $\br{F, r_1, \dots, r_k}$, w.l.o.g. let $u$ be an ancestor of $v$.
    Let $i$ be the height of $u$ and $j$ be the height of $v$.
    Since every label in $L(u)$ is of the form $\alpha_{ih} = \alpha_{hi}$ and every label in $L(v)$ is of the form $\alpha_{jh} = \alpha_{hj}$, the intersection $L(u) \cap L(v)$ can only contain the label $\alpha_{ij} = \alpha_{ji}$.
    Thus, $v$ is adjacent to some vertex of height $i$ in $G$ and, since $u$ is the unique ancestor of $v$ with height $i$, we have $uv \in E(G)$.
\end{proof}

We continue with bandwidth.

\begin{definition}[Bandwidth]
    Let $G$ be a graph and $f : V(G) \to \set{1,2,\dots,n}$ be a bijection.
    The \emph{bandwidth} of $f$ is given by $bw(f, G) = \max\setcond{\abs{f(u) - f(v)}}{uv \in E(G)}$. 
    The \emph{bandwidth} of $G$, denoted by $\bw{G}$, is then defined as the minimum bandwidth over all bijections $f : V(G) \to \set{1,2,\dots,n}$.
\end{definition}

Bandwidth is equivalent to \emph{proper pathwidth}, due to Kaplan and Shamir~\cite{kaplan1996bandwidth}.
Then, by adjusting the proof of Theorem 3.9 in \cite{beisegel2024simultaneousinterval} to use \emph{proper path-decompositions}, we could show that if $\cC$ is the class of proper interval graphs, then for any graph $G$ it holds that $\simul[\cC][G] \leq \br{\bw{G}}^2 + \bw{G}$ and $\simulnonempty[\cC][G] \leq \br{\bw{G}}^2 + \bw{G} + 1$. However, using a simpler proof we show a weaker bound which, nevertheless, is sufficient for our purposes.

 \begin{lemma} \label{lem:bw}
     Let $\cC$ be the class of proper interval graphs.
     Then $\simul[\cC][G] \leq 2 \br{\bw{G}}^2 + 2 \bw{G}$ and $\simulnonempty[\cC][G] \leq 2\br{\bw{G}}^2 + 2\bw{G} + 1$ for any graph $G$.
 \end{lemma}

 \begin{figure}
    \centering
        \begin{tikzpicture}
        \def\labelone{BTUred}
        \def\labeltwo{lipicsYellow}
        \def\labelthree{coolblue}
        \def\nolabel{lipicsGray}
        \draw[ultra thick,\labelone] (0,1.125) -- (2.25,1.125);
        \draw[ultra thick,\labelone] (0.75,0.75) -- (3,0.75);
        \draw[ultra thick,\labelone] (1.5,0.375) -- (3.75,0.375);
        \draw[ultra thick,\labelone] (2.25,0) -- (4.5,0);

        \draw[ultra thick,\labeltwo] (3,1.125) -- (5.25,1.125);
        \draw[ultra thick,\labeltwo] (3.75,0.75) -- (6,0.75);
        \draw[ultra thick,\labeltwo] (4.5,0.375) -- (6.75,0.375);
        \draw[ultra thick,\labeltwo] (5.25,0) -- (7.5,0);

        \draw[ultra thick,\labelone] (6,1.125) -- (8.25,1.125);
        \draw[ultra thick,\labelone] (6.75,0.75) -- (9,0.75);
        \draw[ultra thick,\labelone] (7.5,0.375) -- (9.75,0.375);
        \draw[ultra thick,\labelone] (8.25,0) -- (10.5,0);

        \node at (3.25,-0.5) {$\textcolor{\labelone}{\boldsymbol{\alpha}}$};
        \node at (6.25,-0.5) {$\textcolor{\labeltwo}{\boldsymbol{\beta}}$};
        \node at (9.25,-0.5) {$\textcolor{\labelone}{\boldsymbol{\alpha}}$};
        \node at (1.5,0.225) {$\textcolor{\labelone}{\boldsymbol{\alpha}_{\mathsf{\mathbf{3,2}}}}$};
        \node[font=\sffamily\bfseries] at (4,1.375) {$\textcolor{\labelone}{\boldsymbol{\alpha}_{\mathsf{\mathbf{3,2}}}}$,$\textcolor{\labeltwo}{\boldsymbol{\beta}_{\mathsf{\mathbf{1,2}}}}$,$\textcolor{\labeltwo}{\boldsymbol{\beta}_{\mathsf{\mathbf{1,3}}}}$};
        \node at (4.5,0.25) {$\textcolor{\labeltwo}{\boldsymbol{\beta}_{\mathsf{\mathbf{1,2}}}}$};
        \node at (5.25,-0.125) {$\textcolor{\labeltwo}{\boldsymbol{\beta}_{\mathsf{\mathbf{1,3}}}}$};
        \draw (-0.5,-1) -> (11,-1);
        \foreach \x in {1,2,3,4,5,6,7,8,9,10,11,12,13,14,15} \draw ({(\x-1)*0.75},-1) -- ({(\x-1)*0.75},-1.1) node[below] {\x};
    \end{tikzpicture}
    \caption{A visualization of the construction used in the proof of \cref{lem:bw} for bandwidth three. The proper interval representation is depicted for twelve vertices, which are grouped into three sets of four vertices, the $\alpha$ and the $\beta$ sets. Possible labels for the first vertex of the $\beta$ set and the other vertices where those labels appear are depicted.}
    \label{fig:band}
\end{figure}
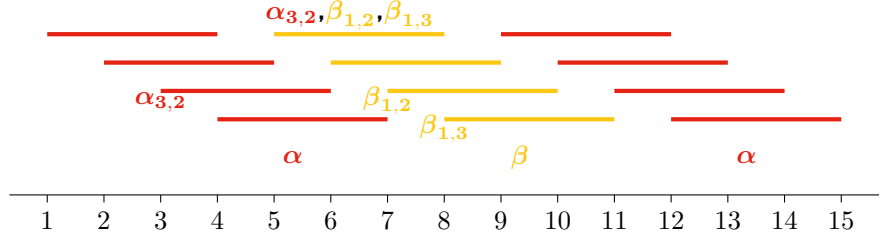

\begin{proof}
    As proper interval graphs are closed under the addition of isolated vertices, we only have to consider the bound on $\simul[\cC]$.
    Let $G$ be a graph of bandwidth $k$ and $f : V(G) \to \set{1,\dots,n}$ a bijection verifying that.
    We define a proper interval model of $G$ by $R(v) := [f(v), f(v) + k]$ for all $v \in V(G)$ and note that the graph $H$ defined by the proper interval model $R$ contains $G$ and the clique number of $H$ is $k+1$.
    Observe that a vertex $v \in V(H)$ is adjacent to every vertex with at most distance $k$ according to $f$ in $H$.
    We introduce labels $\alpha_{ij}$ and $\beta_{ij}$ with $1 \leq i \leq k+1$ and $1 \leq j \leq k$.
    The interpretation of these labels is the following: we cut the set of numbers $\set{1,2,\dots,n}$ into intervals of size $k+1$.
    The $\alpha$ labels correspond to even intervals, the $\beta$ labels to odd intervals.
    The indices $ij$ have the meaning that the $i$-th vertex in the current interval is adjacent to the vertex $j$ steps forward.
    
    We describe a labeling procedure based on this idea (see \cref{fig:band} for an illustration):
    Let $u \in V(G)$ be a vertex with $f(u) = a \cdot (k+1) + i$ with $1 \leq i \leq k+1$ and assume it is adjacent to a vertex $v \in V(G)$ with $f(v) = f(u) +j $ with $j \leq k$.
    If $a$ is even, we add the label $\alpha_{ij}$ to the label sets $L(u)$ and $L(v)$, otherwise we add the label $\beta_{ij}$ to their label sets.
    
    We claim that the so-defined $(H,L)$ is a $2(k+1)k$-simultaneous $\cC$-representation of $G$.
    Assume that there are vertices $u,v \in V(G)$ with $uv \in E(G)$.
    Then, by definition, $uv \in E(H)$ and $L(u) \cap L(v) \neq \emptyset$.
    So consider the case that there are vertices $u,v \in V(G)$ with $uv \in E(H)$ and $L(u) \cap L(v) \neq \emptyset$.
    W.l.o.g. assume that $f(u) < f(v)$ holds.
    Let $\ell \in L(u) \cap L(v)$.
    First consider the case that $\ell$ is of the form $\alpha_{ij}$.
    Observe the fact that the label $\alpha_{ij}$ can, by construction, only appear twice in an interval of $2k+1$ vertices.
    Once created by a vertex $x$ with $f(x) = a \cdot (k+1) + i$ and once obtained by a vertex $y$ with $f(y) = a \cdot (k+1) + i + j$.
    We consider the interval of size $2k+1$ centered around $u$.
    Note that all possible neighbors of $u$ are contained in that interval.
    By the comments above, it must hold that $f(u) = a \cdot (k+1) + i$ and $f(v) = f(u) + j$.
    Thus, the edge $uv$ must exist in $G$, as the label $\alpha_{ij}$ is only added to $L(v)$ if it is adjacent to $u$ in $G$.
    The case that $\ell$ is of the form $\beta_{ij}$ follows analogously.
\end{proof}

A generalization of bandwidth, which considers tree-orders instead of linear orders, was recently introduced in \cite{jacob2025treebandwidth}.

\begin{definition}[Treebandwidth]
    A \emph{tree-layout} of a graph $G = (V,E)$ is a rooted tree $(T,r)$ whose nodes are the vertices $V$, such that, for every edge $uv \in E$, $u$ is an ancestor of $v$ in $(T,r)$ or vice-versa.
    The \emph{bandwidth} of $(T,r)$ is the maximum distance between pairs of neighbors in $G$.
    The \emph{treebandwidth} of $G$, denoted by $\tbw{G}$, is the minimum bandwidth over all tree-layouts of $G$.
\end{definition}

It was shown in \cite{jacob2025treebandwidth} that a graph $G$ having treebandwidth $\leq k$ is equivalent to that graph having a proper chordal supergraph with clique number $\leq k+1$.
\emph{Proper chordal graphs} are graphs that have a tree-layout, where every root-to-leaf path is a proper interval graph \cite{PaulP24}.
The above mentioned proof shows that the proper chordal tree-layout associated with the supergraph $H$ also verifes that the treebandwidth of $G$ is $\leq k$.

\begin{theorem} \label{thm:tbw}
    Let $\cC$ be the class of proper chordal graphs.
    Then $\simul[\cC][G] \leq 2 \br{\tbw{G}}^2 + 2 \tbw{G}$ and $\simulnonempty[\cC][G] \leq 2 \br{\tbw{G}}^2 + 2 \tbw{G} + 1$ for any graph $G$.
\end{theorem}

\begin{proof}
    As proper chordal graphs are closed under the addition of isolated vertices, we only have to consider the bound on $\simul[\cC]$.
    Let $G$ be a graph with treebandwidth $k$ and $H$ a proper chordal supergraph of $G$ with clique number $k+1$.
    Let $(T,r)$ be the proper chordal tree-layout of $H$ which also verifies $\tbw{G} = k$.
    We use the same construction as in the proof of \cref{lem:bw} for every root-to-leaf path of $(T,r)$. The correctness of that construction now works analogously to the correctness in \cref{lem:bw}.
\end{proof}

Together with \cref{lem:c1c2} we can characterize when any of these parameters upper bounds some simultaneous $\cC$-number.

\begin{theorem}
    Let $\cC$ be a hereditary completable graph class.
    \begin{itemize}
        \item Treedepth upper bounds the simultaneous $\cC$-number if and only if the simultaneous $\cC$-number of trivially perfect graphs is bounded by their clique number.
        \item Bandwidth upper bounds the simultaneous $\cC$-number if and only if the simultaneous $\cC$-number of proper interval graphs is bounded by their clique number.
        \item Treebandwidth upper bounds the simultaneous $\cC$-number if and only if the simultaneous $\cC$-number of proper chordal graphs is bounded by their clique number.
        \item Pathwidth upper bounds the simultaneous $\cC$-number if and only if the simultaneous $\cC$-number of interval graphs is bounded by their clique number.
        \item Treewidth upper bounds the simultaneous $\cC$-number if and only if the simultaneous $\cC$-number of chordal graphs is bounded by their clique number.
    \end{itemize}
\end{theorem}

\begin{proof}
    We give the proof for treedepth and trivially perfect graphs, the other cases follow analogously with their respective upper bounds.

    Let $\cK$ be the class of trivially perfect graphs.
    Assume that we are given a hereditary completable graph class $\cC$ such that treedepth upper bounds the simultaneous $\cC$-number.
    Then, as the treedepth on graphs of $\cK$ equals the clique number plus one, the simultaneous $\cC$-number of graphs of $\cK$ is bounded by their clique number.
    
    For the converse direction, let the simultaneous $\cC$-number of graphs of $\cK$ be bounded by their clique number, i.e., there is a function $f$ such that $\simul[\C](G) \leq f(\omega(G))$ for all graphs $G \in \cK$.
    Then consider some graph $G$ and construct the $\binom{\tdbr{G}+1}{2}$-simultaneous $\cK$-representation $(H_1,L_1)$ of $G$ from \cref{thm:td}.
    Note that $H_1 \in \cK$ and the clique number of $H_1$ is equal to $\td{G} + 1$.
    Thus, by \cref{lem:c1c2}, we have that $G$ admits a simultaneous $\cC$-representation with the number of labels is at most $\binom{\tdbr{G}+1}{2} \cdot f(\td{G} + 1)$ which only depends on $\td{G}$. Hence, $\simul[\C][G]$ is bounded by a function in $\td{G}$.
\end{proof}

The same proof works for simultaneous numbers with non-empty label sets.

\subsection{Parameters That Are Never Upper Bounds}

Closing this section, we present a parameter that can never be an upper bounds on any simultaneous $\cC$-number for non-trivial classes $\cC$. This parameter is the modular width. To this end, we construct for every completable hereditary non-trivial $\cC$ a family of graphs with bounded modular-width but unbounded $\cC$-chromatic number, and, thus, with unbounded simultaneous $\cC$-number, due to \cref{lem:chromaticbound}.
For this construction we use \emph{lexicographic powers} of graphs.
Let $G$ and $F$ be two graphs.
The \emph{lexicographic product} or \emph{composition} $G \bullet F$ of $G$ and $F$ is defined as the graph with the following vertex and edge sets:
\begin{align*}
    V\!\br{G \bullet F} &:= V(G) \times V(F), \\
    E\!\br{G \bullet F} &:= \setcond{(u,v)(x,y)}{ux \in E(G) \text{ or } \br{u = x \ \wedge \ vy \in E(F)}}.
\end{align*}
Equivalently, the graph $G \bullet F$ is obtained by substituting every vertex in $G$ with a copy of $F$.
Thus, we can also think of $G \bullet F$ in terms of the substitution operation used in the definition of modular width.

\begin{lemma} \label{thm:lexprod}
    Let $H$ be a graph and $\mathcal{P}$ be a hereditary $H$-free graph class.
    Let $(G_i)_{i \geq 0}$ be the family of graphs with $G_0 := H$ and $G_{i+1} := H \bullet G_i$.
    Then $\chi_{\mathcal{P}} (G_i) \geq \chi_{\mathcal{P}} (H) + i$.
\end{lemma}

\begin{proof}
    We use induction over $i$.
    For $i = 0$ the statement is trivial.
    So we assume that the statement holds for some $i \geq 0$ and show that it holds for $i+1$.
    Assume to the contrary that we have $\chi_{\mathcal{P}} (G_{i+1}) \leq \chi_{\mathcal{P}}(H) + i =: k$.
    Let $X_1, X_2, \dots, X_k$ be a $\mathcal{P}$-coloring of $G_{i+1}$ and
    \[ V_u := \setcond{(u,v) \in V(G_{i+1})}{v \in V(G_i)} \]
    for $u \in V(H)$ be the vertex sets of the copies of $G_i$ used in $G_{i+1}$.
    As $X_1 \cap V_u,\dots, X_k \cap V_u$ is a $\mathcal{P}$-coloring of $G_{i+1}[V_u]$ and we have $\chi_{\mathcal{P}} (G_{i+1}[V_u]) = \chi_{\mathcal{P}}(G_i) \geq k$ for all $u \in V(H)$, it follows that $X_j \cap V_u \neq \emptyset$ for all $1 \leq j \leq k$ and all $u \in V(H)$.
    Thus, in any $X_j$ we find for every $u \in V(H)$ some vertex $(u,v) \in V(G_{i+1})$.
    However, this implies that $G_{i+1}[X_j]$ contains an induced $H$, a contradiction to the fact that $X_1, X_2, \dots, X_k$ form a $\mathcal{P}$-coloring of $G_{i+1}$.
    Thus, the assumption was false and we have shown that $\chi_{\mathcal{P}}(G_{i+1}) \geq \chi_{\mathcal{P}}(H) + i + 1$.
\end{proof}

This, implies the following result.

\begin{lemma}
    For every non-trivial hereditary completable graph class $\mathcal{C}$ there is a family of graphs with bounded modular width but unbounded simultaneous $\mathcal{C}$-number. 
\end{lemma}

\begin{proof}
    As $\mathcal{C}$ is non-trivial and hereditary, there exists some graph $H$ such that $\mathcal{C}$ is $H$-free.
    We choose this graph to have no isolated vertices. 
    This is possible since we just can add a universal vertex to some forbidden induced subgraph.
    By \cref{thm:lexprod} the family of graphs defined by $G_0 := H$ and $G_{i+1} := H \bullet G_i$ has unbounded $\mathcal{C}$-chromatic number. Therefore, it has unbounded $\simulnonempty[\cC]$ by \cref{thm:chromatic}, which is equal to $\simul[\cC]$ on these graphs due to \cref{obsv:isolatedvertices2} and the fact that no graph in this family has isolated vertices.
    On the other hand, the construction of $G_i$ directly gives an algebraic expression constructing $G_i$, which only uses the substitution operation with respect to $H$. Thus, the modular-width of $G_i$ is at most $\abs{V(H)}$ for all $i \in \N_0$. 
\end{proof}

Using this, we can finally deduce the following.

\begin{theorem}
    Let $\cC$ be a hereditary completable graph class.
    Modular-width cannot upper bound the simultaneous $\cC$-number. The same holds for all parameters that are upper bounded by modular-width. 
\end{theorem}

An overview of parameters that are upper bounded by modular-width can be found in \cref{fig:diagram}. It includes parameters as tree-length, cliquewidth, mim-width, and twin-width.

\section{Algorithms}\label{sec:algo}

\subsection{Computation}

To make use of simultaneous $\cC$-numbers as parameters in algorithms, we need to be able to compute optimal simultaneous representations.
Beisegel et al.~\cite{beisegel2024simultaneousinterval} showed that computing the simultaneous $\cC$-number with a corresponding simultaneous representation is $\NP$-hard for classes $\cC$ that contain all complete split graphs and forbid the induced subgraph $C_4$.
Since the simultaneous $\mathsf{complete}$-number is equal to the edge clique cover number it can be computed in $\FPT$ time parameterized by solution size~\cite{gramm2008data}. However, under the Exponential Time Hypothesis, double exponential time is needed to compute the simultaneous $\mathsf{complete}$-number parameterized by solution size~\cite{cygan2016known}. Intuitively, computing other simultaneous $\cC$-numbers should be even harder than computing the simultaneous $\mathsf{complete}$-number. Thus, we should not hope for \FPT{} algorithms with reasonable running time bounds. In fact, it is not even clear whether any other simultaneous $\cC$-number can be computed in $\XP$ time.

On the contrary, Bonomo-Braberman, Brandwein, and Sau~\cite{bonomobraberman2025computing} have shown that computing the simultaneous interval number is $\FPT$ when parameterized by $\mathsf{cluster}$-modular cardinality plus solution size. We generalize and expand on the ideas used in that algorithm, to obtain several algorithms parameterized by suitable $\cG$-modular cardinalities.

First, we give several algorithmic results that heavily build on \cref{thm:quotientgraph}. 
The idea is that for a graph class fulfilling one of the given conditions, we compute a suitable $\cG$-modular partition $P$ and consider the quotient graph $G/P$.
For $G/P$ we use a brute-force approach to compute $\simulnonempty[\cC]$ and can then transform the obtained simultaneous $\cC$-representation to one for $G$.

\begin{theorem} \label{thm:computing}
    Let $\cC$ be a graph class that is decidable, hereditary, and completable.
    \begin{itemize}
        \item If $\cC$ is closed under the addition of false twins, then computing $\simul[\cC]$ and $\simulnonempty[\cC]$ is $\FPT$ parameterized by $\mathsf{edgeless}$-modular cardinality.
        \item If $\cC$ is closed under the addition of true twins, then computing $\simul[\cC]$ and $\simulnonempty[\cC][G]$ is $\FPT$ parameterized by $\mathsf{complete}$-modular cardinality.
        \item If $\cC$ is closed under the addition of true and false twins, then computing $\simul[\cC]$ and $\simulnonempty[\cC]$ is $\FPT$ parameterized by $\mathsf{cograph}$-modular cardinality (iterated type partition).
        \item If $\cC$ is closed under vertex substitution, then computing $\simul[\cC]$ and $\simulnonempty[\cC]$ is $\FPT$ para\-meterized by $\cC$-modular cardinality.
    \end{itemize}
\end{theorem}

\begin{proof}
    Let $G$ be a graph.
    By results in \cite{lafond2023parameterized}, the parameters $\cG$-modular cardinality for $\cG \in \set{\mathsf{edgeless}, \mathsf{complete}, \mathsf{cograph}}$ can be calculated in polynomial time. It was furthermore shown that for so called \emph{trivially mergeable} graph classes $\cG$, $\cG$-modular cardinality can be computed in polynomial time. A graph class is trivially mergeable if the following holds: the optimal $\cG$-modular partition of the disjoint union (or join) $G$ of multiple graphs of $\cG$ is either given by the vertex sets of the used graphs or by $\set{V(G)}$. Note that hereditary completable graph classes $\cC$ which are closed under vertex substitution are trivially mergeable, and therefore the $\cC$-modular cardinality of such a graph class can be computed in polynomial time.
   
    We first assume that $\cC$ is closed under the addition of false twins.
    By the above comments, we can compute an $\mathsf{edgeless}$-modular partition $P$ of $G$ in polynomial time.
    Consider the graph $G/P$ and let $k = \abs{V(G/P)}$.
    We use a brute force approach to compute $\simulnonempty[\cC][G/P]$: Enumerate every possible supergraph $H$ of $G/P$ on $k$ vertices contained in $\cC$. Enumerate for every possible value $d$ with $1 \leq d \leq k^2$ every possible labeling function $L$ on $V(H)$. Then check for every pair $(H,L)$ if it is a simultaneous $\cC$-representation of $G/P$.
    The enumeration of every supergraph in $\cC$ can, for example, be done by enumerating every graph on $k$ vertices, testing if it is a supergraph of $G/P$, and deciding if it is in $\cC$.
    Thus, $\simulnonempty[\cC][G/P]$ can be computed in time only dependent on $k$ and by \cref{thm:quotientgraph} equals $\simulnonempty[\cC][G]$.
    Furthermore, a simultaneous $\cC$-representation of $G$ can be constructed in polynomial time using the process described in the proof of \cref{thm:quotientgraph}.
    To compute $\simul[\cC]$, note that it only differs from $\simulnonempty[\cC]$ if isolated vertices exist. Thus, we can remove any isolated vertices of $G$ in a preprocessing step and later add them back to the graph as twins to an arbitrary vertex and give them empty label sets.

    The other three cases follow analogously.
\end{proof}

For example, interval graphs are closed under addition of true twins and therefore the simultaneous interval number can be computed in $\FPT$ time when parameterized by $\mathsf{complete}$-modular cardinality. As the other three conditions do not hold for interval graphs, this is the only of the above four results we can use to compute the simultaneous interval number.
As mentioned before, Bonomo-Braberman, Brandwein, and Sau~\cite{bonomobraberman2025computing} have given an $\FPT$ algorithm for computing the simultaneous interval number parameterized by $\mathsf{cluster}$-modular cardinality plus solution size. Their algorithm makes implicit use of the fact that interval graphs are closed under addition of true twins, but also somehow deals with false twins.
To this end, they show that, under certain conditions, vertices behaving similarly to false twins can be added to interval graphs. The property that they use is given in the following observation.

\begin{observation}
    Let $G$ be an interval graph.
    Let $u,v \in V(G)$ be two non-adjacent vertices such that $N(u) \cap N(v) \neq \emptyset$.
    Then, the graph obtained by adding a vertex $w$ with neighborhood $N(u) \cap N(v)$ to $G$ is an interval graph.
\end{observation}

The idea is to look at the space between the intervals of $u$ and $v$.
Any intervals in $N(u) \cap N(v)$ must cover this space and we can arrange the interval in such a way that there is some interval of the real line that is contained only in the intervals of $N(u) \cap N(v)$. Adding the interval of the newly added vertex there yields an interval model of the desired graph.

We formalize this property of graph classes as follows.

\begin{definition}
    Let $\cC$ be a graph class.
    We say that $\cC$ has the \emph{intersection sibling property} if for any graph $G \in \cC$ with two non-adjacent vertices $u,v \in V(G)$ that have non-empty neighborhoods, the graph obtained by adding a vertex $w$ with neighborhood $N(u) \cap N(v)$ to $G$ is in $\cC$.
\end{definition}

Other examples of graph classes with this property are trivially perfect graphs, interval graphs, chordal graphs, proper chordal graphs, and split graphs. Graph classes $\cC$ that have the intersection sibling property allow us, under certain conditions, to add false twins to a graph without increasing the simultaneous $\cC$-number.

\begin{lemma} \label{thm:intersectionsibling}
    Let $\cC$ be a hereditary completable graph class that has the intersection sibling property.
    Let $G$ be a graph with false twins $u,v \in V(G)$.
    If $(H,L)$ is a $d$-simultaneous $\cC$-representation of $G$ such that $uv \not\in E(H)$, then the graph $G'$ obtained by adding another false twin to $u$ and $v$ admits a $d$-simultaneous $\cC$-representation $(H',L')$.
\end{lemma}

\begin{proof}
    Let $G$ and $(H,L)$ be given such that $u,v \in V(G)$ are false twins in $G$ and non-adjacent in $H$.
    As $u$ and $v$ are false twins in $G$, it holds that $N_G(u) = N_G(v) \subseteq N_H(u) \cap N_H(v)$.
    Thus, we can add a vertex $w$ to $H$ whose neighborhood is $N_H(u) \cap N_H(v)$ to obtain the graph $H'$.
    We define $L'(x) = L(x)$ for all $x \in V(G)$ and $L'(w) = L(u)$.
    Let $G'$ be the graph defined by the $d$-simultaneous $\cC$-representation $(H',L')$.
    Then, as $N_G(u) \subseteq N_{H'}(w) \subseteq N_{H'}(u)$ and $L'(u) = L'(w)$ holds, it follows that $N_{G'}(w) = N_G(u) = N_{G'}(u)$ and $G'$ is the desired graph.
\end{proof}

This lemma gives us the backbone of the algorithm given by Bonomo-Braberman, Brandwein, and Sau~\cite{bonomobraberman2025computing}. We can use it to formulate an algorithm for computing simultaneous $\cC$-numbers where $\cC$ has the intersection sibling property parameterized by $\mathsf{edgeless}$-modular cardinality plus solution size. Note that $\mathsf{edgeless}$-modular cardinality is never explicitly used in~\cite{bonomobraberman2025computing}.

\begin{theorem} \label{thm:edgelessintersection}
    Let $\cC$ be a graph class that is decidable, hereditary, completable, and has the intersection sibling property.
    Let $G$ be a graph with $\mathsf{edgeless}$-modular cardinality $k$.
    Then, deciding if $G$ admits a $d$-simultaneous $\cC$-representation (with non-empty label sets) and computing one if it exists is $\FPT$ parameterized by $k+d$.
\end{theorem}

\begin{proof}
    We consider the case of simultaneous representations with non-empty label sets.
    The case with possibly empty label sets can be done similarly to \cref{thm:computing} by first removing all isolated vertices, using the algorithm for non-empty label sets on the smaller instance, and later adding the isolated vertices back to the graph as vertices with empty label sets. Removing vertices is allowed since $\cC$ is hereditary. Adding back the vertices is possible as follows: if the graph $H$ used in the simultaneous representation has a pair of non-adjacent vertices, we can add arbitrarily many vertices due to the intersection sibling property. Otherwise, the graph $H$ is a complete graph and we can add arbitrarily many universal vertices due to $\cC$ being completable. 
    
    Let $G$ be a graph and $P = \{M_1, \dots, M_k\}$ be an optimal $\mathsf{edgeless}$-modular partition of $G$, which can be computed in polynomial time~\cite{lafond2023parameterized}.
    Consider some $M_i$ with at least two vertices.
    For a simultaneous $\cC$-representation $(H,L)$ of $G$ there are two cases:
    either at least two vertices of $M_i$ are non-adjacent in $H$, or $M_i$ is a clique in $H$. We now fix for every $M_i$ with at least two vertices which of the two cases holds and consider all $2^k$ possible choices separately. 
    
    If $M_i$ contains two vertices that are not adjacent in $H$, then it suffices to only consider these two vertices of $M_i$ as we can add the other vertices without increasing the number of labels using \cref{thm:intersectionsibling}.
    If $M_i$ is a clique in $H$, then a simultaneous $\cC$-representation with non-empty labels needs at least $\abs{M_i}$ many labels since $M_i$ is independent in $G$ but a clique in $H$.
    So, if $\abs{M_i} > d$, we can discard the problem instance.
    Thus, the number of vertices in all non-discarded instances is bounded by $d \cdot k$.
    
    For each instance, we can use a brute force approach, as described in the proof of \cref{thm:computing}, to decide whether a $d$-simultaneous $\cC$-representation with non-empty label sets of the constrained instance exists.
    If it does, we can use \cref{thm:intersectionsibling} to construct a $d$-simultaneous $\cC$-representation of $G$ with non-empty label sets.
\end{proof}

Every graph class we have mentioned having the intersection sibling property is also closed under the addition of true twins.
Under these conditions, we can give a generalization of the algorithm given for the simultaneous interval number in~\cite{bonomobraberman2025computing}.

\begin{theorem}
    Let $\cC$ be a graph class that is decidable, hereditary, completable, closed under the addition of true twins, and has the intersection sibling property.
    Let $G$ be a graph with $\mathsf{cluster}$-modular cardinality $k$.
    Then, deciding if $G$ admits a $d$-simultaneous $\cC$-representation (with non-empty label sets) and computing one if it exists is $\FPT$ parameterized by $k+d$.
\end{theorem}

\begin{proof}
    We consider the case of simultaneous representations with non-empty label sets.
    The case with possibly empty label sets can be done similarly to \cref{thm:computing} by first removing all isolated vertices, using the algorithm for non-empty label sets on the smaller instance, and later adding the isolated vertices as true twins to some vertex and giving them empty label sets.

    Let $G$ be a graph and $M_1, \dots, M_k$ be an optimal $\mathsf{cluster}$-modular partition which can be computed in linear time~\cite{bonomobraberman2025computing}.
    We can contract the complete subgraphs of any $M_i$ without changing the simultaneous $\cC$-number of the graphs, as $\cC$ is closed under the addition of true twins.
    Thus, we may assume that every $M_i$ is an independent set and can use the algorithm of \cref{thm:edgelessintersection}.
\end{proof}

\subsection{Cliques}

Here, we consider the following problem.

\begin{problem}[Maximum Clique Problem given a simultaneous $\cC$-representation]~ \label{prob:clique}
    \begin{description}
        \item[Input:] A graph $G$ and a $d$-simultaneous $\cC$-representation of $G$.
        \item[Parameter:] $d$.
        \item[Task:] Compute the clique number of $G$.
    \end{description}
\end{problem}

Beisegel et al.~\cite{beisegel2024simultaneousinterval} showed that this problem can be solved in \FPT{} time when $\cC$ is the class of interval graphs. We show that this result can be generalized to all classes $\cC$ on which the maximum clique problem can be solved in polynomial time. Of course, this condition is also necessary.

\begin{theorem} \label{thm:clique}
    Let $\cC$ be a hereditary completable graph class.
    \cref{prob:clique} is solvable in $\FPT$ time if and only if the maximum clique problem can be solved on $\cC$ in polynomial time.
\end{theorem}

\begin{proof}
    Obviously, if \cref{prob:clique} is solvable in $\FPT$ time, then the maximum clique problem is solvable on $\cC$ in polynomial time.
    So we consider only the other direction.

    Let $\cC$ be a hereditary completable graph class such that the maximum clique problem can be solved in polynomial time $f(n,m)$ on $\cC$.
    Let $G$ be a graph with $n = \abs{V(G)}, \ m = \abs{E(G)}$ and $\br{H, L}$ a $d$-simultaneous $\cC$-representation of $G$.
    We first note that for any clique $C \subseteq V(G)$ we have $L(u) \cap L(v) \neq \emptyset$ for any two vertices $u,v \in C$.
    Thus, it suffices to only consider induced subgraphs where the label sets of vertices pairwise intersect.
    Such subgraphs are induced subgraphs of $H$ and are therefore in $\cC$.
    Hence, the maximum clique problem can be solved on such subgraphs in time $f\!\br{n,m}$.
    
    This is the basis of our algorithm.
    We generate all binary vectors in $\set{0,1}^{2^d}$. 
    Every entry represents some label set.
    We only consider those vectors, where the label sets with entry~$1$ pairwise intersect. This can be checked in $\O\!\br{d \cdot 2^{2d}}$ time per vector.
    For every of those vectors, we consider the subgraph induced by the vertices that have a label set with entry $1$ and solve the maximum (weighted) clique problem for every of those subgraphs.
    In total, we need time $\O\!\br{2^{2^d} \cdot (n +m)}$ to generate all these subgraphs and time $\O\!\br{2^{2^d} \cdot f\br{n,m}}$ to solve the maximum (weighted) clique problem on all of them.
\end{proof}

The same result also holds for the weighted maximum clique problem. The proof of this theorem also shows that if the number of maximal cliques in any graph in class $\cC$ is polynomial in $n$, then this also holds for any graph with simultaneous $\cC$-number $d$, for some fixed $d$.
In this case, we can use the algorithm of \cite{tsukiyama1977cliques} to compute all maximal cliques in $\FPT$ time parameterized by $\simul[\cC]$ without needing a simultaneous $\cC$-representation as part of the input.

\begin{corollary}
     Let $\cC$ be a hereditary completable graph class such that graphs of $\cC$ have a polynomial number of maximal cliques.
    Then the Maximum Clique Problem can be solved in \FPT{} time when parameterized by the simultaneous $\cC$-number.
\end{corollary}

Similarly to {\cite[Theorem 5.7]{beisegel2024simultaneousinterval}}, we can show that the double exponential time of the above algorithm can likely not be improved to a single exponential time bound.

\begin{theorem}
    For any hereditary completable graph class $\cC$ and any fixed $k \in \N$, there is no algorithm that solves the maximum clique problem on complements of cubic graphs in time $2^{\O\br{\simulnonempty[\cC]}} \cdot n^k$, unless $\P = \NP$. 
\end{theorem}

\begin{proof}
    Complements of cubic graphs with $n \geq 5$ vertices do not contain isolated vertices.
    Thus, $\simul[\cC]$ and $\simulnonempty[\cC]$ are equal on these graphs.
    Alon~\cite{alon1986cubicecc} showed that there exists a constant $c \in \N$ such that the edge clique cover number on complements of cubic graphs on $n$ vertices is bounded by $c \log n$.
    Furthermore, Beisegel et al.~\cite{beisegel2024simultaneousinterval} showed that $\simul[\cC][G] \leq \ecc{G}$ holds for all hereditary completable $\cC$ and all graphs $G$.
    This together with the observation above implies that $\simulnonempty[\cC]$ is bounded by $c \log n$ on complements of cubic graphs with $n \geq 5$ vertices.
    Therefore, an algorithm that has the runtime given in the statement of the theorem actually is a polynomial time algorithm in $n$.
    However, Mohar~\cite{mohar2001cubicclique} proved $\NP$-hardness of the maximum independent set problem on cubic graphs, implying $\NP$-hardness of the maximum clique problem on their complements.
    Thus, the existence of such an algorithm would imply $\P = \NP$.
\end{proof}

There may be two reasons why this theorem holds. Either, the complexity lower bound comes from the complexity of the clique problem or from complexity of computing a simultaneous representation. In fact, it might be possible that we can achieve single-exponential time if the representation is given. We could rule out this possibility for any class $\cC$ by presenting a polynomial-time algorithm that computes or approximates the edge clique cover number of complements of cubic graphs. To the best of our knowledge, it is open whether this is possible.

\section{Concluding Remarks}

In this paper, we have shown for a large number of structural graph parameters that they are $\simul$-bounded.
This means that the structural properties that these parameters describe behave similarly on a graph class $\cC$ and the class of graphs with simultaneous $\cC$-number $d$ for some fixed $d$.
In particular, we have shown that all parameters which were shown to be lower bounds on the simultaneous interval number by Beisegel et al.~\cite{beisegel2024simultaneousinterval} are $\simul$-bounded, meaning that these lower bounds follow from the fact that simultaneous $\cC$-numbers preserve certain structural properties of their corresponding graph classes and are not directly related to the special structure of interval graphs.
We have furthermore characterized those graph classes for which simultaneous $\cC$-numbers are lower bounds on treedepth, bandwidth, treebandwidth, pathwidth, and treewidth. These characterizations show that it suffices to show that one of these parameters is an upper bound on some simultaneous $\cC$-number only on restricted graph classes. This again seems to suggest that structures of these parameters are somehow preserved under simultaneous $\cC$-representations. 

We hope to complement the structural results we have given with more algorithmic results in the future.
On the one hand, the usability of simultaneous $\C$-numbers to design parameterized algorithms for \NP-hard graph problems should be further studied. In particular, the existence of general theorems similar to \cref{thm:clique} for a family of graph classes would be interesting. On the other hand, less restrictive results or better parameterizations for computing simultaneous $\cC$-numbers would be desirable. 

Another avenue for future research directions with a more structural focus could be given by adapting the idea of local covering numbers to simultaneous $\cC$-numbers. A \emph{local simultaneous $\cC$-number} could count the maximum number of labels each vertex receives instead of the total number of labels used in the representation. Such a parameter would not be relevant regarding algorithmic aspects, as they capture graphs of bounded maximum degree on which many interesting graph problems are already hard. Structural results for these parameters could, however, be very interesting, for example an analogous concept to $\simul$-boundedness. 

\bibliography{lit}
\bibliographystyle{plainurl}

\newpage

\appendix

\section{Omitted Proofs} \label{sec:proofs}

\subsection{Forbidden Subgraph Characterizations of Graph Classes}

\begin{theorem}
    A graph is $(3K_1, C_4, P_4)$-free if and only if it is a complete graph, where the edge set of a complete bipartite graph is removed.
\end{theorem}

\begin{proof}
    It is easy to see that complete graphs, where the edge set of a complete bipartite graph is removed, are $(3K_1, C_4, P_4)$-free.
    
    Consider a $(3K_1, C_4, P_4)$-free graph $G$.
    Assume that there are two vertices $a,b \in V(G)$ such that $a$ and $b$ are non-adjacent (if no such vertices exist, $G$ is a complete graph).
    Any other vertex in $G$ must be adjacent to at least one of $a,b$, as a vertex which is non-adjacent to both of them would induce a $3K_1$ together with $a,b$.
    We consider the sets $ A :=\{a\} \cup (N(a) \setminus N(b))$ and $B := \{b\} \cup (N(b) \setminus N(a))$ and claim that these sets are cliques in $G$ and no edges between these $A$ and $B$ exist.
    Assume first that there are two vertices in $A$ that are non-adjacent. Then they induce a $3K_1$ together with $b$, therefore $A$ is a clique.
    Analogously, $B$ is a clique as well.
    Now assume there exists an edge $cd$ between the sets $A$ and $B$ with $c \in A$, $d \in B$.
    Then we have the induced $P_4$ $a, c, d, b$, thus no edges between $A$ and $B$ exist.
    It remains to show that the non-edges between $A$ and $B$ are the only non-edges in $G$.
    Assume to the contrary that there are vertices $u,v$ such that $u$ and $v$ are non-adjacent and $u \not\in A \cup B$.
    By the above comments, $u$ must be adjacent to both $a$ and $b$ and $v$ must be adjacent to at least one of $a$ and $b$.
    If $v$ is also adjacent to $a$ and $b$, then $a, u, b, v$ induces a $C_4$.
    On the other hand, if $v$ is adjacent to only one of $a,b$, then $a,b,u,v$ induce a $P_4$.
\end{proof}

\begin{theorem}
    A graph is $(2K_2, P_3)$-free if and only if it is the disjoint union of a complete graph and an edgeless graph.
\end{theorem}

\begin{proof}
    The $P_3$-free graphs are exactly the graphs where every connected component is a complete graph.
    Furthermore, $2K_2$-free graphs have at most one connected component that contains edges.
    Combining these two facts, we obtain the statement.
\end{proof}

\subsection{Proof of \cref{lem:partition}}

\ineq*

\begin{proof}
    $\abs{X_1} \leq \abs{X_2} \leq \dots \leq \abs{X_k}$ implies that $\abs{X_i} \leq \frac{n}{k-i+1}$, as otherwise there would be $k-i+1$ sets of cardinality $> \frac{n}{k-i+1}$, which in sum would be $> n$.
    We have:
    \begin{align*}
        \sum_{i = 1}^{k-1} \br{(k-i) \abs{X_i}} &\leq \sum_{i = 1}^{k-1} \frac{k-i}{k-i+1} n = n \sum_{i = 1}^{k-1} \frac{i}{i+1} = n \sum_{i = 1}^{k-1} \br{1 - \frac{1}{i+1}} \\
        &= n \br{k-1 - \sum_{i = 1}^{k-1} \frac{1}{i+1}} = n \br{k-1 - \sum_{i = 2}^{k} \frac{1}{i}} \\
        &\leq n \br{k - 1 - \ln k + 1} = n \br{k - \ln k}
    \end{align*}
    where we used the well-known approximation of the harmonic number
  \[ \ln k \leq \sum_{i=1}^k \frac{1}{i} \leq \ln (k+1) \Leftrightarrow \ln k  - 1 \leq \sum_{i=2}^{k} \frac{1}{i} \leq \ln (k+1) - 1, \]
  which follows from lower/upper sum approximations of the integral of the function $\frac{1}{x}$.
\end{proof}

\section{Better Bounds for Parameters Closed Under FO-Transductions} \label{sec:FObounds}

In this section, we will consider multiple parameters that cannot be defined in terms of individual graphs, but only in terms of graph classes. To this end, we make use of the term graph class parameter.
Remember that we denoted by $\bbG$ the class of all graphs. We denote by $\GC$ the class of all graph class.
A \emph{graph class parameter}~$p$ is a function $p : \GC \to \bbR \cup \set{\infty}$.
We can compare a graph class parameter to a regular graph parameter by considering its \emph{canonical graph class parameter} which is defined as $p(\cC) := \sup_{G \in \cC} p(G)$ for a graph class $\cC$.
A graph parameter and its canonical graph class parameter may be used interchangeably using the term parameter.
Again, if $p(\cC) < \infty$ for some graph class parameter $p$ and some graph class $\cC$, then $p$ is \emph{bounded on $\cC$}, otherwise it is \emph{unbounded on $\cC$}.
A parameter $p'$ \emph{upper bounds} the parameter $p$ if $p'(\cC) < \infty$ implies $p(\cC) < \infty$ for all graph classes $\cC$.
The definition of $\simul$-boundedness works exactly the same for graph class parameters. 
We will use this notion for the parameters flip-width, merge-width, and shrub-depth.

We will now start showing better bounds for parameters whose $\simul$-boundedness we obtained from closure under first-order transductions. We start with the well-known parameter cliquewidth.

\begin{definition}[Cliquewidth]
  Consider the following four operations on $k$-labeled graphs, i.e., graphs where every vertex has one of $k$ labels:
  \begin{enumerate}
    \item $\emptyset_i$: create a vertex with label $i$,
    \item $join(i,j)$: make all vertices with label $i$ adjacent to those with label $j$,
    \item $(i \mapsto j)$: relabel all vertices with label $i$ to label $j$,
    \item $\cup$: take the disjoint union of $G_1$ and $G_2$.
  \end{enumerate}
  The smallest number of labels needed to construct a graph $G$ using the above operations is called the \emph{cliquewidth} of $G$, denoted by $\cw{G}$.
\end{definition}

\begin{theorem} \label{thm:cwsimul}
    Cliquewidth is $\simul$-bounded.
    In particular, $\cw{G} \leq \cw{\cC} \cdot 2^{\simul[\cC][G]}$ for any hereditary completable $\cC$ and any graph $G$. This bound is tight up to a constant factor.
\end{theorem}

\begin{proof}
    Let $k \in \N$ and $\cC$ be a graph class such that $\cw{H} \leq k$ for all $H \in \cC$.
    Furthermore, let $G$ be a graph with $d = \simul[\cC][G]$ and $(H,L)$ be a $d$-simultaneous $\cC$-representation.
    Let a sequence of operations that constructs $H$ and uses $k$ labels be given.
    We introduce $k 2^d$ new labels $i_A$, one for every combination of an old label $1 \leq i \leq k$ and every subset $A$ of $\set{1,2,\dots,d}$.
    The idea is that we have for every vertex $v$ a label with index $L(v)$, simulating the behavior of the simultaneous representation.
    Then, we change the operations as follows:
    \begin{enumerate}
        \item $\emptyset_i$: instead of adding a vertex $v$ with label $i$, we add $v$ with label $i_{L(v)}$,
        \item $join(i,j)$: for all $A \subseteq \set{1,2,\dots,d}$ and all $B \subseteq \set{1,2,\dots,d}$ with $A \cap B \neq \emptyset$, we perform the operation $join(i_A,j_B)$,
        \item $(i \mapsto j)$: for all $A \subseteq \set{1,2,\dots,d}$, we perform the operation $(i_A \mapsto j_A)$,
        \item $\cup$: the operation is not changed.
    \end{enumerate}
    Obviously, the sequence of operations obtained by changing the operations in the original sequence as above results in the graph $G$.

    To show that this bound is tight up to a constant factor, we consider the graphs $G_{d,k}$ used in the proof of \cref{thm:modw} which are constructed using the simultaneous $\cocluster$-representations. Note that the cliquewidth of co-clusters is $2$.
    For some arbitrary $d \in \N$ we choose $k := 2^d + 1$ and show that $\cw{G_{d,2^d +1}} \geq 2^{d-1}$, while the simultaneous $\cocluster$-number of this family is obviously bounded by $d$.
    We assume that we are given $G_{d,2^d +1}$ together with the simultaneous $\cocluster$-representation $(H,L)$ constructed in the proof of \cref{thm:modw}.
    
    It was shown in \cite{alecu2021functionality} that for any graph $G$ of cliquewidth $w$ it holds that there exist two vertices $x,y \in V(G)$ with $\abs{N(x) \triangle N(y)} \leq 2w$.
    Thus, to show that the cliquewidth of $G_{d,2^d + 1}$ is $\geq 2^{d-1}$, it suffices to show that for every pair of vertices $x,y \in V\!\br{G_{d,2^d + 1}}$ we have $\abs{N(x) \triangle N(y)} \geq 2^d$.
    Let $x,y \in V\!\br{G_{d,2^d + 1}}$ be chosen arbitrarily.
    Let $x \in V_i$ and $y \in V_j$.
    We first consider the case that $i \neq j$.
    Then $x$ has at least $2^{d-1}$ neighbors in $V_j$, while $y$ has at least $2^{d-1}$ neighbors in~$V_i$.
    As $V_i$ and $V_j$ are independent sets, $x$ has no neighbors in $V_i$ and $y$ has no neighbors in~$V_j$.
    Thus, $\abs{N(x) \triangle N(y)} \geq 2^d$.
    Now consider the case that $i = j$.
    By construction, we have that $L(x) \neq L(y)$.
    W.l.o.g. let there be some label $\ell$ such that $\ell \in L(x) \setminus L(y)$.
    Then $x$ is adjacent to all vertices $v_{\set{\ell}}^{(h)}$ with $h \neq i$, while $y$ can not be adjacent to these vertices.
    As there are at least $k-1 = 2^d$ of these vertices, we have $\abs{N(x) \triangle N(y)} \geq 2^d$.
\end{proof}

We now examine boolean-width, which is functionally equivalent to cliquewidth.

\begin{definition}[Boolean-width]
    Let $G$ be a graph and $A \subseteq V(G)$.
    Define the set of unions of neighborhoods of $A$ across the cut $\set{A, \overline{A}}$ as 
    \[ U(A) := \setcond{Y \subseteq \overline{A}}{\exists X \subseteq A \text{ and } Y = N(X) \cap \overline{A}}. \]
    Then define $bool\text{-}dim_G : \cP(V(G)) \to \bbR$ by $bool\text{-}dim_G(A) := \log_2 \abs{U(A)}$.
    Using \cref{def:dectree} with $f = (bool\text{-}dim_G)_{G \in \cG}$ we define the \emph{(linear) boolean-width} of $G$ as the (linear) $f$-width of $G$.
\end{definition}

\begin{theorem}
    Let $d \in \N$, $G$ be a graph, and $(H,L)$ be a $d$-simultaneous $\cC$-representation of $G$ for some hereditary completable graph class $\cC$.
    Then $bool\text{-}dim_G(A) \leq 2^d + bool\text{-}dim_H(A)$ for any vertex set $A \subseteq V(G)$.
    Thus, (linear) boolean-width is $\simul$-bounded.
\end{theorem}

\begin{proof}
    Let $G$ be a graph with $d$-simultaneous $\cC$-representation $(H,L)$.
    If $d = 0$, $G$ is edgeless and the inequality holds trivially.
    Thus, let $d \geq 1$.
    Consider some vertex set $A \subseteq V(G)$.
    Denote by $U_H(A)$ the set of unions of neighborhoods of $A$ across the cut $\set{A, \overline{A}}$ in $H$ and $U_G(A)$ analogously for $G$.
    Then we have
    \[ U_G(A) \subseteq \setcond{Y \cap V_\cB}{ Y \in U_H(A), \ \cB \subseteq \cP (\set{1,2,\dots,d})}, \]
    where we define $V_\cB = \setcond{v \in V(G)}{L(v) = B \in \cB}$.
    Thus, $\abs{U_G(A)} \leq 2^{2^d} \abs{U_H(A)}$ and therefore $bool\text{-}dim_G(A) \leq 2^d + bool\text{-}dim_H(A)$.
\end{proof}

We now consider rank-width, which also is functionally equivalent to cliquewidth.

\begin{definition}[Rank-width]
    For a graph $G$ define the function $cut\text{-}rank_G : \cP(V(G)) \to \N$ with $cut\text{-}rank(A)$ for $A \subseteq V(G)$ being the rank of the adjacency matrix of $G[A, \overline{A}]$ in $GF(2)$, the finite field over two elements.
    Using \cref{def:dectree} with $f = (cut\text{-}rank_G)_{G \in \cG}$ we define the \emph{(linear) rank-width} of $G$ as the (linear) $f$-width of $G$.
\end{definition}

\begin{theorem}
    Let $d \in \N$, $G$ be a graph, and $(H,L)$ be a $d$-simultaneous $\cC$-representation of $G$ for some hereditary completable graph class $\cC$.
    Then $cut\text{-}rank_G(A) \leq 2^d \cdot cut\text{-}rank_H(A)$ for any vertex set $A \subseteq V(G)$.
    Thus, (linear) rank-width is $\simul$-bounded.
\end{theorem}

\begin{proof}
    Let $G$ be a graph with $d$-simultaneous $\cC$-representation $(H,L)$.
    If $d = 0$, $G$ is edgeless and the inequality holds trivially.
    Thus, let $d \geq 1$.
    Consider some vertex set $A \subseteq V(G)$.
    Let $k = cut\text{-}rank_H(A)$, that is, $k$ is the rank of the adjacency matrix $\cA_H$ of $H[A, \overline{A}]$ in $GF(2)$.
    Now consider the adjacency matrix $\cA_G$ of $G[A, \overline{A}]$.
    Note that the linear dependency of two column vectors of $\cA_G$ may only change whenever the label sets of the corresponding vertices are different, i.e., if two vertices have the same label set, their corresponding columns in $\cA_G$ are linearly independent if and only if they are linearly independent in $\cA_H$.
    Let us now consider some label set $B \subseteq \set{1,2,\dots,d}$ and the corresponding vertex set $V_B = \setcond{v \in V(G)}{L(v) = B}$.
    By the above argument, at most $k$ of the columns in $\cA_G$ corresponding to $V_B$ are linearly independent.
    Thus, as there are at most $2^d$ used label sets, we have that there are at most $2^d k$ many linearly independent column vectors in $\cA_G$ and, thus, $cut\text{-}rank_G(A) \leq 2^d k$.
\end{proof}

The next parameter we consider is \emph{twin-width}, which gives a generalization of cographs.
A \emph{trigraph} is a triple $G = (V,E,R)$, where $E$ and $R$ are two disjoint sets of edges on $V$: the edges $E$ and red edges $R$.
A trigraph $(V,E,R)$ such that $(V,R)$ has maximum degree at most $d$ is called a \emph{$d$-trigraph}.
Any graph $(V,E)$ may be interpreted as a trigraph $(V,E,\emptyset)$.
Given a trigraph $G = (V,E,R)$ and two vertices $u,v \in V$, we define the trigraph $G / u,v = (V^\prime, E^\prime, R^\prime)$ by \emph{contracting} $u,v$ into a new vertex $w$ as the trigraph on $V^\prime = \br{V \setminus \set{u,v}} \cup \set{w}$ such that $G - \set{u,v} = \br{G / u,v} - \set{w}$ and with the following edges incident to $w$:
\begin{itemize}
    \item $\set{w,x} \in E^\prime$ if and only if $\set{u,x}, \set{v,x} \in E$,
    \item $\set{w,x} \not\in E^\prime \cup R^\prime$ if and only if $\set{u,x} \not\in E \cup R$ and $\set{v,x} \not\in E \cup R$,
    \item $\set{w,x} \in R^\prime$ otherwise.
\end{itemize}
We say that $G / u,v$ is a \emph{contraction} of $G$.
If both $G$ and $G / u,v$ are $d$-trigraphs, $G / u,v$ is a \emph{$d$-contraction}.
A (tri)graph $G$ on $n$ vertices is \emph{$d$-collapsible}, if there exists a sequence of $d$-contractions, which contracts $G$ to a single vertex.

\begin{definition}[Twin-width]
  The minimum $d$ for which a graph $G$ is $d$-collapsible is the \emph{twin-width} of $G$, denoted by $\tww{G}$.
\end{definition}

\begin{theorem}
    Twin width is $\simul$-bounded.
    In particular, $\tww{G} \leq (\tww{\cC} + 1) \cdot 2^{\simulbr{\cC}{G}}$ for any hereditary completable $\cC$ and any graph $G$.
    This bound is tight up to a constant factor.
\end{theorem}

\begin{proof}
    Let $k \in \N$ and $\cC$ be a graph class such that $\tww{H} \leq k$ for all $H \in \cC$.
    Furthermore, let $G$ be a graph with $d = \simul[\cC][G]$ and $(H,L)$ be a $d$-simultaneous $\cC$-representation.
    We make use of Lemma 5.1 in \cite{BonnetKTW22}, which states that the twin-width of a graph augmented by $h$ unary relations grows at most by a factor of $2^h$ when compared to the twin-width of the underlying graph.
    In the case that unary relations are added to the graph, only vertices that agree on all unary relations can be contracted in a contraction sequence, thus leaving up to $2^h$ vertices in the end of the sequence instead of one.
    We can interpret the property of a vertex having label $i$ for $1 \leq i \leq d$ as a unary relation and equip the graph $H$ with these $d$ unary relations.
    Then by the above stated lemma, there is a $2^dk$-contraction sequence of $H$ where only vertices $u,v$ with $L(u)=L(v)$ are contracted. 
    We claim that this sequence is a $2^dk$-contraction sequence of $G$.
    Consider two vertices $u,v$ with $L(u)=L(v)$ which are to be contracted in the sequence.
    Note that a red edge can only appear incident to the contraction of $u$ and $v$, if they disagreed on a neighbor $w$ in $H$ where $L(u) \cap L(w) \neq \emptyset$.
    In the other cases, i.e. if they agreed on $w$ in $H$ or if $L(u) \cap L(w) = \emptyset$ held, they would agree on $w$ in $G$ and no red edge would be created.
    Therefore, any red edge appearing in $G$ must also appear in $H$, which implies that this sequence is indeed a $2^dk$-contraction sequence.
    Note that this sequence does not end on a single vertex, but on up to $2^d$ many vertices.
    These can be arbitrarily contracted, as a graph of size up to $2^d$ can only have red degree up to $2^d$, which gives us the claimed bound.

    To see that this bound is tight, we use the same construction given in the proof of \cref{thm:cwsimul} together with \cite[Lemma 3.2]{AhnHKO22}.
\end{proof}

Next we want to show $\simul$-boundedness of \emph{flip-width}, a family of parameters defined via a cops-and-robbers game.
Let $G = (V,E)$ be a graph.
A \emph{flip} on $G$ between two vertex sets $A,B$ ($A=B$ is possible) results in the graph obtained from $G$ by inverting the adjacency between any vertex pair $a \in A$, $b \in B$.
If $\cP$ is a partition of $V$, we call a graph $G'$ a \emph{$\cP$-flip} of $G$, if $G'$ can be obtained from $G$ by performing a sequence of flips between pairs $A,B \in \cP$.
We call $G'$ a \emph{$k$-flip} of $G$, if $G'$ is a $\cP$-flip of $G$ for some partition $\cP$ of $V$ with $\abs{\cP} \leq k$.
The \emph{flipper game} with radius $r \in \N \cup \set{\infty}$ and width $k \in \N$ is played by two players: the \emph{flipper} and the \emph{runner}.
In each round $i$ of the game, the flipper declares a $k$-flip $G_i$ of $G$ and the runner selects a new position $v_i$ as follows: initially, $G_0 = G$ and $v_0$ is chosen by the runner.
In round $i > 0$, the flipper announces a $k$-flip $G_i$ of $G$.
Before the $k$-flip is put in effect, the runner moves to a new vertex $v_i$ by following a path of length at most $r$ from $v_{i-1}$ in the old graph $G_{i-1}$.
Note that the runner knows the new graph $G_i$ at this point.
The game terminates when the runner is trapped, i.e. when $v_i$ is isolated in $G_i$.

\begin{definition}[Flip-width]
    For $r \in \N \cup \set{\infty}$ the \emph{radius-$r$ flip-width} of a graph $G$, denoted $\fwr{G}$, is the smallest $k \in \N$ such that the flipper has a winning strategy in the flipper game with radius $r$ and width $k$.
    A graph class $\cC$ has \emph{bounded flip-width}, if $\fwr{\cC} < \infty$ for every $r \in \N$.
    A graph class $\cC$ has \emph{almost bounded flip-width}, if for each fixed $r \in \N$, every $n$-vertex graph of $\cC$ has radius-r flip-width at most $n^{o(1)}$.
\end{definition}

We can interpret the property of bounded flip-width in terms of a graph class parameter, which is finite for a graph class if and only if it has bounded flip-width. Similarly, we can interpret almost bounded flip-width as a graph class parameter.

\begin{theorem}
    Let $r \in \N \cup \set{\infty}$.
    Then radius-$r$ flip-width is $\simul$-bounded.
    In particular, $\fwr{G} \leq \fwr{\cC} \cdot 2^{\simulbr{\cC}{G}}$ for any hereditary completable $\cC$ and any graph $G$.
\end{theorem}

\begin{proof}
    Let $k \in \N$ and $\cC$ be a graph class such that $\fwr{H} \leq k$ for all $H \in \cC$.
    Furthermore, let $G$ be a graph with $d = \simul[\cC][G]$ and $(H,L)$ be a $d$-simultaneous $\cC$-representation.
    
    We give a strategy for the flipper on $G$ of width $k 2^d$.
    At the start, the runner chooses a vertex $v_0$.
    Let $H_1$ be the $k$-flip that is chosen by the flipper in a winning strategy, if the game was played on $H$.
    Let $\cP$ be the partition of size $\leq k$ used to obtain $H_1$.
    We consider the partition $\cP'$, where each set $A \in \cP$ is replaced by the $2^d$ sets $A_L := \setcond{v \in A}{L(v)=L}$ for $L \subseteq \set{1,\dots,d}$.
    Consider the sequence of flips which are used to obtain $H_1$ from $H$.
    For every performed flip between sets $A,B \in \cP$ we perform flips between the sets $A_{L_1}, B_{L_2} \in \cP'$ with $L_1 \cap L_2 \neq \emptyset$.
    The graph obtained by this $\cP'$-flip will be called $G_1$ and announced by the flipper.

    We claim that that $G_1$ is a subgraph of $H_1$.
    Consider two vertices $u,v \in V(G)$.
    If $L(u) \cap L(v) = \emptyset$, the edge $uv$ does not exist in $G$.
    Let $u \in A \in \cP$ and $v \in B \in \cP$.
    As $L(u) \cap L(v) = \emptyset$, the flip between $A_{L(u)}$ and $B_{L(v)}$ is never performed, thus the edge $uv$ does not exist in $G_1$ (even if it does exist in $H_1$).
    Now consider the case $L(u) \cap L(v) \neq \emptyset$.
    Then $uv \in E(G)$ if and only if $uv \in E(H)$.
    Let $u \in A \in \cP$ and $v \in B \in \cP$.
    Then, for every flip between $A, B$ used to create $H_1$, the flip between $A_{L(u)}$ and $B_{L(v)}$ is performed to create $G_1$.
    Thus, the edge $uv$ exists in $G$ if and only if it exists in $H$.
    As $G_1$ is a subgraph of $H_1$, the set of vertices reachable from $v_0$ on paths of length $\leq r$ in $G_1$ is a subset of the vertices reachable from $v_0$ on paths of length $r$ in $H_1$.
    For all of these vertices, the flipper has a strategy on how to proceed, were the runner to move to that vertex in $H_1$.
    By iteratively modifying the strategy as demonstrated above, we obtain a winning strategy of the flipper of width $\leq k 2^d$ on $G$.
\end{proof}

From this theorem we obtain the fact that bounded flip-width and almost bounded flip-width are $\simul$-bounded.

\begin{corollary}
    Let $\cC$ be a hereditary completable graph class with bounded flip-width.
    Then any graph class with bounded simultaneous $\cC$-number also has bounded flip-width.
    In particular, the graph class parameter corresponding to bounded flip-width is $\simul$-bounded.
    If $\cC$ has almost bounded flip-width, then any graph class with bounded simultaneous $\cC$-number also has almost bounded flip-width.
    In particular, the graph class parameter corresponding to almost bounded flip-width is $\simul$-bounded.
\end{corollary}

A family of parameters closely related to flip-width is given by \emph{merge-width}.
Indeed, it is conjectured that bounded flip-width and bounded merge-width are equivalent~\cite{DreierT25}, which would imply $\simul$-boundedness of merge-width.
However, as this conjecture remains unproven, we still consider merge-width separately.
Assume we are given a vertex set $V$.
A \emph{construction sequence} is a sequence of steps, maintaining a partition $\cP$ of $V$ and three sets $E, N, U$ which partition $\binom{V}{2}$, the subsets of $V$ with cardinality two.
In the first step we have that $\cP$ consists of singleton sets and $U = \binom{V}{2}$.
Each step performs one of the following operations:
\begin{itemize}
    \item merge two sets $A,B \in \cP$, replacing them with $A \cup B$,
    \item resolve positively a pair $A,B \in \cP$ ($A=B$ is possible), declaring all unresolved pairs $ab \in U$ with $a \in A$, $b \in B$ as edges, moving them to $E$,
    \item resolve negatively a pair $A,B \in \cP$ ($A=B$ is possible), declaring all unresolved pairs $ab \in U$ with $a \in A$, $b \in B$ as non-edges, moving them to $N$.
\end{itemize}
In the final step, $\cP$ has one part and $U = \emptyset$.
We say that the sequence is a construction sequence of the graph $G = (V,E)$.
The \emph{radius-$r$ width} of the construction sequence is the least number $k \in \N$ such that for every step, the following holds:
for all $v \in V$, at most $k$ sets of the current partition $\cP$ can be reached from $v$ by a path of length $\leq r$ in $(V, E \cup N)$.

\begin{definition}[Merge-width]
    The \emph{radius $r$ merge-width} of $G$, denoted $\mwr{G}$, is the least radius-$r$ width of a construction sequence of $G$.
    A graph class $\cC$ has \emph{bounded merge-width}, if $\mwr{\cC} < \infty$ for all $r \in \N$.
    A graph class $\cC$ has \emph{almost bounded merge-width}, if for each fixed $r \in \N$, every $n$-vertex graph of $\cC$ has radius-$r$ merge-width at most $n^{o(1)}$.
\end{definition}

We again can interpret the properties of having bounded or almost bounded merge-width as graph class parameters.

\begin{theorem}
    Let $r \in \N$.
    Then radius-$r$ merge-width is $\simul$-bounded.
    In particular, $\mwr{G} \leq \mwr{\cC} \cdot 2^{\simulbr{\cC}{G}}$ for any hereditary completable $\cC$ and any graph $G$.
\end{theorem}

\begin{proof}
    Let $k \in \N$ and $\cC$ be a graph class such that $\mwr{H} \leq k$ for all $H \in \cC$.
    Furthermore, let $G$ be a graph with $d = \simul[\cC][G]$ and $(H,L)$ be a $d$-simultaneous $\cC$-representation.
    Let $(\cP_1, E_1, N_1, U_1), \dots, (\cP_m, E_m, N_m, U_m)$ be a construction sequence of $H$ with radius-$r$ width $k$.
    We describe how to adapt this construction sequence to a construction sequence of $G$ with radius-$r$ width $\leq 2^dk$.
    For every set $A \in \cP_t$ of some partition $\cP_t$ in the above sequence, we consider the sets $A_L := \setcond{v \in A}{L(v) = L}$ for all $L \subseteq \set{1,\dots,d}$.
    In our construction sequence, any set $A$ will be replaced by the sets $(A_L)_{L \subseteq \set{1,\dots,d}}$.
    Then, for every step in the original construction sequence, we adapt them in the following ways:
    \begin{itemize}
        \item merging two sets $A,B \in \cP$: merging $A_L$ and $B_L$ for every $L \subseteq \set{1,\dots, d}$,
        \item resolving $A,B \in \cP$ negatively: resolving $A_{L_1}, B_{L_2}$ negatively for every $L_1, L_2 \subseteq \set{1,\dots,d}$,
        \item resolving $A,B \in \cP$ positively: resolving $A_{L_1}, B_{L_2}$ with $L_1 \cap L_2 \neq \emptyset$ positively and $A_{L_1}, B_{L_2}$ with $L_1 \cap L_2 = \emptyset$ negatively.
    \end{itemize}
    After each of these bundles of steps, the same edges are resolved in $G$ that were resolved after the single step in $H$.
    Thus, we have the same paths and reachability.
    As every original set of the partition is now split into up to $2^d$ subsets, any vertex can reach at most $k2^d$ parts of the partition using a resolved path of length $\leq r$.
    Observing a step of the sequence inside a bundle of steps, we have that strictly less edges are resolved, thus there can not be more paths.
    At the end of the process, we have that the partition is not a single set, but contains the sets $V_L = \setcond{v \in V}{L(v) = L}$.
    Then, in the last $2^d$ steps, these sets are merged in no particular order.
    During that, at most $2^d$ parts of the partition can be reached during any step.
\end{proof}

Again, we obtain that bounded merge-width and almost bounded merge-width are $\simul$-bounded.

\begin{corollary}
    Let $\cC$ be a hereditary completable graph class with bounded merge-width.
    Then any graph class with bounded simultaneous $\cC$-number also has bounded merge-width.
    In particular, the graph class parameter corresponding to bounded merge-width is $\simul$-bounded.
    If $\cC$ has almost bounded merge-width, then any graph class with bounded simultaneous $\cC$-number also has almost bounded merge-width.
    In particular, the graph class parameter corresponding to almost bounded merge-width is $\simul$-bounded.
\end{corollary}

The last parameter we consider in this subsection is \emph{shrub-depth}.
Let $m,k \in \N_0$.
A \emph{tree-model of $m$ colors and depth $k$} for a graph $G$ is a pair $((T,r),S)$ of a rooted tree $(T,r)$ of height $k$ and a set $S \subseteq \set{1,\dots,m}^2 \times \set{1,\dots,k}$, called a \emph{signature} of the tree-model, such that:
\begin{itemize}
    \item the length of each path from $r$ to some leaf is exactly $k$,
    \item the set of leaves of $T$ is $V(G)$,
    \item each leaf of $T$ is assigned one of the colors $\set{1,\dots,m}$,
    \item for any $i, j, \ell$ it holds $(i,j,\ell) \in S$ if and only if $(j,i,\ell) \in S$ and for any two vertices $u,v \in V(G)$ such that $u$ is colored $i$ and $v$ is colored $j$ and the distance between $u,v$ in $T$ is $2 \ell$, the edge $uv$ exists in $G$ if and only if $(i,j,\ell) \in S$.
\end{itemize}
The class of all graphs having a tree-model of $m$ colors and depth $k$ is denoted by $\TM_m(k)$.

\begin{definition}[Shrub-depth]
    A graph class $\cC$ has \emph{shrub-depth $k$} if there exists an $m$ such that $\cC \subseteq \TM_m(k)$, while for all $m'$ it holds that $\cC \not\subseteq \TM_{m'}(k-1)$.
\end{definition}

\begin{theorem}
    Shrub-depth is $\simul$-bounded.
    In particular, if $\cC \subseteq \TM_m(k)$ for some $m, k \in \N$, then any graph class with simultaneous $\cC$-number at most $d$ is contained in $\TM_{2^d m} (k)$.
\end{theorem}

\begin{proof}
    Let $\cC \subseteq \TM_m(k)$ be a hereditary completable graph class.
    Let $G$ be a graph with $d = \simul[\cC][G]$ and $(H,L)$ be a $d$-simultaneous $\cC$-representation of $G$.
    Let $((T,r),S)$ be a tree-model of $H$ of $m$ colors and depth $k$.
    For every color $i$ and every label set $L \subseteq \set{1,\dots,d}$ we introduce a new color $i_L$.
    This gives us $2^d m$ many colors.
    If a vertex $v$ has color $i$ in the tree-model, we give it the new color $i_{L(v)}$.
    Furthermore, we define a new signature $S'$ as follows:
    \[ (i_{L_1}, j_{L_2}, \ell) \in S' \Leftrightarrow (i,j,\ell) \in S \text{ and } L_1 \cap L_2 \neq \emptyset. \]
    Thus, in the graph defined by the tree-model $((T,r), S')$, two vertices $u$ and $v$ with colors $i$ and $j$ and distance $2 \ell$ in $T$ are adjacent if and only if $(i,j,\ell) \in S$, i.e., if they are adjacent in $H$, and if $L(u) \cap L(v) \neq \emptyset$.
    The graph defined this way is exactly the graph $G$, therefore $G \in \TM_{2^dm}(k)$.
\end{proof}

As $\TM_m(1)$ is exactly the class of graphs with neighborhood diversity $m$, we immediately get the following:

\begin{corollary}
    Neighborhood diversity is $\simul$-bounded.
    In particular, for any hereditary completable $\cC$ and any graph $G$, $\nd{G} \leq \nd{\cC} \cdot 2^{\simulbr{\cC}{G}}$.
\end{corollary}

\section{Justification of Figure 1} \label{sec:justification}

We now justify \cref{fig:diagram} following the argumentation scheme given in \cite{simchord}.  We have to prove that a parameter $p$ bounds parameter $q$ if and only if there is a directed path from $p$ to $q$ in \Cref{fig:diagram}.
For each parameter $p$, we justify only pairs $(p,q)$ for which one of the following conditions hold:
\begin{itemize}
    \item $p$ bounds $q$ where $q$ is an immediate successor of $p$ in the diagram, or 
    \item $p$ does not bound $q$, all the parameters $p'\neq p$ that bound $p$ also bound $q$, and $q$ is a vertex of out-degree $0$ in the subgraph induced by the vertices corresponding to the parameters that are not bounded by $p$. 
\end{itemize}
The remaining relations follow by transitivity.

Additionally, we justify why a parameter is closed or is not closed under first-order transductions.

We will often make use of the following three results.

\begin{lemma} \label{lem:completebipartite}
    The following parameters are unbounded on the class of complete bipartite graphs: Cubicity, interval number, and tree independence number.
\end{lemma}

This lemma follows from results in \cite{DallardMS24treeindependence,GriggsW80intervalnumber,roberts1969boxicity}.

\begin{lemma} \label{lem:compmatching}
    Thinness and chordality are unbounded on complements of matchings.
\end{lemma}

This follows from \cite{mannino2002solving} and \cite[Theorem 4.16]{exphierarchy}.

\begin{lemma} \label{lem:grids}
    Grids have unbounded sim-width.
\end{lemma}

\begin{proof}
    This follows from the fact that grids have unbounded mim-width~\cite{vatshelle2012mimwdith} and that $K_3$-free graphs with unbounded mim-width have unbounded sim-widh~\cite{kang2017sim}.
\end{proof}

\subsection{Edge Clique Cover Number}

Edge clique cover number upper bounds
\begin{itemize}
    \item cubicity by \cite[Theorem 4]{MichaelQ06sphericity}
    \item track number: every clique is an interval graphs
    \item tree independence number: By \cite[Corollary 3.13]{beisegel2024simultaneousinterval} edge clique cover number upper bounds the simultaneous interval number, which in turn upper bounds the tree independence number by \cite[Corollary 3.7]{beisegel2024simultaneousinterval}
    \item neighborhood diversity by \cite[Theorem 4.1]{exphierarchy}
\end{itemize}
Edge clique cover number does not upper bound sm-width by \cref{lem:smecc}. \\
The edge clique cover number is not closed under first-order transductions as it is not closed under graph complementation. This can be seen using the family of graphs $(2K_n)_{n \in \N}$.

\subsection{Bandwidth}

Bandwidth upper bounds
\begin{itemize}
    \item cubicity by \cite[Theorem 5]{ChandranFS13}
    \item treewidth: this can easily be seen by considering the equivalent definition of bandwidth as \emph{proper pathwidth}~\cite{kaplan1996bandwidth}
    \item thinness: bandwidth upper bounds pathwidth, which in turn upper bounds thinness~\cite{mannino2007thinness}
\end{itemize}
Bandwidth does not upper bound
\begin{itemize}
    \item precedence thinness: unions of graphs with bounded bandwidth have bounded bandwidth, while unions of graphs with bounded precedence thinness $\geq 2$ have unbounded precedence thinness by \cref{lem:precthinunion}
    \item tree-length: Cycles have bandwidth $2$ and unbounded tree-length~\cite[Lemma 4]{DourisboureG07}.
\end{itemize}
Bandwidth is not closed under first-order transductions as it is not closed under graph complementation. This can be seen using edgeless graphs.

\subsection{Treedepth}

Treedepth upper bounds
\begin{itemize}
    \item treewidth by \cite[Lemma 11]{bodlaender1995approximating}
    \item thinness: treedepth upper bounds pathwidth~\cite[Lemma 11]{bodlaender1995approximating} which in turn upper bounds thinness~\cite{mannino2007thinness}
    \item shrub-depth by \cite[Proposition 3.4]{GanianHNOM19}
\end{itemize}
Treedepth does not upper bound
\begin{itemize}
    \item cubicity: Stars have treedepth $1$ and unbounded cubicity~\cite[Theorem 1]{roberts1969boxicity}.
    \item precedence thinness: unions of graphs of bounded treedepth have bounded treedepth, while unions of graphs with bounded precedence thinnes $\geq 2$ have unbounded precedence thinness by \cref{lem:precthinunion}
    \item modular-width: Trees of bounded height have bounded treedepth and unbounded modular-width (e.g. rooted height $2$ trees where every non-leaf has $\geq d$ children for arbitrary $d$).
\end{itemize}
Treedepth is not closed under first-order transductions as it is not closed under graph complementation. This can be seen using edgeless graphs.

\subsection{Neighborhood Diversity}

Neighborhood diversity upper bounds
\begin{itemize}
    \item precedence thinness: Let $G$ be a graph with neighborhood diversity $k$. Let $V_1, \dots, V_k$ be the type partition of $G$, that is, each $V_i$ is a module in $G$ and induces either a complete or edgeless subgraph. We claim that any vertex ordering $v_1, \dots, v_n$ is consistent with the type partition. Let $v_r, v_s, v_t \in V(G)$ with $r < s < t$, $v_r, v_s \in V_i$, and $v_r v_t \in E(G)$ be given. If $v_t$ is also contained in $V_i$, then $V_i$ induces a complete subgraph and we have $v_s v_t \in E(G)$. Otherwise, $v_t$ is in another set $V_j$. Then we have that $v_s v_t \in E(G)$ as $V_i$ is a module-
    \item shrub-depth by definition (see also \cite{GanianHNOM19})
    \item induced matching treewidth: Let $G$ be a graph, $k = \nd{G}$ and $V_1, \dots, V_k$ the type partition of $G$. Let $M \subseteq E(G)$ be some maximum induced matching in $G$. Consider the sets of endpoints $A,B \subseteq V(G)$ of $M$. Obviously, no two vertices in $A$ can be of the same type, similiarly for $B$. Thus, as the induced matching treewidth of $G$ is bounded by $\abs{M} \leq \nd{G}$, neighborhood diversity upper bounds induced matching treewidth.
    \item iterated type partition by definition
\end{itemize}
Neighborhood diversity does not upper bound
\begin{itemize}
    \item cubicity: Complete bipartite graphs have neighborhood diversity $2$ and unbounded cubicity (\cref{lem:completebipartite})
    \item interval number: Complete bipartite graphs have neighborhood diversity $2$ and unbounded interval number (\cref{lem:completebipartite})
    \item tree independence number: Complete bipartite graphs have neighborhood diversity $2$ and unbounded tree independence number (\cref{lem:completebipartite})
\end{itemize}
Neighborhood diversity is closed under first-order transductions, as shrub-depth is closed under first-order transductions~\cite{GanianHNOM19}. This is because the class of graphs of neighborhood diversity $k$ is exactly the class $\TM_k(1)$ and any first-order transduction of $\TM_k(1)$ is contained in $\TM_{k'}(1)$ for some $k'$, which has bounded neighborhood diversity.

\subsection{Precedence Thinness}

Precedence thinness upper bounds thinness by definition. \\
Precedence thinness does not upper bound
\begin{itemize}
    \item induced matching treewidth: The family of graphs with unbounded induced matching treewidth given in \cite[Proposition 7.10]{lima2024inducedmatchingarxiv} can easily be seen to have precedence thinness $2$
    \item tree-length: Cycles have unbounded tree-length~\cite[Lemma 4]{DourisboureG07} but bounded precedence thinness. A cyclic ordering $v_1, \dots, v_n$ of $C_n$ of a cycle $C_n$ together with the partition $V^1 := \set{v_1}, V^2 := \set{v_2, \dots, v_{n-1}}, V^3 := \set{v_n}$ verify this
    \item flip-width: Interval graphs have precedence thinness $1$ but unbounded flip-width~\cite{EppsteinM23geometric}
\end{itemize}
Precedence thinness is not closed under first-order transductions as it is not closed under graph complementation. This can be seen using matchings, which have bounded precedence thinness, while their complements have already unbounded thinness~\cite{mannino2002solving}.

\subsection{Thinness}

Thinnes upper bounds
\begin{itemize}
    \item boxicity by \cite[Theorem 4.5]{mannino2002solving}
    \item mim-width by \cite[Theorem 9]{bonomo2019properthinness}
\end{itemize}
Thinness does not upper bound precedence thinness: Unions of graphs of bounded thinness have bounded thinness, while unions of graphs of bounded precedence thinness $\geq 2$ have unbounded precedence thinness by \cref{lem:precthinunion}. \\
Thinness is not closed under first-order transductions using the same arguments as precedence thinness.

\subsection{Iterated Type Partition}

Iterated type partition upper bounds modular-width, which follows easily from the definitions, see also \cite{cordascoetal2020itp}. \\
Iterated type partition does not upper bound
\begin{itemize}
    \item chordality: Complements of matchings have iterated type partition $1$ and unbounded chordality
    \item thinness: Complements of matchings have iterated type partition $1$ and unbounded thinness
    \item shrub-depth: The graph class of \cite[Example 5.4a)]{GanianHNOM19} with unbounded shrub-depth has iterated type partition $1$
    \item induced matching treewidth: The family of graphs with unbounded induced matching treewidth given in \cite[Proposition 7.10]{lima2024inducedmatchingarxiv} has iterated type partition $1$
\end{itemize}
Iterated type partition is not closed under first-order transductions as it is not $\simul$-bounded.

\subsection{Treewidth}

Treewidth upper bounds
\begin{itemize}
    \item boxicity by \cite[Theorem 14]{ChandranS07boxicity}
    \item track number: by \cite[Theorem 4.1]{DingOSV98starforests} every graph of treewidth $k$ can be obtained as the union of $k+1$ star forests. Star forests are interval graphs and thus treewidth upper bounds track number.
    \item tree independence number by definition
    \item sm-width by \cite[Proposition 33]{saether2016between}
\end{itemize}
Treewidth does not upper bound thinness: Trees have bounded treewidth and unbounded thinness~\cite{bonomobraberman2025thinnesstrees}. \\
Treewidth is not closed under first order transductions as it is not closed under graph complementation. This can be seen using edgeless graphs.

\subsection{Cubicity}

Cubicity upper bounds boxicity by definition. \\
Cubicity does not upper bound
\begin{itemize}
    \item sim-width: Grids have cubicity $3$ and unbounded sim-width. To see that grids have cubicity $3$, consider an $n \times m$ grid with vertices $v_{i,j}$, $1 \leq i \leq n$, $1 \leq j \leq m$ such that $v_{i,j}$ is adjacent to $v_{i+1,j}$ and $v_{i,j+1}$ for $i < n$, $j < m$. Now consider the three unit interval graphs with the following interval representations: $R_1(v_{i,j}) = [i,i+1]$, $R_2(v_{i,j}) = [j, j+1]$, and $R_3(v_{i,j}) = [i+j, i+j+1]$. Then the $n \times m$ grid is exactly the intersection of these unit interval graphs.
    \item flip-width: Intersection graphs of axis-parallel unit squares have cubicity $2$ and unbounded flip-width~\cite{EppsteinM23geometric}.
\end{itemize}

It seems to be open whether cubicity upper bounds interval number and track number. However, cubicity does upper bound the track number on $K_3$-free graphs, as it upper bounds the \emph{claw number}~\cite{AdigaC10clawcubicity}, which is an upper bound on the maximum degree in $K_3$-free graphs, which in turn is an upper bound on the track number (this can be seen by adapting the proof of \cite[Theorem 2]{GriggsW80intervalnumber}). It is not immediately clear, whether a similar argument can be used to show that cubicity upper bounds track number on all graphs.\\
Cubicity is not closed under first order transductions as it is not closed under graph complementation. This can be seen using matchings, which have bounded cubicity, while their complements have already unbounded chordality.

\subsection{Boxicity}

Boxicity upper bounds chordality by definition. \\
Boxicity does not upper bound
\begin{itemize}
    \item cubicity: Stars have boxicity $1$ and unbounded cubicity~\cite[Theorem 1]{roberts1969boxicity}
    \item interval number: Complete bipartite graphs have boxicity $2$ and unbounded interval number.
\end{itemize}
Boxicity is not closed under first-order transductions using the same arguments as cubicity.

\subsection{Chordality}

Chordality does not upper bound boxicity: It was shown in \cite{ChandranFM11chordalbipartite} that chordal bipartite graphs with unbounded boxicity exist, while every bipartite graph is the intersection of $2$ split graphs, which are chordal. \\
Chordality is not closed under first-order transductions using the same arguments as cubicity.

\subsection{Modular-Width}

Modular-width upper bounds
\begin{itemize}
    \item tree-length: it has been shown in \cite[Theorem 4.7]{exphierarchy} that modular-width upper bounds the maximum diameter of components of a graph, which is a trivial upper bound on tree-length
    \item cliquewidth by \cite[Corollary 3.6]{courcelle2007upper}
\end{itemize}
Modular-width does not upper bound iterated type partition: The family of lexicographic powers of $P_4$ have modular-width $4$ but unbounded iterated type partition. That is a sequence of graphs defined as follows: the first graph in the sequence is the path $P_4$. Any later graph in the sequence is obtained by substituting every vertex of a $P_4$ with the previous graph in the sequence. \\
Modular-width is not closed under first-order transductions as it is not $\simul$-bounded.

\subsection{Shrub-Depth}

Shrub-depth upper bounds
\begin{itemize}
    \item cliquewidth by \cite[Proposition 3.4]{GanianHNOM19}
    \item tree-length: For any class of bounded shrub-depth there is a $k \in \N$ such that no graph in that class has the induced subgraph $P_k$ (as paths have unbounded shrub-depth). Thus, it must hold that the diameter (and thus the tree-length) of any graph in that class is bounded by $k$
\end{itemize}
Shrub-depth does not upper bound
\begin{itemize}
    \item chordality: Complements of matchings have bounded shrub-depth ($\leq 2$) and unbounded chordality
    \item induced matching treewidth: The family of graphs with unbounded induced matching treewidth given in \cite[Proposition 7.10]{lima2024inducedmatchingarxiv} has bounded shrub-depth ($\leq 3$)
\end{itemize}
Shrub-depth is closed under first-order transductions~\cite{GanianHNOM19}.

\subsection{Tree-Length}

Tree-lengt does not upper bound
\begin{itemize}
    \item sim-width: To construct a family of graphs with bounded tree-length and unbounded sim-width, consider any family of graphs with unbounded sim-width (e.g. grids) and add to every graph a universal vertex. This does not decrease the sim-width, but as the diameter of graphs containing a universal vertex is bounded, so is the tree-length.
    \item flip-width: Interval graphs have bounded tree-length and unbounded flip-width~\cite{EppsteinM23geometric}
\end{itemize}
Tree-length is not closed under first-order transductions as it is not $\simul$-bounded.

\subsection{Sm-Width}

Sm-width upper bounds cliquewidth by \cite[Proposition 33]{saether2016between}. \\
Sm-width does not upper bound
\begin{itemize}
    \item chordality: complements of matchings are cographs and therefore have sm-width $1$~\cite[Proposition 35]{saether2016between}, while they have unbounded chordality
    \item interval number: complete bipartite graphs have sm-width $1$ but unbounded interval number
    \item induced matching treewidth: the graphs with unbounded induced matching treewidth given in \cite[Proposition 7.10]{lima2024inducedmatchingarxiv} are cographs and therefore have sm-width $1$~\cite[Proposition 35]{saether2016between}
\end{itemize}
Sm-width is not closed under first-order transductions as it is not $\simul$-bounded.

\subsection{Cliquewidth, Rank-Width, and Boolean-Width}

Cliquewidth, rank-width and boolean-width pairwise upper bound each other by \cite[Theorem 4.1.3, Theorem 4.2.9]{vatshelle2012mimwdith}. \\
Rank-width upper bounds mim-width by \cite[Theorem 4.2.10]{vatshelle2012mimwdith}. \\
Boolean-width upper bounds twin-width by \cite[Theorem 4.2]{BonnetKTW22}. \\
Cliquewidth is closed under first-order transductions~\cite{Colcombet07, Courcelle_Engelfriet_2012}.

\subsection{Track Number and Interval Number}

Track number upper bounds interval number by definition. \\
Track number does not upper bound
\begin{itemize}
    \item chordality: Complements of matchings have track number $2$ and unbounded chordality
    \item sim-width: Grids have track number $2$ and unbounded sim-width
    \item flip-width: Interval graphs have track number $1$ and unbounded flip-width~\cite{EppsteinM23geometric}
\end{itemize}
Interval number does not upper bound track number: Line graphs of complete graphs have interval number $2$ and unbounded track number~\cite[Section 3]{milans2015ordered}. \\
Track number and interval number are not closed under first-order transductions as they are not closed under graph complementation. This can be seen using the family $(2K_n)_{n \in \N}$ which has bounded track number, while its complement family of complete bipartite graphs has unbounded interval number.

\subsection{Mim-Width}

Mim-width upper bounds o-mim-width by definition. \\
Mim-width does not upper bound flip-width: Interval graphs have mim-width $1$ and unbounded flip-width~\cite{EppsteinM23geometric}. \\
Mim-width is not closed under first-order transductions as it is not closed under graph complementation. This can be seen using complements of grids, which have bounded mim-width due to them being complements of $2$-degenerate graphs~\cite[Lemma 2, Corollary 14]{belmonte2013graphclasses}, while grids already have unbounded sim-width.

\subsection{Twin-Width}

Classes of bounded twin-width have bounded merge-width by \cite[Theorem 1.4]{DreierT25}. \\
Twin-width does not upper bound
\begin{itemize}
    \item boolean-width: Proper interval graphs have bounded twin-width~\cite[Lemma 3.6]{BonnetGKTW24twinwidth3} but unbounded cliquewidth~\cite[Theorem 1.3]{GolumbicR00}
    \item sim-width: Grids have bounded twin-width and unbounded sim-width
\end{itemize}
Twin-width is closed under first-order transductions~\cite{BonnetKTW22}.

\subsection{Tree Independence Number}

Tree independence number upper bounds
\begin{itemize}
    \item induced matching treewidth as the existence of an induced matching with $k$ edges implies the existence of an independent set of size $k$
    \item o-mim-width by \cite[Theorem 2]{bergougnoux2023omim}
\end{itemize}
Tree independence number does not upper bound
\begin{itemize}
    \item chordality: Complements of matchings have tree independence number $2$ and unbounded chordality
    \item interval number: Split graphs have tree independence number $1$ and unbounded interval number~\cite[Theorem 1.5]{BaloghOP04ontheinterval}
    \item mim-width: Chordal graphs have tree independence number $1$ and unbounded mim-width~\cite{mengel2018lower}
    \item flip-width: Interval graphs have tree independence number $1$ and unbounded flip-width~\cite{EppsteinM23geometric}
\end{itemize}
Tree independence number is not closed under first-order transductions as it is not closed under graph complementation. This can be seen using the family $(2K_n)_{n \in \N}$ which has bounded tree independence number, while its complement family of complete bipartite graphs has unbounded tree independence number.

\subsection{O-Mim-Width}

O-mim-width upper bounds sim-width by definition. \\
O-mim-width is not closed under first-order transductions using the same arguments as mim-width.

\subsection{Merge-Width}

Bounded merge-width implies bounded flip-width by \cite[Theorem 1.7]{DreierT25}. It is open whether there are graph classes with bounded flip-width and unbounded merge-width. \\
Bounded merge-width does not imply bounded twin-width: $3$-regular graphs have bounded merge-width~\cite{DreierT25} and unbounded twin-width~\cite{BonnetGKTW22twinwidth2}. \\
Merge-width is closed under first-order transductions~\cite{DreierT25}.

\subsection{Induced Matching Treewidth}

Induced matching treewidth upper bounds sim-width by \cite[Theorem 15]{bergougnoux2023omim}. \\
Induced matching treewidth does not upper bound o-mim-width, which is shown in \cite{bergougnoux2023omimarxiv}. \\
Induced matching treewidth is not closed under first-order transductions as it is not closed under graph complementation. This can be seen using the family of graphs with unbounded induced matching treewidth given in \cite[Proposition 7.10]{lima2024inducedmatchingarxiv}, whose complement family can easily be seen to have bounded induced matching treewidth.

\subsection{Sim-Width and Flip-Width}

It is open whether there are classes with bounded flip-width and unbounded merge-width~\cite{DreierT25}. \\
Flip-width is closed under first-order transductions~\cite{Torunczyk23}. \\
Sim-width is not closed under first-order transductions using the same arguments as mim-width.

\end{document}